\documentclass[12pt]{article}
\usepackage{amsmath}
\usepackage{amssymb}
\usepackage{amsfonts}
\usepackage[margin=1in]{geometry}
\usepackage{graphicx}
\usepackage{float}
\usepackage[toc]{appendix}
\usepackage{caption}
\usepackage{natbib}
\usepackage{comment}
\usepackage{hyperref}
\usepackage{booktabs}
\usepackage{tikz}
\usepackage{setspace}

\usetikzlibrary{positioning} 
\usetikzlibrary{calc}
\newtheorem{theorem}{Theorem}

\newtheorem{algorithm}{Algorithm}

\newtheorem{condition}{Condition}

\newtheorem{lemma}[theorem]{Lemma}

\newtheorem{proposition}[theorem]{Proposition}

\newenvironment{proof}[1][Proof]{\noindent\textbf{#1.} }{\ \rule{0.5em}{0.5em}}
\input{tcilatex}
\begin{document}

\title{Process Utility in High-Stakes Competition.\thanks{%
We are very grateful to Jeff Sackmann for providing us access to the data
from the Match Charted Project. We thank seminar participants at Maastricht
University and, in particular, Lex Borghans, Arian Schwidder, and Christian
Seel for interesting discussions and very useful comments on an earlier
draft of the paper.}}
\author{Arnaud Dupuy\thanks{%
University of Luxembourg, email: arnaud.dupuy@uni.lu}}
\date{ Last updated \today}
\maketitle

\begin{abstract}
We study how individuals trade off outcome (\textquotedblleft
what\textquotedblright ) and process (\textquotedblleft
how\textquotedblright ) utility in high-stakes strategic decisions, namely
professional tennis. Using optimality conditions and the second-service
rule, we derive a sufficient condition for the nonparametric lower bound on
the weight of process utility to be positive. Under mild shape restrictions,
the high-frequency data indicate that most players likely value process
utility positively. We then develop a structural model that recovers
player-specific preferences over outcomes and processes. Estimates show that
players systematically sacrifice success probabilities to increase process
utility, with economically meaningful consequences for match outcomes and
expected earnings.

\end{abstract}



\newpage

\begin{quote}
\textquotedblleft When we act because we enjoy the activity itself, not
because of what it leads to, the experience is autotelic --- literally,
having its goal within itself.\textquotedblright\ --- Mihaly
Csikszentmihalyi, \cite{csik90}.
\end{quote}

\bigskip

\begin{quote}
\textquotedblleft Success is a journey, not a destination. The doing is
often more important than the outcome.\textquotedblright\ --- Arthur Ashe. 
\newline
https://www.edwardssports.co.uk/news/post/practice-tennis-like-a-pro
\end{quote}

\section{Introduction}

A central question in economics is the extent to which individuals are
willing to trade off extrinsic rewards against the intrinsic enjoyment of an
activity (e.g., \cite{frey02}; \cite{benabou03}).\footnote{%
See also \cite{deci99}, \cite{Bandiera05}, \cite{falk06}, \cite{bryan11}, 
\cite{larkin12}, \cite{falk13}, \cite{kube13}, \cite{corgnet16} among others.%
} In many instances, this trade-off can be framed as one between an outcome
(what) and the process (how) by which it is obtained. For instance, Adam
Smith's discussion of compensating wage differentials and its formalization
in modern labor economics (e.g., \cite{rosen86}) acknowledges that job
attributes such as effort, risk, or enjoyment enter workers' utility
alongside wages. Empirically, however, quantifying this trade-off remains
challenging. Observational data typically raises two important issues:
preferences are confounded with constraints, selection, and unobserved
heterogeneity, and identification relies on between-person comparisons, as
indeed the frequency of such decisions is generally limited.\footnote{%
For instance, in our labor market example, workers who accept lower wages
for more enjoyable jobs may differ systematically in ability or outside
options (e.g., \cite{mas25}; \cite{lavetti23}) and only change jobs a
limited number of times in the course of their career.} Experimental
evidence circumvents these issues and has provided important insights (e.g., 
\cite{gneezy00a,gneezy00b}; \cite{ariely08}), but often relies on low-stakes
laboratory environments with limited external validity. As a result, there
is relatively little direct field evidence at the individual level on
whether people are willing to sacrifice measurable economic outcomes in
exchange for intrinsically rewarding processes in high-stakes, competitive
settings.

Professional tennis provides a uniquely well-suited setting to study this
trade-off. Players operate in a high-stakes environment with strong
extrinsic motivation (monetary and reputational incentives) to maximize
performance. Yet, they face repeated, strategically rich decisions---most
notably in serving with the second serve rule---where they can choose
between higher-risk strategies that increase the probability of winning the
point immediately in one shot (unreturned serve), but also of making a
mistake, and more conservative strategies that lead to longer rallies.
Crucially, because points are won either in a single shot or through
multi-shot exchanges, tennis offers a setting to separately identify outcome
and process motivations. Finally, being one of the most broadcast sports on
the planet, a wealth of point-level data is recorded for many matches,
allowing a rich analysis of tennis rallies for a large number of players.
This combination of high incentives, clear strategic margins, and granular
data makes tennis an ideal setting to test whether individuals deviate from
outcome-maximizing behavior to engage in more intrinsically rewarding forms
of play.

We propose a simple framework in which players derive utility from winning
points -outcome utility-- but also from the way in which they win points
-process utility--. In particular, we allow process utility to depend on
whether a point is won in one shot, immediately after the serve, or through
multi-shot rallies, and model serving decisions as the outcome of utility
maximization under this extended preference structure. We call the relative
importance of process utility in total utility the \emph{salience weight}.
Our framework yields a natural interpretation of observed deviations from
outcome-maximizing strategies\footnote{%
See \cite{borghans95, klaassen09} for instance.} as reflecting heterogeneous
salience weights, which we recover for each player from point-level data.

Using optimal serving decisions in professional tennis, we develop a
nonparametric identification strategy that exploits optimality conditions
and the second-serve rule to derive a sufficient condition under which
player-specific bounds on the salience weight are positive. Confronting this
condition to detailed point-level data from the Sackmann Charting Project,
we show that, under mild shape restrictions, a large majority of players
likely have a positive salience weight. Further exploiting these optimality
conditions, we adopt a parametric approach and propose an algorithm to
structurally recover player-specific salience weights on process utility. We
then estimate player-specific skills and salience weights using this
algorithm in a maximum likelihood procedure. We find that $79\%$ of
professional players have a positive salience weight---of which $64\%$ are
statistically significant at the $5\%$ level---and therefore place greater
weight on winning multi-shot rallies, resulting in systematically more
conservative (second-)serve strategies than predicted by outcome-maximizing
models.\footnote{%
For 119 out of 151 players in the data, the estimated salience weight for
winning multi-shot rallies is larger than 0. For 76 of them, this weight is
significantly different from 0 at the $5\%$ level. For only 3 players, this
weight is significantly lower than 0 at the $5\%$ level.} Counterfactual
exercises show that eliminating process utility would increase point-winning
probability on serve by about 0.4 percentage points, translating into a 2.4
percentage point increase in match-winning probability and an increase of
\$33,000 (13.5\%) in expected prize money at a Grand Slam tournament (US
Open 2025). These results provide field evidence that individuals are
willing to forgo measurable extrinsic rewards to engage in intrinsically
rewarding activities, and illustrate how small deviations from
outcome-maximizing behavior can have economically meaningful consequences in
high-stakes competitive settings.

Methodologically, our model fits naturally within the canonical framework of
trade-offs between extrinsic and intrinsic motivation. To fix ideas,
consider an agent choosing among alternatives leading to an outcome. Each
alternative $x$ is associated with outcome utility $p(x)$ and process
utility $k(x)$. In its simplest form, utility is a weighted sum of the two,
where $\delta \geq 0$ denotes the relative weight on process utility. When $%
\delta = 0$, the agent is purely outcome-maximizing, while larger values of $%
\delta$ reflect increasing salience of process utility.

This formulation is consistent with leading models of motivation, where
process utility arises from civic duty (\cite{frey97}), identity (\cite%
{loewenstein99}, \cite{akerlof00}), implicit contracts (\cite%
{gneezy00b,gneezy00a}), beliefs (\cite{benabou03}), procedures (\cite{frey04}%
), meaning (\cite{ariely08}, \cite{norton12}), gambling (\cite{LeMenestrel01}%
), or the act of choice itself (\cite{sen97}). In our setting, a player
chooses a serve strategy $x$ and trades off the probability of winning a
point on serve $p(x)$ against the probability of winning the point through
multiple shots $k(x)$. By structurally recovering player-specific salience
weights from optimal choices, we quantify the extent to which individuals
deviate from outcome maximization to engage in intrinsically rewarding play.

Our analysis builds on the premise that winning multi-shot rallies is
intrinsically more rewarding than winning one-shot rallies. We motivate this
in three steps. First, from the psychological literature, we learn that
enjoyment and self-determination are important drivers of intrinsic
motivation. Research first formalized by \cite{csik75, csik90} indicates
that \emph{flow} is a state of complete absorption in an activity, often
described as being \textquotedblleft in the zone,\textquotedblright\ that
arises when the challenge of the task aligns with the performer's skill
level, and therefore \textquotedblleft \lbrack e]njoyment appears at the
boundary between boredom and anxiety, when the challenges are just balanced
with the person's capacity to act.\textquotedblright\ \cite{csik90}, pp.
52--53.\footnote{%
Csikszentmihalyi illustrates this with tennis: \textquotedblleft One cannot
enjoy doing the same thing [at tennis] at the same level for long. We grow
either bored or frustrated, and then the desire to enjoy ourselves again
pushes us to stretch our skills, or to discover new opportunities for using
them.\textquotedblright\ \cite{csik90}, p.75.} Flow theory connects closely
with \emph{Self-Determination Theory} (SDT) \cite{deci00}, which posits that
intrinsic motivation is fostered when autonomy, competence, and relatedness
are satisfied.\footnote{%
In tennis, autonomy arises from controlling shot selection and tactics;
competence from improving skills and executing difficult shots; and
relatedness from interactions with coaches, opponents, and the public.} Flow
states and Self-Determination are therefore more likely to arise during
multi-shot rallies than through the execution of a single shot, the serve.

Second, from players' testimonies, we learn that enjoyment is indeed an
important component of their motivation. Professional players frequently
highlight enjoyment as a central goal: for instance, Carlos Alcaraz stated,
\textquotedblleft I just want to step on court \ldots {} and try to enjoy as
much as I can.\textquotedblright \footnote{%
\url{https://www.atptour.com/en/news/alcaraz-lehecka-us-open-2025-qf}. A
broader set of players' quotes is collected in Online Appendix (\ref%
{App:Testimonies}).}

Third, in tennis, although the server wins roughly 45\% of his points on
unreturned serves (one-shot), top players spend most of their training time
(about $90\%$) on baseline rallies as shown in \cite{oshannessy19} and \cite%
{fitzpatrick24a}. This apparent paradox in the behavior of players during
practice supports the idea that players attach more importance to multi-shot
rallies, likely so because they enjoy playing such rallies more. For these
reasons, we expect process utility to arise from winning multi-shot rallies,
and player-specific salience weights for these rallies capture the
preference for winning points through multi-shot rallies relative to
one-shot points.


This paper contributes in several important ways. First, it extends the
analysis of tennis serving strategies of \cite{klaassen09} to account for
process utility and distinguish between one-shot and multi-shot points. It
proposes a novel method to compute player-specific salience weights from
point-level data and quantifies the economic consequences of process-driven
strategy, demonstrating that small deviations from outcome maximization can
have substantial effects.

Second, more broadly, our results speak to a general class of preferences in
which individuals derive utility not only from outcomes but also from the
process by which these outcomes are achieved. This idea is closely related
to the compensating differentials literature in labor economics, which
documents that workers are willing to accept lower wages in exchange for
non-pecuniary job attributes such as meaningful or intrinsically rewarding
tasks, with estimated trade-offs on the order of 8--20\% depending on the
job attributes considered (e.g., \cite{Stern04,Mas17}). Similarly, in
consumer markets, a large body of evidence on fair trade, ethical
consumption, and product provenance shows that individuals are willing to
pay substantial price premia for goods produced under socially or
environmentally desirable conditions, at comparable product quality.
Experimental and field evidence typically finds willingness-to-pay premia in
the range of approximately 5\% to 30\% for fair trade or ethically certified
products, depending on product category (e.g., see \cite%
{Maertens09,Dragusanu14}). Across these settings, individuals appear willing
to trade off outcome against process, be it about having autonomy, providing
meaning, or being ethical. Our contribution is to show that this same
trade-off can be identified from high-frequency, within-individual,
continuous choices in a high-stakes, competitive strategic environment, in
contrast to labor and consumer studies, which typically rely on
low-frequency or discrete-choice settings.

The remainder of the paper is structured as follows. Section 2 presents the
model, discusses non-parametric bounds and introduces a parametric approach
together with an algorithm to recover player-specific salience weights from
the data. Section 3 describes the data and the estimation method, while
Section 4 presents the results. Section 5 provides robustness checks, and
Section 6 concludes.

\section{Model}

\subsection{Preliminary observations}

In tennis, a central strategic decision for the server is how much risk to
take. A ``safe'' serve increases the probability that the ball lands in the
service box (i.e., a higher serve percentage) but reduces the likelihood of
winning the point outright through an ace or unreturned serve. Conversely, a
``risky'' serve raises the probability of winning the point immediately, at
the cost of a higher probability of a fault.

While this resembles a standard high-risk--high-reward trade-off, the tennis
setting involves an additional margin: not only whether the point is won,
but how it is won. A safe serve increases the likelihood that the point
evolves into a multi-shot rally, forcing the player to win through multiple
shots, whereas a risky serve increases the likelihood of a one-shot win.

The server's decision, therefore, determines both the overall probability of
winning the point and the distribution of this probability over one-shot and
multi-shot points. As a result, the serve strategy reflects a trade-off not
only between high-risk--high-reward outcomes, but also between one-shot wins
and multi-shot wins, which may be intrinsically valued differently.

In support of this distinction, we propose the following defintions.

\textbf{Definition (Outcome Utility):} Outcome utility captures preferences
over the results of an action, independent of how they are achieved. In our
setting, it corresponds to the probability of winning a point, regardless of
whether it is won through one-shot or multi-shot rallies.

\textbf{Definition (Process Utility):} Process utility captures preferences
over how outcomes are achieved. In our setting, it is represented by a
preference for winning points through multi-shot rallies rather than
one-shot points.

\subsection{Set up}

Let the probability of a serve being in be denoted by $x\in \left[ 0,1\right]
$. This probability reflects the choice of the server. As depicted in the
above observations, if the server wants to take more risk, he will choose a
lower value of $x$, hence a lower serve percentage. In contrast, if he wants
to take fewer risks, he will choose a higher value of $x$, a higher serve
percentage. Since in tennis, players can serve a second serve if the first
is out, the server's strategy consists, in fact, of two numbers $x_{1}$ and $%
x_{2}$ reflecting respectively the probability of the first and second serve
to be in. A player's probability to win a point on his serve, denoted $%
p\left( x_{1},x_{2}\right) $, depends on his strategy $\left(
x_{1},x_{2}\right) $. Denote $y\left( x\right) $ the probability of winning
a point conditional on the serve being in as a function of the serve
probability, $x$. $y\left( x\right) $ reflects the skills of the player,
encompassing both his serving and rally skills (relative to his opponent).
In condition (\ref{Assumption:Theory}) below we assume that $y\left(
x\right) $ is twice differentiable on $x\in \left[ 0,1\right] $ and in
particular, strictly decreasing, i.e., $y^{\prime }\left( x\right) <0$, so
that the safer the serve, that is, the higher the probability that it is in,
the lower the probability of winning the point, conditional on the serve
being in, and strictly concave, $y^{\prime \prime }\left( x\right) <0$.

With these definitions, the probability of winning a point on one's own
serve, the outcome utility, reads as,%
\begin{equation*}
p\left( x_{1},x_{2}\right) =w\left( x_{1}\right) +\left( 1-x_{1}\right)
w\left( x_{2}\right) ,
\end{equation*}%
where $w\left( x\right) :=xy\left( x\right) $ is the unconditional
probability of winning a point, and a server aiming at maximizing his
probability of winning a point on his serve then does $\max_{x_{1},x_{2}}p%
\left( x_{1},x_{2}\right)$.

This setting corresponds to the basis of the model presented in \cite%
{klaassen09} and discussed in more detail in Online Appendix (\ref{Model:KM}%
). In this paper, we depart from \cite{klaassen09} by assuming that players
are perfect maximizers and allowing them to care about process utility, and
hence possibly put different weights to the various possible ways of winning
a point. This requires distinguishing between 4 possible ways to win a
point: with one shot on the first or second serve, i.e., an ace or an
unreturned serve, or with multiple shots on the first or second serve, i.e.,
a rally of more than 2 shots. We hence decompose the conditional probability
to win a point into a conditional probability to win in one shot, say $%
f\left( x\right) $, and in multiple shots, say $k\left( x\right) $. By
definition, one has $y\left( x\right) =f\left( x\right) +k\left( x\right) $. 
%
It seems natural to expect that $f^{\prime }\left( x\right) <0$ and $%
f^{\prime \prime }\left( x\right) <0$ so $f$ is strictly concave, meaning
that the conditional probability of winning a one-shot point is decreasing
with the probability of the serve to be in (increasing with risk), more so
as the level of risks decreases (concave).\footnote{%
This assumption is supported in the data. Indeed, for all professional
players in the data, the percentage of first serves in is lower than the
percentage of second serves in ($x_{1}<x_{2}$), but only for two players,
the share of one-shot points won on first serve ($f\left( x_{1}\right) $) is
lower than that on the second serve ($f\left( x_{2}\right) $).} In contrast,
the conditional probability of winning a multi-shot point may be increasing
or decreasing with the probability that the serve is in, depending on the
skills of the player. Hence, we impose that, if it is decreasing, it is also
concave, i.e., $k^{\prime }\left( x\right) <0$ and $k^{\prime \prime }\left(
x\right) \leq 0$, whereas, if it is increasing, it is convex, i.e., $%
k^{\prime }\left( x\right) >0$ and $k^{\prime \prime }\left( x\right) >0$.

To summarize, we assume that the following standing assumptions hold.

\begin{condition}
\label{Assumption:Theory}The conditional probability of winning a one-shot
point continuous and twice differentiable, strictly decreasing $f^{\prime
}\left( x\right) <0$ and concave $f^{\prime \prime }\left( x\right) <0$,
whereas the conditional probability to win a multi-shot point is continuous
and either strictly decreasing $k^{\prime }\left( x\right) <0$ and concave $%
k^{\prime \prime }\left( x\right) \leq 0$, or increasing $k^{\prime }\left(
x\right) \geq 0$ and convex $k^{\prime \prime }\left( x\right) \geq 0$ and
so that $y\left( x\right) :=f\left( x\right) +k\left( x\right) $ is so that $%
y^{\prime }\left( x\right) <0$ and $y^{\prime \prime }\left( x\right) <0$.
\end{condition}

The second departure from the model in \cite{klaassen09}, is that we
consider the case where a player may attach more or less, but not
necessarily the same, importance to winning a one-shot point rather than a
multi-shot point. In the model, there are four possible (winning) outcomes
for the server, listed below with their specific probability to occur and
specific utility:

\begin{enumerate}
\item One shot on first serve: probability $x_{1}f\left( x_{1}\right) $ and
utility of $\alpha$,

\item Multiple shots on first serve: probability $x_{1}k\left( x_{1}\right) $
and utility of $\beta $,

\item One shot on second serve: probability $\left( 1-x_{1}\right)
x_{2}f\left( x_{2}\right) $ and utility of $\alpha$,

\item Multiple shots on second serve: probability $\left( 1-x_{1}\right)
x_{2}k\left( x_{2}\right) $ and utility of $\beta $.
\end{enumerate}

Hence, we assume that the player maximizes, not his probability to win a
point $p\left( x_{1},x_{2}\right) $, but rather, $\tilde{p}\left(
x_{1},x_{2}\right) $ defined as the weighted average of the probability to
win a one-shot point and the probability to win a multi-shot point, which
reads as%
\begin{equation*}
\tilde{p}\left( x_{1},x_{2}\right) =\tilde{w}\left( x_{1}\right) +\left(
1-x_{1}\right) \tilde{w}\left( x_{2}\right) .
\end{equation*}%
where $\tilde{w}\left( x\right) =x\tilde{y}\left( x\right) $, $\tilde{y}%
\left( x\right) =\alpha f\left( x\right) +\beta k\left( x\right) $, $\alpha $
is the utility attached to winning a one-shot point, and $\beta $ is the
utility attached to winning a multi-shot point.

A first important remark is that normalizing the utility of winning one-shot
rallies to $\alpha =1$ is without loss of generality, as it does not affect
the optimal solution of a player. From now on, we therefore set $\alpha =1$
and interpret $\beta $ as the \emph{relative} preference parameter for
multi-shot rallies. Moreover, in the case $\beta =1$, distinguishing between 
$f$ and $k$ is irrelevant since all that matters is the conditional
probability of winning a point $y\left( x\right) =f\left( x\right) +k\left(
x\right) $ and not how. In terms of notation, we call $\tilde{y}$ the \emph{%
perceived} conditional probability of winning a point because of the
presence of preference parameter $\beta $ in the expression. It only
coincides with the true probability when $\beta =1$. A similar
interpretation and notation is used for $w(x)$ and $p\left(
x_{1},x_{2}\right) $.

A second important remark is that this utility rewrites as%
\begin{equation*}
\tilde{p}\left( x_{1},x_{2}\right) =\underset{\text{Outcome Utility}}{%
p\left( x_{1},x_{2}\right) }+\underset{\text{Process Utility}}{\left( \beta
-1\right) \left( x_{1}k\left( x_{1}\right) +\left( 1-x_{1}\right)
x_{2}k\left( x_{2}\right) \right) }
\end{equation*}%
clearly distinguishing the outcome utility, i.e., the probability of winning
a point on one's own serve, and process utility, i.e., the probability of
winning a multi-shot point on one's own serve. Presented this way, it is
clear that our model relates to the simple model of outcome and process
utility presented in the introduction, where the salience weight $\delta $
obtains as $\delta =\beta -1$.

To summarize, our setting distinguishes between one-shot and multi-shot
rallies and allows preference weights to be different for these two types of
rallies, replacing $w\left( x\right) $ by $\tilde{w}\left( x\right) $ and
decomposing $y\left( x\right) $ into the constituants $f\left( x\right) $
and $k\left( x\right) $ to compute $\tilde{y}\left( x\right) $.

The FOCs to this problem are obtained as%
\begin{eqnarray*}
\tilde{w}^{\prime }\left( x_{1}\right) -\tilde{w}\left( x_{2}\right) &=&0, \\
\tilde{w}^{\prime }\left( x_{2}\right) &=&0.
\end{eqnarray*}


The second order conditions require that the expected utility is concave.
For this, the Hessian of the expected utility needs to be semi-definite
negative and since at optimum, the Hessian is diagonal (see Online Appendix (%
\ref{FOCs:Theory})), the SOCs therefore are%
\begin{eqnarray*}
2\tilde{y}^{\prime }\left( x_{1}\right) +x_{1}\tilde{y}^{\prime \prime
}\left( x_{1}\right) &\leq &0, \\
\left( 1-x_{1}\right) \left( 2\tilde{y}^{\prime }\left( x_{2}^{\ast }\right)
+x_{2}^{\ast }\tilde{y}^{\prime \prime }\left( x_{2}^{\ast }\right) \right)
&\leq &0.
\end{eqnarray*}


Interestingly, the SOCs provide restrictions on the curvature of the
perceived conditional probability of winning a point. Indeed, rearranging
both SOCs, one obtains $x_{j}^{\ast }\frac{\tilde{y}^{\prime \prime }\left(
x_{j}^{\ast }\right) }{\tilde{y}^{\prime }\left( x_{j}^{\ast }\right) }\geq
-2$, provided that $\tilde{y}^{\prime }\left( x_{j}^{\ast }\right) <0$, $%
\forall j=1,2$.\footnote{%
Note that, while it is possible that $\beta $ is so that $\tilde{y}^{\prime
}\left( x\right) >0$, even if $y^{\prime }\left( x\right) <0$, in that case,
the optimal second serve strategy would be $x_{2}^{\ast }=1$ since $\tilde{w}%
^{\prime }\left( x\right) >0$. For all players in the data, the observed
second serve percentage is strictly lower than $1$, and hence this situation
never occurs.} Note that $x\frac{\tilde{y}^{\prime \prime }\left( x\right) }{%
\tilde{y}^{\prime }\left( x\right) }$ is the elasticity of the perceived
marginal probability $\tilde{y}^{\prime }\left( x\right) $ and the SOCs, in
fact, indicate that this elasticity should be larger than $-2$ on the
interval $x\in \left[ x_{1}^{\ast },x_{2}^{\ast }\right] $. Moreover, the
FOCs reveal important information about the shape of the perceived
conditional probability $\tilde{y}$ at the optimum. Indeed, consider the FOC
associated with the optimal second serve strategy. Rearranging, one has $%
\tilde{y}^{\prime }\left( x_{2}^{\ast }\right) =\frac{\tilde{y}\left(
x_{2}^{\ast }\right) }{x_{2}^{\ast }}$. By a similar procedure, the FOC for
the first serve strategy obtains as $\tilde{y}^{\prime }\left( x_{1}^{\ast
}\right) =\frac{x_{2}^{\ast }\tilde{y}\left( x_{2}^{\ast }\right) -\tilde{y}%
\left( x_{1}^{\ast }\right) }{x_{1}^{\ast }}$. It follows that the optimal
serve strategy $\left( x_{1}^{\ast },x_{2}^{\ast }\right) $ pins down the
slope of the perceived conditional probability of winning a point at both $%
x_{2}^{\ast }$ and $x_{1}^{\ast }$ and hence the average curvature of $%
\tilde{y}\left( x\right) $ in the interval $\left[ x_{1}^{\ast },x_{2}^{\ast
}\right] $ as%
\begin{equation*}
\frac{1}{x_{2}^{\ast }-x_{1}^{\ast }}\int_{x_{1}^{\ast }}^{x_{2}^{\ast }}%
\tilde{y}^{\prime \prime }\left( x\right) dx=\frac{\frac{\tilde{y}\left(
x_{2}^{\ast }\right) }{x_{2}^{\ast }}-\frac{x_{2}^{\ast }\tilde{y}\left(
x_{2}^{\ast }\right) -\tilde{y}\left( x_{1}^{\ast }\right) }{x_{1}^{\ast }}}{%
x_{2}^{\ast }-x_{1}^{\ast }}.
\end{equation*}

Figure (\ref{fig_OptRF}) shows how the optimal service strategy is
determined given the skills parameters of the player (the shapes of $f$ and $%
k$) and his relative preference for winning multi-shot rallies ($\beta $).%
\footnote{%
The additional source of identification provided by distinguishing between
one-shot and multi-shot rallies is illustrated in Figure (\ref{fig_Opt2ndRF}%
) of Online Appendix (\ref{SecondServe}).} We herewith use the example of
Roger Federer.

\begin{figure}[tbp]
\caption{Optimal service strategy: Roger Federer}
\label{fig_OptRF}\centering
\includegraphics[scale=1]{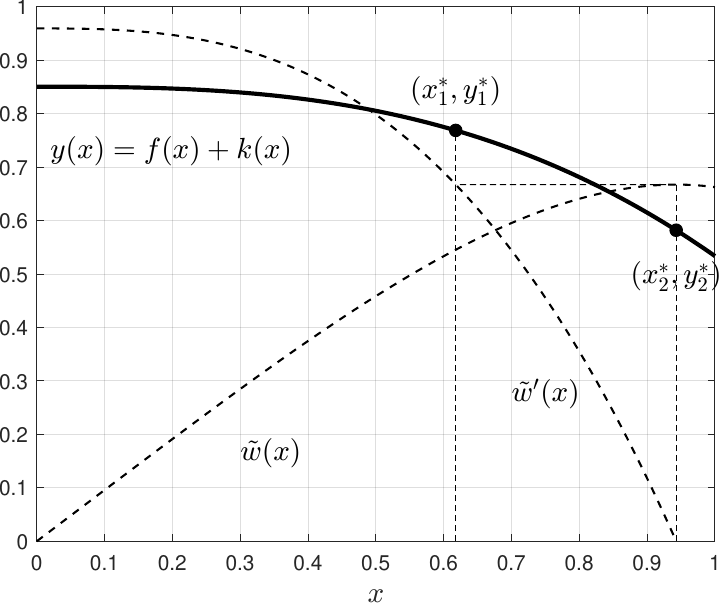}
\end{figure}

First, the second serve strategy is found by looking at the value of $x$ for
which $\tilde{w}^{\prime }\left( x\right) =0$, say $x_{2}^{\ast }$. Then,
the first serve strategy is derived by looking at the value $x$ for which $%
\tilde{w}^{\prime }\left( x\right) =\tilde{w}\left( x_{2}^{\ast }\right) $,
say $x_{1}^{\ast }$. The conditional probabilities of winning a point given
that the first (second) serve is in are indicated on the curve $y\left(
x\right) $, for the respective values of $x_{1}^{\ast }$ and $x_{2}^{\ast }$%
. This delivers two observed points $\left( x_{1}^{\ast },y_{1}^{\ast
}\right) $ and $\left( x_{2}^{\ast },y_{2}^{\ast }\right) $.

\subsection{Nonparametric approach}

Our aim is to nonparametrically bound the salience weight $\delta $. The key
idea is that optimal serve choices impose inequality restrictions on $\tilde{%
p}(.,.)$, which translate into player-specific bounds on $\delta $. To show
this, we first need to briefly introduce the data. Suppose that, for each
player $i=1,...,N$, we observe the probabilities $\left(
x_{1i},x_{2i},f_{1i},f_{2i},k_{1i},k_{2i}\right) $ where $x_{1i}$ and $%
x_{2i} $ are the probabilities of first and second serves in, $%
f_{1i}=f\left( x_{1i}\right) $ and $f_{2i}=f\left( x_{2i}\right) $ are the
conditional probabilities of winning a point with one shot on first and
second serves, respectively, and $k_{1i}=k\left( x_{1i}\right) $ and $%
k_{2i}=k\left( x_{2i}\right) $ are the conditional probabilities of winning
a point with multiple shots on first and second serves, respectively.%
\footnote{%
In the data section, we discuss how to estimate these probabilities $\left(
x_{1i},x_{2i},f_{1i},f_{2i},k_{1i},k_{2i}\right) $ given data on rallies
played, in possibly multiple matches, on the serve of each player $i$.}

We assume that each player $i$ is a perfect maximizer of his perceived
probability of winning a point $\tilde{p}\left( x_{1},x_{2}\right) $ so that
the observed data reflect the optimum of each player, i.e., $\left(
x_{1i},x_{2i}\right) =\left( x_{1i}^{\ast },x_{2i}^{\ast }\right) $.
Dropping the index $i$ for notational simplicity, optimality implies for
each player%
\begin{align}
\tilde{p}(x_{1}^{\ast },x_{2}^{\ast })& \geq \tilde{p}(x_{1}^{\ast
},x_{1}^{\ast }),  \label{eqNPI1} \\
\tilde{p}(x_{1}^{\ast },x_{2}^{\ast })& \geq \tilde{p}(x_{2}^{\ast
},x_{2}^{\ast }).  \label{eqNPI2}
\end{align}

The first inequality compares the observed (optimal) strategy to a deviation
in which the player uses the first-serve strategy in both serves; the second
compares it to always using the second-serve strategy.\footnote{%
Note that inequality (\ref{eqNPI1}) implies $\tilde{p}\left( x,x_{2}^{\ast
}\right) \geq \tilde{p}\left( x,x_{1}^{\ast }\right) $ for all $x$, so that
the third inequality derived from optimality, i.e., $\tilde{p}(x_{1}^{\ast
},x_{2}^{\ast })\geq \tilde{p}(x_{2}^{\ast },x_{1}^{\ast })$, is in fact
implied by inequalities (\ref{eqNPI1}-\ref{eqNPI2}). Indeed, note that, as
long as one maintains the first serve strategy constant, the first serve
strategy contributes the same term $\tilde{w}\left( x\right) $ and a slope $%
\left( 1-x\right) $ to $\tilde{p}\left( x,x_{2}\right) $ and $\tilde{p}%
\left( x,x_{1}\right) $ so that if the inequality $\tilde{p}\left(
x,x_{2}\right) >\tilde{p}\left( x,x_{1}\right) $ holds for $x$, it holds for
all $x^{\prime }\neq x$.} Each inequality above can be written in linear
form $A+\delta B\geq 0$ implying $\delta \geq -\frac{A}{B}$ if $B>0$ and $%
\delta \leq -\frac{A}{B}$ if $B<0$. The parameter $A$ is in fact the
difference in the probability of winning a point ($p\left( .,.\right) $ not $%
\tilde{p}\left( .,.\right) $) between the two strategies being compared in
the inequality, whereas the parameter $B$ is a similar difference, but for
the probability of winning a multi-shot point. By construction, $A-B:=C$ is
the same difference, but for the probability of winning a one-shot point.

The optimality conditions can be used to define a lower bound $L$ (when $B>0$%
) and an upper bound $U$ (when $B<0$) on the value of the salience weight $\delta
\in \left[ L,U\right] $ of each player. For all 151 players in our data, we
note that the lower bound $L$ is determined by inequality (\ref{eqNPI1})
whereas the upper bound is either unrestricted ($+\infty $) by the
optimality conditions, which is the case for 23 players, or determined by
inequality (\ref{eqNPI2}). Regarding the sign of the bounds, the data
actually show that for 8 players the lower bound is positive ($B>0$ and $A<0$%
) so that we can conclude that the salience weight of these players must be
positive. For the remaining players, the lower bound is negative ($B>0$ and $%
A>0$), and the upper bound positive ($B<0$ and $A>0$).

The optimality conditions alone are only enough to determine the sign of the
salience weight $\delta $ of 8 players. However, these conditions highlight
a strategy to derive a sufficient condition for a positive lower bound that
can be checked against the data for the remaining players. We note that the
lower bound $L$ is determined by an upper bound of $\tilde{p}\left(
x_{1}^{\ast },x_{1}^{\ast }\right) $, i.e., inequality (\ref{eqNPI1}).
Building on this observation, we therefore investigate whether there exists
a $x_{0}\in \left[ x_{1}^{\ast },x_{2}^{\ast }\right] $ so that the
inequality%
\begin{equation}
\tilde{p}\left( x_{0},x_{2}^{\ast }\right) \geq \tilde{p}\left( x_{1}^{\ast
},x_{1}^{\ast }\right)  \label{eqLowerBtb1}
\end{equation}%
is satisfied, where the associated values of $A$ and $B$ for this inequality
are%
\begin{eqnarray*}
A\left( x_{0}\right) &=&p\left( x_{0},x_{2}^{\ast }\right) -p\left(
x_{1}^{\ast },x_{1}^{\ast }\right) =B\left( x_{0}\right) +C\left(
x_{0}\right) , \\
B\left( x_{0}\right) &=&x_{0}k\left( x_{0}\right) +\left( 1-x_{0}\right)
x_{2}^{\ast }k_{2}^{\ast }-\left( 2-x_{1}^{\ast }\right) x_{1}^{\ast
}k_{1}^{\ast }, \\
C\left( x_{0}\right) &=&x_{0}f\left( x_{0}\right) +\left( 1-x_{0}\right)
x_{2}^{\ast }f_{2}^{\ast }-\left( 2-x_{1}^{\ast }\right) x_{1}^{\ast
}f_{1}^{\ast }.
\end{eqnarray*}

If there exists a $x_{0}\in \left[ x_{1}^{\ast },x_{2}^{\ast }\right] $ so
that $A\left( x_{0}\right) \leq 0$ and $B\left( x_{0}\right) >0$ then this
inequality delivers a positive lower bound on the salience weight.\footnote{%
We already know that for 8 players in the data when $x_{0}=x_{1}^{\ast }$,
which corresponds to Inequality (\ref{eqNPI1}), $A\left( x_{1}^{\ast
}\right) <0$ and $B\left( x_{1}^{\ast }\right) >0$.}

We proceed in two steps. First, we show in Lemma (\ref{LemmaNonPara}),
conditions under which there exists a unique $x_{0}\in \left[ x_{1}^{\ast
},x_{2}^{\ast }\right] $ so that $A\left( x_{0}\right) =0$. Then, we show in
Lemma (\ref{LemmaNonPara2}) conditions under which $B\left( x_{0}\right) $
is positive.

\begin{lemma}
\label{LemmaNonPara}If (a): Condition (\ref{Assumption:Theory}) is
satisfied, (b): $B\left( x_{1}^{\ast }\right) >A\left( x_{1}^{\ast }\right)
>0$, and (c): $B\left( x_{2}^{\ast }\right) >0>A\left( x_{2}^{\ast }\right) $%
, then there exists a unique $x_{0}\in \left[ x_{1}^{\ast },x_{2}^{\ast }%
\right] $ so that $A\left( x_{0}\right) =0$.
\end{lemma}

\begin{proof}
(Sketch) By continuity and an application of the Intermediate Value Theorem
there exists a $x_{0}\in \left[ x_{1}^{\ast },x_{2}^{\ast }\right] $ so that 
$A\left( x_{0}\right) =0$. Uniqueness follows from the strict concavity of $%
A\left( x\right) $ ($y\left( x\right) $). See Appendix (\ref{App:Lemma}) for
the full proof.
\end{proof}

Conditions (b)--(c) of Lemma (\ref{LemmaNonPara}) are sign restrictions
directly verifiable in the data. As it turns out, we find that, for all but
one player, these conditions are satisfied and hence Lemma (\ref%
{LemmaNonPara}) applies.

Next, we want to show conditions under which $B\left( x_{0}\right) >0$. We
first note that, since $A\left( x_{0}\right) =0$,%
\begin{eqnarray}
B\left( x_{0}\right) &>&0\Leftrightarrow C\left( x_{0}\right) <0  \notag \\
&\Leftrightarrow &x_{0}f\left( x_{0}\right) <x_{2}^{\ast }f_{2}^{\ast
}x_{0}-x_{2}^{\ast }f_{2}^{\ast }+\left( 2-x_{1}^{\ast }\right) x_{1}^{\ast
}f_{1}^{\ast }.  \label{eqIneqPara}
\end{eqnarray}

Inequality (\ref{eqIneqPara}) cannot directly be tested in the data since we
do not observe $x_{0}$. However, a sufficient condition for this inequality
to hold is%
\begin{equation}
xf\left( x\right) <x_{2}^{\ast }f_{2}^{\ast }x-x_{2}^{\ast }f_{2}^{\ast
}+\left( 2-x_{1}^{\ast }\right) x_{1}^{\ast }f_{1}^{\ast }\text{, }\forall
x\in \left[ x_{1}^{\ast },x_{2}^{\ast }\right] ,  \label{eqIneqPara2}
\end{equation}%
and Lemma (\ref{LemmaNonPara2}) below shows that one can quantify the extent
to which Inequality (\ref{eqIneqPara2}) restricts the graph of the function $%
m\left( x\right) :=xf\left( x\right) $ using observable data for each player
and hence quantify how likely this inequality is met.

\begin{lemma}
\label{LemmaNonPara2}If (a): Condition (\ref{Assumption:Theory}) is
satisfied, and (b): $x_{1}^{\ast }f_{1}^{\ast }>x_{2}^{\ast }f_{2}^{\ast }$,
then, the larger $A_{2}/A_{1}$ the more restrictive the sufficient
condition, i.e., Inequality (\ref{eqIneqPara2}), and the less likely the
condition $B\left( x_{0}\right) >0$, where%
\begin{equation*}
A_{1}=\frac{1}{2}\left\vert x_{2}^{\ast }\left( f_{1}^{\ast }-f_{2}^{\ast
}\right) \left( x_{12}-x_{1}^{\ast }\right) \right\vert .
\end{equation*}%
and%
\begin{equation*}
A_{2}=\QATOPD\{ . {\frac{1}{2}\left\vert \left( x_{12}-x_{14}\right) \left( 
\frac{x_{2}^{\ast }f_{2}^{\ast }}{1-x_{2}^{\ast }}\left( 1-x_{24}\right)
-f_{1}^{\ast }x_{24}\right) \right\vert \text{, if }x_{14}<x_{24}}{0\text{,
else}},
\end{equation*}%
and $x_{12}=\frac{x_{2}^{\ast }f_{2}^{\ast }x_{1}^{\ast }}{\left(
1-x_{2}^{\ast }\right) x_{1}^{\ast }f_{1}^{\ast }+x_{2}^{\ast }f_{2}^{\ast
}x_{1}^{\ast }}$, $x_{14}=b_{12}\frac{x_{1}^{\ast }}{x_{1}^{\ast
}f_{1}^{\ast }-x_{2}^{\ast }f_{2}^{\ast }x_{1}^{\ast }}$, $x_{24}=\frac{%
x_{2}^{\ast }f_{2}^{\ast }-\left( 1-x_{2}^{\ast }\right) b_{12}}{x_{2}^{\ast
}f_{2}^{\ast }\left( 2-x_{2}^{\ast }\right) }$ and $b_{12}=\left(
2-x_{1}^{\ast }\right) x_{1}^{\ast }f_{1}^{\ast }-x_{2}^{\ast }f_{2}^{\ast }$%
.
\end{lemma}

\begin{proof}
(Sketch) From the concavity of $m\left( x\right) $, and the fact that $%
m\left( 0\right) =0$ and $f\left( 1\right) \geq 0$, the graph of the
function $m\left( x\right) :=xf\left( x\right) $ on $\left[ x_{1}^{\ast
},x_{2}^{\ast }\right] $ must lie in the triangle $\mathcal{T}_{1}$ with
vertices $\left( x_{1}^{\ast },x_{1}^{\ast }f_{1}^{\ast }\right) $, $\left(
x_{2}^{\ast },x_{2}^{\ast }f_{2}^{\ast }\right) $ and $\left(
x_{12},f_{1}^{\ast }x_{12}\right) =\left( \frac{x_{2}^{\ast }f_{2}^{\ast
}x_{1}^{\ast }}{\left( 1-x_{2}^{\ast }\right) x_{1}^{\ast }f_{1}^{\ast
}+x_{2}^{\ast }f_{2}^{\ast }x_{1}^{\ast }},\frac{x_{1}^{\ast }f_{1}^{\ast
}x_{2}^{\ast }f_{2}^{\ast }}{\left( 1-x_{2}^{\ast }\right) x_{1}^{\ast
}f_{1}^{\ast }+x_{2}^{\ast }f_{2}^{\ast }x_{1}^{\ast }}\right) $. The area
of $\mathcal{T}_{1}$ obtains from the data as 
\begin{equation*}
A_{1}=\frac{1}{2}\left\vert x_{2}^{\ast }\left( f_{1}^{\ast }-f_{2}^{\ast
}\right) \left( x_{12}-x_{1}^{\ast }\right) \right\vert .
\end{equation*}%
Furthermore, from (b), the line whose expression appears on the right-hand
side of Inequality (\ref{eqIneqPara2}) either slices triangle $\mathcal{T}%
_{1}$ to form triangle $\mathcal{T}_{2}$ with vertices $\left(
x_{14},f_{1}^{\ast }x_{14}\right) $, $\left( x_{12},f_{1}^{\ast
}x_{12}\right) $ and $\left( x_{24},\frac{x_{2}^{\ast }f_{2}^{\ast }}{%
1-x_{2}^{\ast }}\left( 1-x_{24}\right) \right) $, where $x_{14}=b_{12}\frac{%
x_{1}^{\ast }}{x_{1}^{\ast }f_{1}^{\ast }-x_{2}^{\ast }f_{2}^{\ast
}x_{1}^{\ast }}$ and $x_{24}=\frac{x_{2}^{\ast }f_{2}^{\ast }-\left(
1-x_{2}^{\ast }\right) b_{12}}{x_{2}^{\ast }f_{2}^{\ast }\left(
2-x_{2}^{\ast }\right) }$ and $b_{12}=\left( 2-x_{1}^{\ast }\right)
x_{1}^{\ast }f_{1}^{\ast }-x_{2}^{\ast }f_{2}^{\ast }$, or passes above
triangle $\mathcal{T}_{1}$, i.e., $\mathcal{T}_{2}=\left\{ \emptyset
\right\} $, leaving the graph of $m\left( x\right) $ unrestricted. Since the
graph of $m\left( x\right) $ on $\left[ x_{1}^{\ast },x_{2}^{\ast }\right] $
must lie in the unrestricted area (trapezoid) $\mathcal{T}_{1}\backslash 
\mathcal{T}_{2}$ in order to satisfy Inequality (\ref{eqIneqPara2}), the
extent to which this inequality is restrictive can be measured as the share
of the area of $\mathcal{T}_{2}$ relative to that of $\mathcal{T}_{1}$,
where the area of $\mathcal{T}_{2}$ is given as%
\begin{equation*}
A_{2}=\frac{1}{2}\left\vert \left( x_{12}-x_{14}\right) \left( \frac{%
x_{2}^{\ast }f_{2}^{\ast }}{1-x_{2}^{\ast }}\left( 1-x_{24}\right)
-f_{1}^{\ast }x_{24}\right) \right\vert ,
\end{equation*}%
if $x_{14}<x_{24}$ and $A_{2}=0$ otherwise. The larger the share $%
A_{2}/A_{1} $ the more restrictive the sufficient condition, i.e.,
Inequality (\ref{eqIneqPara2}), and the less likely the condition $B\left(
x_{0}\right) >0$ is satisfied.
\end{proof}

Note that, condition (b) and quantity $A_{2}/A_{1}$ of Lemma (\ref%
{LemmaNonPara2}) can be checked/computed for each player in the data. We
find that condition (b) is satisfied for all but one player (Pedro Martinez)
so that Lemma (\ref{LemmaNonPara2}) applies for these players. Regarding the
share $A_{2}/A_{1}$, we find the following interesting results. First, for
the 8 players whose optimality conditions already guarantee a positive lower
bound on the salience weight, we find that $A_{2}/A_{1}=0$ for 7 of them,
and $A_{2}/A_{1}=11.4\%$ for the remaining one (Mats Wilander). This confirms
that the metric $A_{2}/A_{1}$ is informative about the likelihood that the
sufficient condition for having a positive lower bound for $\delta $ is met.
We further find that the sufficient condition for a positive salience weight
is met with certainty for 15 players ($A_{2}/A_{1}=0$) and restricts the
feasible area by $5\%$ or less for 69 players. For 101 players, the
restriction represents less than $10\%$; it is more than $20\%$ for only 20
players. Note that the median value of $A_{2}/A_{1}$ is $5.8\%$. Finally, as
a last means of comparison, we note that there are 110 players with a value
of $A_{2}/A_{1}<11.4\%$, i.e., the value obtained for Mats Wilander, the
player known to have a positive lower bound for the salience weight from the
optimality conditions but a non $0$ ratio for $A_{2}/A_{1}$. We conclude
that the sufficient condition for the salience weight to be positive is
likely satisfied for a large majority of players.

\subsection{Parametric approach}

Our objective is to estimate the salience weight $\delta$ (or equivalently,
the relative preference for winning a multi-shot point $\beta$) to conduct
counterfactual analyses. This requires imposing a parametric structure on
the model, in particular on the functions $f(x)$ and $k(x)$. Recall that the
SOCs of the optimization problem give a restriction on the elasticity of $%
\tilde{y}^{\prime }\left( x\right) $, the perceived marginal probability of
winning a point. Since $\tilde{y}^{\prime }\left( x\right) =f^{\prime
}\left( x\right) +\beta k^{\prime }\left( x\right) $, we propose to
parametrize the elasticity of the marginal probabilities $f^{\prime }\left(
x\right) $ and $k^{\prime }\left( x\right) $ as in the following condition.

\begin{condition}
\label{Assumption:Elasticity}$x\frac{f^{\prime \prime }\left( x\right) }{%
f^{\prime }\left( x\right) }=x\frac{k^{\prime \prime }\left( x\right) }{%
k^{\prime }\left( x\right) }=\lambda -1>0$.\footnote{%
Note that this assumption corresponds to the CRRA\ assumption of the von
Neumann-Morgenstern utility function in the context of a classical problem
of expected utility. In Online Appendix (\ref{Section:Softmax}) we present
the model using the CARA assumption instead, i.e., $\frac{f^{\prime \prime
}\left( x\right) }{f^{\prime }\left( x\right) }=\frac{k^{\prime \prime
}\left( x\right) }{k^{\prime }\left( x\right) }=\lambda >0$. The associated
functions are of the form $f\left( x\right) =a_{f}+\tau _{f}\exp \left(
\lambda x\right) $.}
\end{condition}

Condition (\ref{Assumption:Elasticity}) imposes that the elasticities of the
marginal probability of winning a point with one-shot and with multi-shot,
conditional on the serve being in, are equal to each other and to a strictly
positive constant $\lambda -1$.

This condition has three main implications. First, the elasticity of both
the perceived and true marginal probability of winning a point is also equal
to $\lambda -1$, as under condition (\ref{Assumption:Elasticity}) one has $x%
\frac{\tilde{y}^{\prime \prime }\left( x\right) }{\tilde{y}^{\prime }\left(
x\right) }=x\frac{y^{\prime \prime }\left( x\right) }{y^{\prime }\left(
x\right) }=\lambda -1$.\footnote{%
Indeed, since one has $x\frac{f^{\prime \prime }\left( x\right) }{f^{\prime
}\left( x\right) }=x\frac{k^{\prime \prime }\left( x\right) }{k^{\prime
}\left( x\right) }=\lambda -1$, it follows that $xk^{\prime \prime }\left(
x\right) =\left( \lambda -1\right) k^{\prime }\left( x\right) $ and $%
xf^{\prime \prime }\left( x\right) =\left( \lambda -1\right) f^{\prime
}\left( x\right) $ and since $xk^{\prime \prime }\left( x\right) +xf^{\prime
\prime }\left( x\right) =xy^{\prime \prime }\left( x\right) $ and $%
xk^{\prime \prime }\left( x\right) +x\beta f^{\prime \prime }\left( x\right)
=x\tilde{y}^{\prime \prime }\left( x\right) $, one has%
\begin{eqnarray*}
x\tilde{y}^{\prime \prime }\left( x\right) &=&\left( \lambda -1\right) 
\tilde{y}^{\prime }\left( x\right) , \\
xy^{\prime \prime }\left( x\right) &=&\left( \lambda -1\right) y^{\prime
}\left( x\right) ,
\end{eqnarray*}%
which for $\tilde{y}^{\prime }\left( x\right) ,y^{\prime }\left( x\right) >0$
$\ $yields the result in the text.} Furthermore, since, by assumption $%
y^{\prime }\left( x\right) <0$, this means that $y\left( x\right) $ follows
the law of diminishing marginal \emph{returns} (read conditional probability
of winning a point).

Second, it means that $f\left( x\right) $ and $k\left( x\right) $ are power
functions of the form $f\left( x\right) =\frac{a_{f}-x^{\lambda }}{\tau _{f}}
$ and $k\left( x\right) =\frac{a_{k}-x^{\lambda }}{\tau _{k}}$, offering
great flexibility with only 5 unknown parameters.

Third, as a by product of the two preceeding remarks, the conditional
probability of winning a point $y\left( x\right) $ is itself a power
function with power $\lambda $, as indeed $y\left( x\right) =f\left(
x\right) +k\left( x\right) =\frac{a-x^{\lambda }}{\tau }$, where $\tau =%
\frac{\tau _{f}\tau _{k}}{\tau _{f}+\tau _{k}}$ and $a=\frac{a_{f}\tau
_{k}+a_{k}\tau _{f}}{\tau _{f}+\tau _{k}}$, see Online Appendix (\ref%
{SumPower}).\footnote{%
This corresponds to the functional shape assumed in Klaassen and Magnus
(2009) for $y\left( x\right) $.}

Associated with this parametric choice, Conditions (\ref{Assumption:Theory})
are met with the following restrictions on the parameters of $f\left(
x\right) $ and $k\left( x\right) $.

\begin{condition}
\label{Assumption:Parametric}i) $\lambda > 1$, ii) $\tau _{f}>0$ and either $%
\tau _{k}>0$ or $-\tau _{k}>\max \left( \tau _{f},\beta \tau _{f}\right) $,
iii) $0<\frac{a_{f}\tau _{k}+a_{k}\tau _{f}}{\tau _{k}+\tau _{f}}<\lambda +1$%
.
\end{condition}

Note that Condition (\ref{Assumption:Parametric}.ii) garantees that $y\left(
x\right) =f\left( x\right) +k\left( x\right) $ is decreasing, as indeed it
leads to $\tau =\frac{\tau _{f}\tau _{k}}{\tau _{f}+\tau _{k}}>0$, and $%
f^{\prime }\left( x\right) <0$ from $\tau _{f}>0$.\footnote{%
In our data, these conditions are met for all players except for the
condition $-\tau _{k}>\beta \tau _{f}$, which is not met for 3 of them.}

With these parametric choices, the optimal first and second serve strategies
can be derived in closed form from the first-order conditions (see Online
Appendix (\ref{FOCs:Parametric})). One obtains for the second serve strategy%
\begin{equation*}
x_{2}^{\ast \lambda }=\frac{a_{f}+a_{k}\beta \frac{\tau _{f}}{\tau _{k}}}{%
\left( 1+\lambda \right) \left( 1+\beta \frac{\tau _{f}}{\tau _{k}}\right) },
\end{equation*}%
yielding an interior solution, i.e., $0<x_{2}^{\ast }<1$, if and only if $%
\lambda +1>\frac{a_{f}+a_{k}\beta \frac{\tau _{f}}{\tau _{k}}}{1+\beta \frac{%
\tau _{f}}{\tau _{k}}}>0.$

And, the first serve strategy, therefore obtains, after simple substitution,
as%
\begin{equation*}
x_{1}^{\ast \lambda }=x_{2}^{\ast \lambda }\left( 1-\frac{\lambda }{%
1+\lambda }x_{2}^{\ast }\right) .
\end{equation*}

This is a remarkable result\footnote{%
This result arises not only with power functions but also with softmax
functions, see Online Appendix (\ref{Section:Softmax}).} indicating that,
when players are optimizers, data on first and second serve percentages
uniquely identify the curvature parameter $\lambda $ as shown in section (%
\ref{Section:IdentificationComputation}).

To summarize, under our standing assumptions, at optimality, the server
adopts the following service strategy

\begin{equation*}
\left( x_{1}^{\ast \lambda },x_{2}^{\ast \lambda }\right) =\left(
x_{2}^{\ast \lambda }-\frac{\lambda }{1+\lambda }x_{2}^{\ast }x_{2}^{\ast
\lambda },\frac{a_{f}+a_{k}\beta \frac{\tau _{f}}{\tau _{k}}}{\left(
1+\lambda \right) \left( 1+\beta \frac{\tau _{f}}{\tau _{k}}\right) }\right)
,
\end{equation*}%
and enjoys the following conditional probabilities of winning a point in one
and multiple shots on first and second serve:%
\begin{eqnarray*}
f\left( x_{1}^{\ast }\right) &=&\frac{a_{f}-x_{1}^{\ast \lambda }}{\tau _{f}}%
,\text{ }k\left( x_{1}^{\ast }\right) =\frac{a_{k}-x_{1}^{\ast \lambda }}{%
\tau _{k}}, \\
f\left( x_{2}^{\ast }\right) &=&\frac{a_{f}-x_{2}^{\ast \lambda }}{\tau _{f}}%
,\text{ }k\left( x_{2}^{\ast }\right) =\frac{a_{k}-x_{2}^{\ast \lambda }}{%
\tau _{k}}.
\end{eqnarray*}

An important remark is that although the preference parameter $\beta $ and
the relative skills $\frac{\tau _{f}}{\tau _{k}}$ only enter the expressions
of the optimal serve strategy on first and second serve, through the term $%
\beta \frac{\tau _{f}}{\tau _{k}}$, the expressions for the conditional
probabilities of winning a one-shot or multi-shot point at the optimum
depend respectively, only (directly) on the skills parameters $\tau _{f}$
and $\tau _{k}$. At same value of $\beta \frac{\tau _{f}}{\tau _{k}}$, i.e.,
at same optimal serve strategy, players with different skills $\tau _{f}$
and $\tau _{k}$ have different optimal conditional probabilities of winning
a one-shot or multi-shot point. This is the source for the separate
identification of the preference parameter and the skills parameters to
exploit in the data.

$\left( x_{1}^{\ast },x_{2}^{\ast }\right) $ is an optimum if the expected
utility is concave at $\left( x_{1}^{\ast },x_{2}^{\ast }\right) $. With the
parametric shapes assumed above, the conditions for the expected utility to
be concave are given as%
\begin{eqnarray*}
-\tau _{f}\lambda \left( 1+\lambda \right) \left( x_{1}^{\ast }\right)
^{\lambda -1}\left( 1+\beta \frac{\tau _{f}}{\tau _{k}}\right) &\leq &0 \\
-\tau _{f}\left( 1-x_{1}\right) \lambda \left( 1+\lambda \right) \left(
x_{2}^{\ast }\right) ^{\lambda -1}\left( 1+\beta \frac{\tau _{f}}{\tau _{k}}%
\right) &\leq &0
\end{eqnarray*}

Since $\tau _{f}>0$, $\lambda >0$, $x_{1}^{\ast }>0$, and from Condition (%
\ref{Assumption:Parametric}.ii) one has either $\tau _{k}>0$ or $-\tau
_{k}>\beta \tau _{f}$ so that $1+\beta \frac{\tau _{f}}{\tau _{k}}>0$, we
conclude that the expected utility is concave.

It is easy to show by simple substitution that when $\beta =1$, the optimal
strategy is%
\begin{equation*}
\left( x_{1}^{\ast \lambda },x_{2}^{\ast \lambda }\right) =\left(
x_{2}^{\ast \lambda }\left( 1-\frac{\lambda }{\lambda +1}x_{2}^{\ast
}\right) ,\frac{a}{\lambda +1}\right)
\end{equation*}%
and the SOC reads as $\tau \geq 0$.\footnote{%
See Online Appendix \ref{CompStatics} for comparative statics.}

\subsection{Identification and computation\label%
{Section:IdentificationComputation}}

For each player $i=1,...,N$, we observe the probabilities $\left(
x_{1i},x_{2i},f_{1i},f_{2i},k_{1i},k_{2i}\right) $ and have 5 unknown skill
parameters $\left( \lambda _{i},a_{fi},\tau _{fi},a_{ki},\tau _{ki}\right) $
and 1 unknown preference parameter $\beta _{i}$. The two points $\left(
x_{1i},f_{1i}\right) $ and $\left( x_{2i},f_{2i}\right) $ can be used to
identify 2 of the 3 parameters of the function $f\left( x\right) $, whereas
the two points $\left( x_{1i},k_{1i}\right) $ and $\left(
x_{2i},k_{2i}\right) $ can be used to identify 2 of the 3 parameters of the
function $k\left( x\right) $. Since $f\left( x\right) $ and $k\left(
x\right) $ have one parameter in common, i.e., $\lambda _{i}$, this means
that the four points $\left( x_{1i},f_{1i}\right) $, $\left(
x_{1i},k_{1i}\right) $, $\left( x_{2i},f_{2i}\right) $ and $\left(
x_{2i},k_{2i}\right) $ together only identify 4 of the 5 skills parameters.
Assuming that these players are perfect optimizers, the optimality
conditions imply that $\left( x_{1i},x_{2i}\right) =\left( x_{1i}^{\ast
},x_{2i}^{\ast }\right) $ which provides 2 restrictions to the system of
equations. Hence, we have 6 parameters to be identified by 4 data points and
2 optimality restrictions.

\subsubsection{Identification\label{subSection:Identification}}

To show the identification of parameters $\left( \lambda _{i},a_{fi},\tau
_{fi},a_{ki},\tau _{ki},\beta _{i}\right) $ given data $\left(
x_{1i},x_{2i},f_{1i},f_{2i},k_{1i},k_{2i}\right) $, we first show
identification of parameters $\left( a_{fi},\tau _{fi},a_{ki},\tau
_{ki},\beta _{i}\right) $ given data $\left(
x_{1i},x_{2i},f_{1i},f_{2i},k_{1i},k_{2i}\right) $ conditional on $\lambda
_{i}$, and propose a bisection algorithm that searches for the curvature
parameter $\lambda $ given the data $\left( x_{1i},x_{2i}\right) $.

First, the slope and constant terms of the functions $f\left( x\right) $ and 
$k\left( x\right) $ are identified given $\lambda _{i}=\lambda $. Indeed, as
soon as the value of $\lambda $ is known, data $\left( x_{1i},x_{2i}\right) $
can be used to compute $\left( z_{1i},z_{2i}\right) =\left( x_{1i}^{\lambda
},x_{2i}^{\lambda }\right) $. It follows that, using the functional form for 
$f\left( x\right) $ and $k\left( x\right) $, from the points $\left(
z_{1i},f_{1i}\right) $ and $\left( z_{2i},f_{2i}\right) $, by simply
rearranging terms, one can deduce the slopes%
\begin{equation*}
\tau _{fi}=-\frac{z_{2i}}{f_{2i}}\frac{\Delta z_{i}}{\Delta f_{i}},\quad
\tau _{ki}=-\frac{z_{2i}}{k_{2i}}\frac{\Delta z_{i}}{\Delta k_{i}},
\end{equation*}%
where $\Delta l_{i}=\frac{l_{1i}-l_{2i}}{l_{2i}}$ $\forall l=z,f,k$, and
then the constants%
\begin{equation*}
a_{fi}=\tau _{fi}f_{1i}+z_{1i},\quad a_{ki}=\tau _{ki}k_{1i}+z_{1i}.
\end{equation*}

Second, one can uncover the preference parameter $\beta _{i}$ once the slope
and constant parameters of $f\left( x\right) $ and $k\left( x\right) $ are
known. This is done by using the previous results together with the equation
for the optimal second serve strategy to isolate%
\begin{equation*}
\beta _{i}=-\frac{\tau _{ki}}{\tau _{fi}}\frac{a_{fi}-z_{2i}\left( 1+\lambda
\right) }{a_{ki}-z_{2i}\left( 1+\lambda \right) }
\end{equation*}%
provided $z_{2i}\left( 1+\lambda \right) -a_{ki}\neq 0$. Using the
expressions previously obtained, for $x_{2i}>0$, $f_{2i}>0$ and $k_{2i}>0$,
it can be shown that in fact

\begin{equation*}
\beta _{i}=-\frac{f_{2i}}{k_{2i}}\frac{1+\lambda \frac{\Delta f_{i}}{\Delta
z_{i}}}{1+\lambda \frac{\Delta k_{i}}{\Delta z_{i}}},
\end{equation*}%
and it follows that one must have $\Delta z_{i}+\lambda \Delta k_{i}\neq 0$.

Note that after following these steps, the curvature parameter is the only
remaining unknown parameter. We can then use the last remaining condition,
i.e., the FOC for the optimal first serve strategy, to implicitly solve for
the curvature parameter and obtain%
\begin{equation*}
\lambda _{i}=\frac{a_{fi}+a_{ki}\beta _{i}\frac{\tau _{fi}}{\tau _{ki}}%
-x_{2i}\left( a_{fi}+a_{ki}\beta _{i}\frac{\tau _{fi}}{\tau _{ki}}-\left(
1+\beta _{i}\frac{\tau _{fi}}{\tau _{ki}}\right) z_{2i}\right) }{\left(
1+\beta _{i}\frac{\tau _{fi}}{\tau _{ki}}\right) z_{1i}}-1
\end{equation*}%
since $z_{1i}>0$.

Importantly, as shown in Online Appendix (\ref{ExpressionLambda}), this
expression simplifies considerably to read as%
\begin{equation*}
\lambda _{i}=\frac{z_{2i}}{z_{1i}}\left( 1+\lambda \left( 1-x_{2i}\right)
\right) -1.
\end{equation*}

This expression actually shows that the curvature parameter only depends on
data $x_{1i}$ and $x_{2i}$ through $z_{1i}$ and $z_{2i}$ and the initial
value of $\lambda $ selected. Of course, the value $\lambda _{i}$ herewith
obtained might be different than the $\lambda $ used to compute $z_{1i}$ and 
$z_{2i}$. However, as it turns out, this equation together with the
structure of the problem, provides a fixed-point, so that there exists, for
each player $i$, a $\lambda $ such that $\lambda _{i}=\lambda $. We shall
see below that a relatively simple algorithm allows us to recover this
value, for all the players in the data, and that this value is larger than
unity for all players.

\subsubsection{Computation}

We propose the following algorithm to compute the parameter $\lambda _{i}$
of each player $i$, from the associated data $\left( x_{1i},x_{2i}\right) $.

\begin{algorithm}
\label{Algorithm}Data $\left( x_{1i},x_{2i}\right) $, tolerance $\varepsilon 
$, Output $\lambda _{i}$.

\begin{enumerate}
\item Initializing step. Set $t=0$ and initialize the interval $\left[
l^{\left( t\right) },u^{\left( t\right) }\right] $ so that $\Lambda
_{i}\left( l^{\left( t\right) }\right) <l^{\left( t\right) }$ and $\Lambda
_{i}\left( u^{\left( t\right) }\right) >u^{\left( t\right) }$. In practice, $%
l^{\left( t\right) }$ should be close to $0$ and $u^{\left( t\right) }$
large enough.

\item Solving step. Set $\lambda =\frac{l^{(t)}+u^{\left( t\right) }}{2}$
and using data $\left( x_{1i},x_{2i}\right) $ compute%
\begin{equation*}
\lambda ^{new}=\left( \frac{x_{2i}}{x_{1i}}\right) ^{\lambda }\left(
1+\lambda \left( 1-x_{2i}\right) \right) -1.
\end{equation*}

\item Incrementation step. If $\lambda ^{new}<\lambda $, $l^{\left( t\right)
}=\lambda $, else $u^{\left( t\right) }=\lambda $. If $\left\vert l^{\left(
t\right) }-u^{\left( t\right) }\right\vert <\varepsilon $, stop, else, set $%
t\leftarrow t+1$ and go back to solving step.
\end{enumerate}
\end{algorithm}

Note that the iteration step of Algorithm (\ref{Algorithm}) defines a map $%
\Lambda _{i}:%
\mathbb{R}
_{0}^{+}\rightarrow 
\mathbb{R}
_{0}^{+}$, associating for each $\lambda ^{(t)}$ a new value $\lambda
^{(t+1)}$ for the curvature parameter given the data $\left(
x_{1i},x_{2i}\right) $. If it exists, a solution is therefore a fixed-point $%
\Lambda _{i}\left( \lambda _{i}^{\ast }\right) =\lambda _{i}^{\ast }$. The
following theorem shows that such a solution exists and is unique under mild
conditions for the structure of the data $\left( x_{1i},x_{2i}\right) $.

\begin{theorem}[Existence and uniqueness of the solution]
\label{thm:lambda_unique} For each player $i$, with data $\left(
x_{1i},x_{2i}\right) $ so that: (A1): $x_{1i},x_{2i}>0$ with $%
0<x_{1i}<x_{2i}<1$, (A2): $\ln \frac{x_{2i}}{x_{1i}}<x_{2i}$, (A3): $2\left(
x_{2i}-x_{1i}\right) <\left( x_{2i}\right) ^{2}$, there exists a unique $%
\lambda _{i}^{\ast }\geq 1$ such that $\lambda _{i}^{\ast }=\Lambda
_{i}(\lambda _{i}^{\ast })$.
\end{theorem}

\begin{proof}
See Appendix (\ref{App:Proof})
\end{proof}

We illustrate the theorem using the example of Roger Federer. Figure (\ref%
{fig_FixedPointRF}) plots the map $\Lambda _{i}(\lambda )$ for Roger
Federer. We see that the properties of the map exploited in the proof, are
so that the map crosses the 45 degree line only ones at $\lambda ^{\ast }$,
the point $\left( \lambda ^{\ast },\Lambda _{i}(\lambda ^{\ast })\right) $
forming the unique (non zero) fixed-point of the map. Suppose that the
algorithm starts with lower bound $l^{\left( 0\right) }=1$ and upper bound $%
u^{\left( 0\right) }=4$. The mid point value is 2.5. At $\lambda =2.5$, the
value of the map is below the 45 degree line and the lower bound is updated
to 2.5. The termination condition is not satisfied (at conventional levels
of precision) and the algortihm goes back to the iteration step with lower
bound 2.5 and upper bound 4. The mid point is now 3.75 so that the value of
the map is higher than the 45 degree line. The upper bound is updated to
3.75, etc.. The algorithm converges very fast to the value $\lambda ^{\ast
}=2.81$.

\begin{figure}[tbph]
\caption{Fixed-point of the map $\Lambda (\protect\lambda )$: Roger Federer.}
\label{fig_FixedPointRF}\centering\includegraphics[width=0.6%
\textwidth]{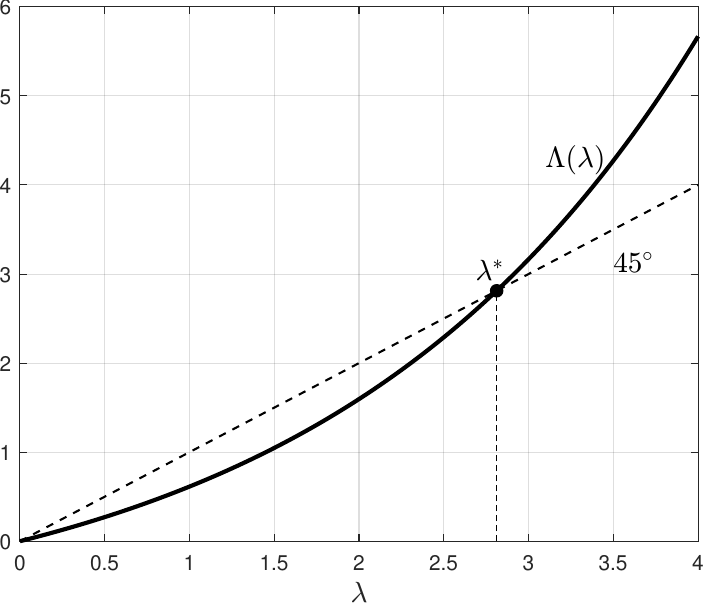}
\end{figure}

Once the curvature parameter $\lambda _{i}$ is obtained from the algorithm,
one can compute the remaining parameters using the steps outlined in the
Section (\ref{subSection:Identification}).

\section{Data}

We use the Match Charting Project by Jeff Sackmann,\footnote{%
https://github.com/JeffSackmann/tennis\_MatchChartingProject} which collects
information about professional tennis matches, encoded by dozens of
contributors. In particular, we use the point-by-point data for men's
matches with information on more than 1,200,000 rallies of over 7,100
matches by the end of January 2026.

The unit of observation is a rally in a match. For each rally, we know who
is serving, whether it is a rally on first or second serve, and the length
of the rally, i.e., the number of shots that were recorded ``in" the court.
A rally of length 1 is necessarily ending with either an ace or an
unreturned serve. The server wins all rallies of odd length, whereas the
returner wins rallies of even length. For our analysis, for each player, we
need to observe a large number of rallies on their own serve. For this
reason, we only select those players with at least 20 matches charted in the
data. 151 players satisfy this criterion.

\begin{table}
\caption{First and second serves statistics for selected players and aggregate sample.\label{SmallDesc}}
\footnotesize
\resizebox{\textwidth}{!}{ 
\begin{tabular}{l|c|cccc|cccc}
\toprule
 & & \multicolumn{4}{c}{First serve} & \multicolumn{4}{c}{Second serve} \\
\cline{3-6} \cline{7-10}
\hline
Player & Nb pt. & $\%$ in &  $\%$ Multi & $\%$ One won & $\%$ Multi won & $\%$ in & $\%$ Multi & $\%$ One won & $\%$ Multi won  \\
 & ($\times 1000$) & $x_{1}$ &  & $f(x_{1})$ & $k(x_{1})$ & $x_{2}$ & & $f(x_{2})$ & $k(x_{2})$  \\
\hline
Novak Djokovic & 44.8 & 64.9 & 65.9 & 34.1 & 39.2 & 92.0 & 83.2 & 16.8 & 42.0 \\
Rafael Nadal & 34.0 & 68.6 & 72.4 & 27.6 & 42.9 & 92.6 & 83.9 & 16.1 & 43.4 \\
Roger Federer & 58.0 & 61.8 & 58.5 & 41.5 & 35.4 & 94.3 & 80.8 & 19.2 & 39.0 \\
Pete Sampras & 20.3 & 56.5 & 46.4 & 53.6 & 26.7 & 89.7 & 78.5 & 21.5 & 34.2 \\
... &  &  &  &  &  &  &  &  &   \\
Boris Becker & 15.7 & 55.0 & 52.3 & 47.7 & 29.5 & 88.2 & 80.5 & 19.5 & 32.9 \\
... &  &  &  &  &  &  &  &  &   \\
Carlos Alcaraz & 16.5 & 64.6 & 68.0 & 32.0 & 40.0 & 90.9 & 82.2 & 17.8 & 42.0 \\
Jannik Sinner & 20.1 & 60.3 & 61.6 & 38.4 & 37.4 & 93.3 & 83.0 & 17.0 & 42.2 \\
... &  &  &  &  &  &  &  &  &   \\
Ivo Karlovic & 3.6 & 65.4 & 43.0 & 57.0 & 23.8 & 86.3 & 69.0 & 31.0 & 28.8 \\
John Isner & 7.9 & 69.5 & 45.9 & 54.1 & 23.5 & 91.9 & 73.6 & 26.4 & 32.1 \\
Reilly Opelka & 4.4 & 65.6 & 45.0 & 55.0 & 22.4 & 91.9 & 74.0 & 26.0 & 33.8 \\
... &  &  &  &  &  &  &  &  &   \\
David Ferrer & 6.0 & 61.7 & 76.0 & 24.0 & 40.2 & 91.3 & 86.8 & 13.2 & 38.9 \\
Diego Schwartzman & 3.9 & 63.9 & 78.4 & 21.6 & 42.0 & 89.6 & 85.7 & 14.3 & 41.2 \\
... &  &  &  &  &  &  &  &  &   \\
\hline
Mean & 6.4 & 61.1 & 63.6 & 36.4 & 35.2 & 90.9 & 83.1 & 16.9 & 37.4 \\
Std & 7.5 & 4.1 & 8.6 & 8.6 & 4.8 & 3.1 & 3.9 & 3.9 & 3.2 \\
Min & 1.3 & 51.5 & 39.4 & 19.5 & 21.3 & 74.0 & 59.3 & 10.2 & 27.1 \\
Median & 3.7 & 61.4 & 63.7 & 36.3 & 35.5 & 91.6 & 83.7 & 16.3 & 37.6 \\
Max & 58.0 & 71.3 & 80.5 & 60.6 & 43.3 & 96.8 & 89.8 & 40.7 & 43.7 \\
\bottomrule
\end{tabular}
}
\end{table}

Table~(\ref{SmallDesc}) presents service statistics for selected players,
along with sample descriptive statistics. The selected players include the
four players with the most Grand Slam titles\footnote{%
The four most prestigious tournaments of the year are the Australian Open,
Roland Garros, Wimbledon, and the US Open. These tournaments also have the
most generous prize money distributions.}, i.e., the Greatest Of All Times
(GOATs): Novak Djokovic, Rafael Nadal, Roger Federer, and Pete Sampras;
Boris Becker; the best two players of the current generation, Carlos Alcaraz
and Jannik Sinner; three players known for their big serves, John Isner,
Reilly Opelka, and Ivo Karlovic; and two players known for their baseline
game, David Ferrer and Diego Schwartzman.

The data include between $1{,}300$ and $58{,}000$ rallies per player. On
average, the first serve percentage is about $61\%$, while the second serve
percentage is roughly $91\%$. Players win approximately $72\%$ of points
played on their first serve, with an equal distribution between one-shot and
multi-shot rallies, whereas $63\%$ of points on the first serve are
multi-shot rallies. On the second serve, being more conservative, players
win fewer one-shot points ($17\%$) and about $37\%$ of multi-shot rallies.
The percentage of multi-shot rallies on second serves is $83\%$, roughly $20$
percentage points higher than on the first serve.

There are also notable disparities across players. Big servers, such as John
Isner, win about twice as many one-shot points as multi-shot points on their
first serve, whereas the reverse is nearly true for baseline specialists
like David Ferrer. Comparing newer top players to earlier ones, Carlos
Alcaraz's statistics resemble those of Rafael Nadal, and Jannik Sinner's
resemble those of Roger Federer. Finally, the percentage of points won on
one-shot rallies ranges from $10\%$ to $41\%$ on the second serve and from $%
20\%$ to $61\%$ on the first serve, while for multi-shot rallies, the
corresponding ranges are $27\%$ to $44\%$ on the second serve and $21\%$ to $%
43\%$ on the first serve.

\subsection{Maximum likelihood estimation of $\left(
x_{1i},x_{2i},f_{1i},f_{2i},k_{1i},k_{2i}\right) $}

Let $\theta _{i}=\left( x_{1i},x_{2i},f_{1i},f_{2i},k_{1i},k_{2i}\right) $
be the unknown probabilities of interest. From the data, for each player $%
i=1,...,N$, we observe $N_{i}$ points played on his serve. We can compute
the number of points played on the first serve, i.e., $n_{x_{1}i}$ and on
the second serve $n_{x_{2}i}=N_{i}-n_{x_{1}i}$ and since the data identifies
for each point, the length of the rally, i.e., the number of shots played
into the court during the point, we can also compute the number of rallies
of length 1 on first and second serves, i.e., $n_{f_{1}i}$ and $n_{f_{2}i}$
and the number of other rallies of odd length on the first and second serves
as well, i.e., $n_{k_{1}i}$ and $n_{k_{2}i}$.

Note that \cite{klaassen01a} showed that even though points in Tennis are
not i.i.d., the deviations from the i.i.d. hypothesis are small. As a
result, the i.i.d. hypothesis can still be used as a good approximation when
aggregating over a large number of points as we do in this paper, where we
use \textquotedblleft averages\textquotedblright\ over points by players. We
therefore maintain the assumption that points are i.i.d. so that each of the
aforementioned variables follows a binomial distribution.\footnote{%
See Online Appendix (\ref{Tree}) for a tree representation of the statistics
for each player.} For instance, on the first serve, one has:

\begin{enumerate}
\item the number of first serves in, $n_{x_{1}i}$, represents the number of
successes resulting from $N_{i}$ independent draws with unknown probability
of success $x_{1i}$, meaning that%
\begin{equation*}
n_{x_{1}i}\sim B\left( N_{i},x_{1i}\right) ,
\end{equation*}

\item the number of rallies of length 1 on first serves, $n_{f_{1}i}$,
represents the number of successes resulting from $n_{x_{1}i}$ independent
draws with unknown probability of success $f_{1i}$, meaning that%
\begin{equation*}
n_{f_{1}i}\sim B\left( n_{x_{1}i},f_{1i}\right) ,
\end{equation*}

\item the number of rallies won with multiple shots on first serve, $%
n_{k_{1}i}$, represents the number of successes resulting from $%
n_{x_{1}i}-n_{f_{1}i}$ independent draws with unknown probability of success 
$k_{1i}$, meaning that\footnote{%
Note that, conditional on the serve being in, a multi-shot rally only occurs
if a one-shot rally fails, giving tennis rallies data, as structured in this
paper, a very natural representation in terms of (conditional, sequential)
binomial distributions.}%
\begin{equation*}
n_{k_{1}i}\sim B\left( n_{x_{1}i}-n_{f_{1}i},k_{1i}\right) .
\end{equation*}
\end{enumerate}

It implies, for instance, that the log-likelihood of observing data $\left(
n_{f_{1}i},n_{x_{1}i}\right) $ given probability $f_{1i}$ is%
\begin{equation*}
l_{f_{1}i}\left( n_{f_{1}i},n_{x_{1}i}|\theta _{i}\right) =\log \binom{%
n_{x_{1}i}}{n_{f_{1}i}}+n_{f_{1}i}\log f_{1i}+\left(
n_{x_{1}i}-n_{f_{1}i}\right) \log \left( 1-f_{1i}\right) .
\end{equation*}

Applying the same logic to all data and collecting the associated terms of
the log-likelihood, obtains%
\begin{eqnarray*}
l_{.i}\left(
N_{i},n_{x_{1}i},n_{x_{2}i},n_{f_{1}i},n_{f_{2}i},n_{k_{1}i},n_{k_{2}i}|%
\theta _{i}\right) &=&l_{x_{1}i}\left( n_{x_{1}i},N_{i}|\theta _{i}\right)
+l_{x_{2}i}\left( n_{x_{2}i},N_{i}-x_{1i}|\theta _{i}\right) \\
&&+l_{f_{1}i}\left( n_{f_{1}i},n_{x_{1}i}|\theta _{i}\right)
+l_{f_{2}i}\left( n_{f_{2}i},n_{x_{2}i}|\theta _{i}\right) \\
&&+l_{k_{1}i}\left( n_{k_{1}i},n_{x_{1}i}-n_{f_{1}i}|\theta _{i}\right)
+l_{k_{2}i}\left( n_{k_{2}i},n_{x_{2}i}-n_{f_{2}i}|\theta _{i}\right) .
\end{eqnarray*}

The first order condition of the maximum likelihood with respect to for
instance $f_{1i}$ requires that%
\begin{eqnarray*}
\frac{\partial l_{.i}\left(
N_{i},n_{x_{1}i},n_{x_{2}i},n_{f_{1}i},n_{f_{2}i},n_{k_{1}i},n_{k_{2}i}|%
\theta _{i}\right) }{\partial f_{1i}} &=&0\Leftrightarrow \\
\frac{n_{f_{1}i}}{f_{1i}}-\frac{n_{x_{1}i}-n_{f_{1}i}}{1-f_{1i}}
&=&0\Leftrightarrow f_{1i}=\frac{n_{f_{1}i}}{n_{x_{1}i}}.
\end{eqnarray*}

Hence, the frequency $\frac{n_{f_{1}i}}{n_{x_{1}i}}$ is the maximum
likelihood estimate of probability $f_{1i}$. A similar argument holds for
all other probabilities. Let $\hat{f}_{1i}=\frac{n_{f_{1}i}}{n_{x_{1}i}}$
denote the empirical of frequency of $f_{1i}$ with a similar notation for
the other terms. Then, by maximum likelihood, the empirical frequencies $%
\left( \hat{x}_{1i},\hat{x}_{2i},\hat{f}_{1i},\hat{f}_{2i},\hat{k}_{1i},\hat{%
k}_{2i}\right) $ are the estimates of $\theta _{i}=\left(
x_{1i},x_{2i},f_{1i},f_{2i},k_{1i},k_{2i}\right) $.

\subsection{Checking the data}

Theorem (\ref{thm:lambda_unique}) applies on data satisfying conditions
(A1)-(A3), which garantees the convergence of Algorithm (\ref{Algorithm}) to
a solution $\lambda $ larger than unity. As indicated in Table (\ref%
{SmallDesc}), for all professional tennis players in our data, the first
serve percentage ($x_{1}$) ranges between $0.51$ and $0.72$ while the second
serve percentage ranges between\footnote{%
In fact, Maxime Cressy and Alexandre Bublik are the only two players with
second serve percentages lower than $84\%$.} $0.73$ and $0.97$. This
trivially shows that condition (A1) is always met. Since $x_{1i}>0.5$, one
has%
\begin{equation*}
\max_{i}\ln \left( \frac{x_{2i}}{x_{1i}}\right) <\ln 2\simeq
0.69<\min_{i}x_{2i}=0.73,
\end{equation*}%
which also garantees that Condition (A2) is satisfied for the ranges of
values for $x_{1}$ and $x_{2}$ observed in the data. It follows that, for
all players in the data, there exists a unique solution $\lambda _{i}^{\ast
} $. Last, note that Condition (A3) rewrites as $x_{2i}\left(
2-x_{2i}\right) <2x_{1i}$ and in the data,%
\begin{equation*}
\max_{i}x_{2i}\left( 2-x_{2i}\right) <1<\min_{i}2x_{1i}=1.02
\end{equation*}%
so that this condition is also satisfied for all players in the data. Hence,
the unique solution $\lambda _{i}^{\ast }$ is strictly greater than unity
for all $i=1,..,N$.

\section{Results}

\begin{table}[H]
\centering
\caption{Parameter estimates with 95\% bootstrap confidence intervals for selected players\label{SmallEsti_CI}}
\resizebox{\textwidth}{!}{%
\begin{tabular}{l|c|c|cc|cc}
\toprule
Player & Salience weight & Curvature & Slope $f$ & Const. $f$ & Slope $k$ & Const. $k$ \\
 & $\delta$ & $\lambda$ & $\tau_f$ & $a_f$ & $\tau_k$ & $a_k$ \\
\midrule
Novak Djokovic & \shortstack[l]{$0.27$\textsuperscript{***} \\ \footnotesize [0.15, 0.41]} & \shortstack[l]{$3.67$\textsuperscript{***} \\ \footnotesize [3.56, 3.81]} & \shortstack[l]{$3.09$\textsuperscript{***} \\ \footnotesize [2.93, 3.26]} & \shortstack[l]{$1.26$\textsuperscript{***} \\ \footnotesize [1.21, 1.30]} & \shortstack[l]{$-19.58$\textsuperscript{***} \\ \footnotesize [-31.23, -14.32]} & \shortstack[l]{$-7.48$\textsuperscript{***} \\ \footnotesize [-12.11, -5.36]} \\
Rafael Nadal & \shortstack[l]{$0.21$\textsuperscript{***} \\ \footnotesize [0.03, 0.47]} & \shortstack[l]{$4.89$\textsuperscript{***} \\ \footnotesize [4.70, 5.10]} & \shortstack[l]{$4.62$\textsuperscript{***} \\ \footnotesize [4.23, 5.07]} & \shortstack[l]{$1.43$\textsuperscript{***} \\ \footnotesize [1.34, 1.54]} & \shortstack[l]{$-105.09$\textsuperscript{} \\ \footnotesize [-1079.64, 953.52]} & \shortstack[l]{$-44.96$\textsuperscript{} \\ \footnotesize [-463.38, 410.43]} \\
Roger Federer & \shortstack[l]{$0.32$\textsuperscript{***} \\ \footnotesize [0.23, 0.42]} & \shortstack[l]{$2.81$\textsuperscript{***} \\ \footnotesize [2.73, 2.90]} & \shortstack[l]{$2.64$\textsuperscript{***} \\ \footnotesize [2.55, 2.74]} & \shortstack[l]{$1.35$\textsuperscript{***} \\ \footnotesize [1.32, 1.39]} & \shortstack[l]{$-16.27$\textsuperscript{***} \\ \footnotesize [-21.06, -13.41]} & \shortstack[l]{$-5.50$\textsuperscript{***} \\ \footnotesize [-7.26, -4.44]} \\
Pete Sampras & \shortstack[l]{$0.42$\textsuperscript{***} \\ \footnotesize [0.31, 0.56]} & \shortstack[l]{$1.93$\textsuperscript{***} \\ \footnotesize [1.82, 2.05]} & \shortstack[l]{$1.49$\textsuperscript{***} \\ \footnotesize [1.43, 1.56]} & \shortstack[l]{$1.13$\textsuperscript{***} \\ \footnotesize [1.10, 1.16]} & \shortstack[l]{$-6.41$\textsuperscript{***} \\ \footnotesize [-7.86, -5.45]} & \shortstack[l]{$-1.38$\textsuperscript{***} \\ \footnotesize [-1.80, -1.10]} \\
Boris Becker & \shortstack[l]{$0.56$\textsuperscript{***} \\ \footnotesize [0.39, 0.79]} & \shortstack[l]{$1.73$\textsuperscript{***} \\ \footnotesize [1.62, 1.85]} & \shortstack[l]{$1.59$\textsuperscript{***} \\ \footnotesize [1.51, 1.68]} & \shortstack[l]{$1.12$\textsuperscript{***} \\ \footnotesize [1.08, 1.15]} & \shortstack[l]{$-13.23$\textsuperscript{***} \\ \footnotesize [-24.26, -9.18]} & \shortstack[l]{$-3.55$\textsuperscript{***} \\ \footnotesize [-6.93, -2.30]} \\
Carlos Alcaraz & \shortstack[l]{$0.05$\textsuperscript{} \\ \footnotesize [-0.12, 0.26]} & \shortstack[l]{$3.66$\textsuperscript{***} \\ \footnotesize [3.45, 3.88]} & \shortstack[l]{$3.56$\textsuperscript{***} \\ \footnotesize [3.22, 3.98]} & \shortstack[l]{$1.34$\textsuperscript{***} \\ \footnotesize [1.24, 1.46]} & \shortstack[l]{$-25.78$\textsuperscript{***} \\ \footnotesize [-115.53, -13.94]} & \shortstack[l]{$-10.11$\textsuperscript{***} \\ \footnotesize [-46.84, -5.27]} \\
Jannik Sinner & \shortstack[l]{$0.06$\textsuperscript{} \\ \footnotesize [-0.04, 0.19]} & \shortstack[l]{$2.52$\textsuperscript{***} \\ \footnotesize [2.39, 2.66]} & \shortstack[l]{$2.61$\textsuperscript{***} \\ \footnotesize [2.46, 2.77]} & \shortstack[l]{$1.28$\textsuperscript{***} \\ \footnotesize [1.24, 1.33]} & \shortstack[l]{$-11.65$\textsuperscript{***} \\ \footnotesize [-16.24, -9.01]} & \shortstack[l]{$-4.08$\textsuperscript{***} \\ \footnotesize [-5.87, -3.04]} \\
Ivo Karlovic & \shortstack[l]{$1.19$\textsuperscript{***} \\ \footnotesize [0.63, 2.36]} & \shortstack[l]{$4.37$\textsuperscript{***} \\ \footnotesize [3.78, 5.11]} & \shortstack[l]{$1.42$\textsuperscript{***} \\ \footnotesize [1.13, 1.74]} & \shortstack[l]{$0.96$\textsuperscript{***} \\ \footnotesize [0.78, 1.15]} & \shortstack[l]{$-7.41$\textsuperscript{***} \\ \footnotesize [-20.52, -4.43]} & \shortstack[l]{$-1.61$\textsuperscript{***} \\ \footnotesize [-4.88, -0.88]} \\
John Isner & \shortstack[l]{$0.79$\textsuperscript{***} \\ \footnotesize [0.52, 1.16]} & \shortstack[l]{$5.33$\textsuperscript{***} \\ \footnotesize [4.89, 5.83]} & \shortstack[l]{$1.78$\textsuperscript{***} \\ \footnotesize [1.59, 2.01]} & \shortstack[l]{$1.11$\textsuperscript{***} \\ \footnotesize [1.00, 1.23]} & \shortstack[l]{$-5.68$\textsuperscript{***} \\ \footnotesize [-7.92, -4.45]} & \shortstack[l]{$-1.19$\textsuperscript{***} \\ \footnotesize [-1.75, -0.88]} \\
Reilly Opelka & \shortstack[l]{$0.36$\textsuperscript{***} \\ \footnotesize [0.14, 0.65]} & \shortstack[l]{$3.89$\textsuperscript{***} \\ \footnotesize [3.47, 4.33]} & \shortstack[l]{$1.81$\textsuperscript{***} \\ \footnotesize [1.60, 2.06]} & \shortstack[l]{$1.19$\textsuperscript{***} \\ \footnotesize [1.07, 1.32]} & \shortstack[l]{$-4.63$\textsuperscript{***} \\ \footnotesize [-6.32, -3.61]} & \shortstack[l]{$-0.84$\textsuperscript{***} \\ \footnotesize [-1.29, -0.58]} \\
David Ferrer & \shortstack[l]{$-0.01$\textsuperscript{} \\ \footnotesize [-0.32, 0.61]} & \shortstack[l]{$2.90$\textsuperscript{***} \\ \footnotesize [2.62, 3.24]} & \shortstack[l]{$4.82$\textsuperscript{***} \\ \footnotesize [4.04, 5.83]} & \shortstack[l]{$1.41$\textsuperscript{***} \\ \footnotesize [1.24, 1.62]} & \shortstack[l]{$40.72$\textsuperscript{} \\ \footnotesize [-431.32, 261.69]} & \shortstack[l]{$16.62$\textsuperscript{} \\ \footnotesize [-172.34, 103.65]} \\
Diego Schwartzman & \shortstack[l]{$-0.38$\textsuperscript{} \\ \footnotesize [-0.75, 0.30]} & \shortstack[l]{$3.57$\textsuperscript{***} \\ \footnotesize [3.15, 4.07]} & \shortstack[l]{$6.45$\textsuperscript{***} \\ \footnotesize [4.70, 9.97]} & \shortstack[l]{$1.60$\textsuperscript{***} \\ \footnotesize [1.26, 2.24]} & \shortstack[l]{$62.62$\textsuperscript{} \\ \footnotesize [-427.40, 425.94]} & \shortstack[l]{$26.49$\textsuperscript{} \\ \footnotesize [-174.09, 177.09]} \\
\hline
Mean & 0.34 & 2.94 & 2.96 & 1.24 & 3.96 & 2.52 \\
Std & 0.70 & 0.98 & 1.36 & 0.19 & 96.83 & 38.20 \\
Min & -3.86 & 1.28 & 1.04 & 0.62 & -352.94 & -129.32 \\
Median & 0.33 & 2.86 & 2.62 & 1.22 & -10.22 & -3.17 \\
Max & 2.79 & 6.27 & 7.99 & 1.69 & 900.62 & 359.60 \\
\bottomrule
\end{tabular}}

\vspace{0.5em}
{\footnotesize Notes: 95\% confidence intervals obtained by bootstrap with 300 replications.$^{***}$ indicates significance at the 5\% level: for $\lambda$, CI excludes 1; for other parameters, CI excludes 0.}
\end{table}

Table (\ref{SmallEsti_CI}) presents the estimates of the parameters for our
selected players as well as the summary statistics for all players (bottom
rows). The mean and median salience weights in our sample are both $1/3$,
indicating a preference for winning multi-shot rallies. In fact, $79\%$
(119/151) of the players have a positive salience weight, and for about $%
64\% $ (76/119) of them, that coefficient is statistically significant at $%
5\%$. In contrast, there are 32 players with negative salience weight, and
only 3 of them for whom that estimate is statistically significant at $5\%$.
The first four listed players in the table are the GOATs. We see that for
all of them, the salience weight is positive and significant. The salience
weight is also positive and significant for big servers like John Isner,
Ralley Opelka, and Ivo Karlovic. Interestingly, the salience weight for the
new top players, Carlos Alcaraz and Jannik Sinner, is positive (0.05) but
not significant. In particular, although the statistics in Table (\ref%
{SmallDesc}) for Jannik Sinner were quite close to those of Roger Federer
and those of Carlos Alcaraz to those of Rafael Nadal, their salience weights
are quite different. This reflects the fact that the identification of the
salience weight shown in Section (\ref{subSection:Identification}) is
non-trivial.

The curvature parameter $\lambda$ is on average 2.9 and ranges from 1.3 to
6.3. It is significantly different from unity for all players in the data,
so that the conditional probability of winning points of all players abbeys
the law of diminishing marginal returns. The slope and constant parameters
of its constituents ($f(x)$ and $k(x)$) are also mostly significantly
different from 0. They can be best interpreted using a graphical
representation, as in Figure (\ref{fig_skills_goats}). This figure shows the
skills parameters of the four GOATs through the plot of $f$ (top panel) and $%
k$ (bottom panel) of these players. The figure clearly indicates that $f$ is
decreasing and concave for these players (true for all players) while $k$ is
increasing and convex (for a few players, i.e., David Ferrer and Diego
Schwartzman, for instance, $k$ is slightly decreasing and concave).
Interestingly, we clearly see that Pete Sampras has the most efficient serve
(measured as the probability to win one-shot rallies) for serve percentages
between 0 and about $80\%$ where Roger Federer's serve becomes more
efficient. However, in terms of winning-point percentages on multi-shot
rallies, Figure (\ref{fig_skills_goats}) shows that Rafael Nadal is the
dominating player at all serve percentages, with a relatively flat profile,
and the profile of Pete Sampras lies at the other extreme.

\begin{figure}[tbph]
\caption{Conditional probabilities of winning one-shot (top) and multi-shot
(bottom) rallies of the four GOATs.}
\label{fig_skills_goats}\centering
\includegraphics[width=0.8%
\textwidth]{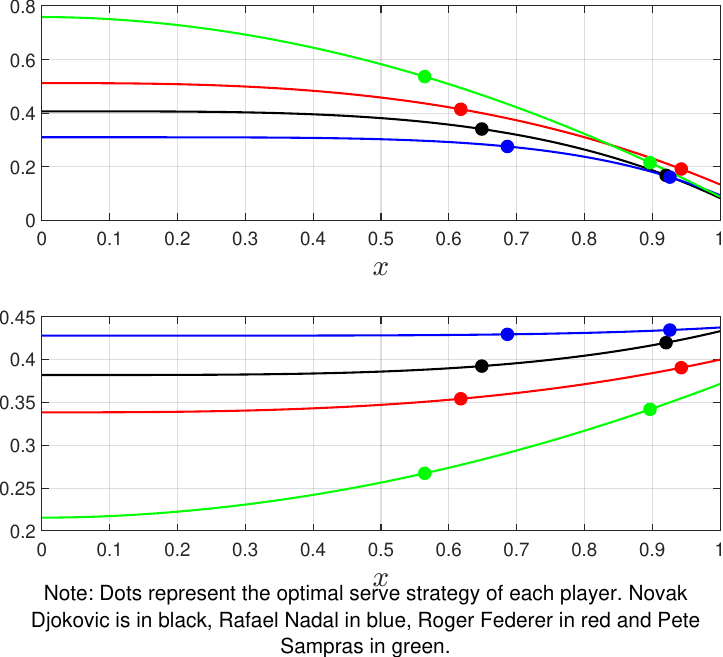}
\end{figure}

Having estimated the salience weight for each player, we can now ask the
natural question: what would have been the optimal service strategy if
players were mere outcome-maximizers, that is, if their salience weight were
zero and they paid no attention to process utility? In other words, this
exercise quantifies how much players are willing to sacrifice in
point-winning probability to enjoy a more appealing style of play.

The model outlined in this paper allows us to answer this question through a
simple counterfactual exercise. Using the estimated parameters for each
player, we set the salience weight to zero and compute the optimal serve
strategy, along with the corresponding probabilities of winning points on
first and second serves. Table (\ref{TabCounter}) presents the differences
in optimal serve strategies between the observed and counterfactual
scenarios ($\delta=0$) in columns~3 and~4. As expected, the change in serve
percentage is negative for almost all players, meaning their strategies
would be more aggressive under the counterfactual, which is especially true
on the second serve. Setting the salience weight to zero increases the
point-winning probability on a player's own serve by approximately
0.39\,\%-pt. While this change appears small, the cumulative nature of
tennis scoring amplifies its impact.

To illustrate, we compute the probability of winning a set and a match
(best-of-five). For this exercise, we assume that the opponent's probability
of winning their own serve equals the player's observed probability under
the estimated $\delta $.\footnote{%
Hence, the player has a $50\%$ chance of winning a set and a match under the
estimated $\delta $.} Columns~7 and~8 show the corresponding changes in set-
and match-winning probabilities. For instance, although a player increases
the probability of winning a point on their serve by only 0.39\thinspace
\%-pt on average, their probability of winning a best-of-five match
increases by 2.42\thinspace \%-pt.

Finally, we can estimate the probability of reaching each round of a Grand
Slam and the expected prize money under both observed and counterfactual
strategies. Using the 2025 US Open prize distribution, Table (\ref%
{TabCounter}) shows that, on average, players forego approximately $\$33{,}%
000$, i.e., $13.5\%$, in expected price money per Grand Slam by optimizing
for both outcome and process utility, with $50\%$ of the players forgoing
more than $\$12{,}880$, i.e., $5.5\%$. When interpreting these results, one
should also bear in mind that there are four Grand Slam tournaments per
year, and while smaller tournaments offer lower prize money, this amount
accumulates across tournaments over a professional career.

\begin{table}
\caption{Counterfactual optimal service strategy when $\delta=0$, by player and aggregate effects.\label{TabCounter}}
\footnotesize
\resizebox{\textwidth}{!}{ 
\begin{tabular}{l|c|cc|cccc|c}
\toprule
Player & Salience weight & $\Delta \%-st$ & $\Delta \%-nd$ & $\Delta \%-pt$ & $\Delta \%-gm$ & $\Delta \%-set$ & $\Delta \%-mat$ & $\Delta$ Prize $(\times\$1000)$ \\
 & $\delta$  & in        &   in      & won       & won         & won        & won        &  won   \\
\hline
Novak Djokovic & 0.27 & -1.06 & -4.31 & 0.12 & 0.18 & 0.39 & 0.73 & 8.07  \\
Rafael Nadal & 0.21 & -0.63 & -2.36 & 0.04 & 0.06 & 0.13 & 0.24 & 2.56  \\
Roger Federer & 0.32 & -1.07 & -6.01 & 0.20 & 0.27 & 0.64 & 1.20 & 13.51  \\
Pete Sampras & 0.42 & -2.25 & -9.94 & 0.57 & 0.82 & 1.82 & 3.41 & 42.01  \\
... &  &  &  &  &  &  &  &  \\
Boris Becker & 0.56 & -2.53 & -10.79 & 0.64 & 1.14 & 2.13 & 3.99 & 50.26  \\
... &  &  &  &  &  &  &  &  \\
Carlos Alcaraz & 0.05 & -0.21 & -0.89 & 0.00 & 0.01 & 0.02 & 0.03 & 0.31  \\
Jannik Sinner & 0.06 & -0.23 & -1.61 & 0.01 & 0.02 & 0.04 & 0.07 & 0.80  \\
... &  &  &  &  &  &  &  &   \\
Ivo Karlovic & 1.19 & -4.69 & -10.79 & 1.53 & 1.61 & 4.64 & 8.66 & 132.87  \\
John Isner & 0.79 & -3.87 & -10.18 & 1.20 & 1.33 & 3.67 & 6.86 & 97.63  \\
Reilly Opelka & 0.36 & -2.24 & -7.71 & 0.48 & 0.57 & 1.48 & 2.78 & 33.36  \\
... &  &  &  &  &  &  &  &    \\
David Ferrer & -0.01 & 0.04 & 0.21 & 0.00 & 0.00 & 0.00 & 0.00 & 0.01  \\
Diego Schwartzman & -0.38 & 1.14 & 6.25 & 0.17 & 0.38 & 0.60 & 1.12 & 12.56  \\
\hline
Mean & 0.34 & -1.46 & -4.94 & 0.39 & 0.66 & 1.29 & 2.41 & 33.22  \\
Std & 0.70 & 1.52 & 5.48 & 0.47 & 0.77 & 1.53 & 2.85 & 46.76  \\
Min & -3.86 & -6.33 & -16.26 & 0.00 & 0.00 & 0.00 & 0.00 & 0.00  \\
Median & 0.33 & -1.17 & -5.54 & 0.18 & 0.33 & 0.61 & 1.15 & 12.88  \\
Max & 2.79 & 1.40 & 11.63 & 2.44 & 4.61 & 8.15 & 15.02 & 304.06  \\
\bottomrule
\end{tabular}
}
\begin{minipage}{\textwidth}
\footnotesize
Notes: $\Delta$ means the difference in the variable concerned under the counterfactual $\delta=0$ and the observed situation. To derive $\Delta$ Prize, we compute the probability that the player reaches each of round of the US Open, first, when he has 50$\%$ chances of winning each match (best-of-five), and then when his chances are 50+$\Delta \%-mat$.
 We use the prize money distribution of the 2025 US open for this computation.\end{minipage}
\end{table}


\section{Robustness checks}

Our parametric results rely crucially on Condition (\ref%
{Assumption:Elasticity}), which has two components: first, the elasticity of
the marginal probability of winning one-shot points equals that of winning
multi-shot points, and second, this elasticity is constant. We conduct two
robustness checks by relaxing these assumptions one at a time.

In the first robustness check, we allow the elasticity of the marginal
conditional probabilities to vary linearly with $x$, i.e., proportional to $%
x\lambda$ with $\lambda>0$ (Condition (\ref{Assumption:Elasticity2}), Online
Appendix (\ref{Section:Softmax})) rather than being constant. Under this
specification of the model, we find that the estimated salience weights are
slightly larger than in the baseline, indicating that our earlier
computations of trade-offs are robust and, if anything, conservative (see
Table (\ref{SmallEsti_softmax_CI}), Online Appendix (\ref{Section:Softmax})).

The second robustness check relaxes the equality of elasticities across
one-shot and multi-shot points. Specifically, we let the former be $\lambda
-1$ and the latter $t\lambda -1$ with $t>0$ (Condition (\ref%
{Assumption:Elasticity3}), Online Appendix (\ref{Section:DiffCurv})). Note
that, for most players, the conditional probabilities of winning multi-shot
rallies on first and second serve, $k(x_1)$ and $k(x_2)$, are closer to each
other than for one-shot points, so that $k(x)$ is likely to be ``flatter''
than $f(x)$, which would obtain for $t<1$. Nevertheless, we evaluate the
model over a grid of $t $ values from 0.5 to 2, so that the curvature of $k$
ranges from half to twice that of $f$. This grid approach is necessary
because the data do not provide sufficient degrees of freedom to estimate $t$
jointly with the other parameters. Across this range, the estimated salience
weights are remarkably stable as shown in Table (\ref{SmallEsti_t}) of
Online Appendix (\ref{Section:DiffCurv}).

We conclude from these results that the estimates of salience weights
presented above are very robust to departure from Condition (\ref%
{Assumption:Elasticity}).

To further challenge our findings, we consider two alternative explanations
for the apparent preference for winning multi-shot rallies and resulting
conservative serve strategies adopted by tennis players. First, one could
argue that players value one-shot and multi-shot rallies differently because
these rallies require different levels of effort. One-shot rallies by nature
are less demanding and more energy-efficient, so if effort considerations
were driving deviations from outcome-maximization, players would be expected
to put more importance on winning one-shot rallies, conserving energy for
later points on their opponents' serve. We actually find the opposite:
players put more weight on winning multi-shot rallies on their own serve,
suggesting that enjoyment, rather than energy considerations, drives these
choices.

Second, one might argue that conservative second-serve strategies arise due
to risk aversion and, in particular, aversion to double faults, rather than
process utility. To examine this, we develop and estimate a model (in Online
Appendix~(\ref{Section:AversionDF})) where players have a disutility for
double faults, rather than a utility for process, and derive the associated
optimal serve strategies. Importantly, this alternative model does not rely
on the distinction between the conditional probabilities of winning one-shot
and multi-shot rallies, only on the sum of the two. Therefore, if
double-fault aversion was the true mechanism, the estimated disutility
parameter should be unrelated to the probability of winning one-shot rallies
conditional on serve strategy and point-winning probabilities. However, our
analysis shows that $\mathbb{E}\left[ \gamma \times
f_{j}|x_{1},x_{2},y_{1},y_{2}\right] $ for $j=1,2$, is systematically
positive and significant at $5\%$. This indicates that there is additional
information in the conditional probability of winning one-shot rallies that
is being forced into the parameter of double-fault aversion. This, in turns,
supports the relevance of process utility in explaining strategic deviations
from outcome-maximizing behavior.

\section{Conclusion and discussion}

We examine how individuals trade off outcome (\textquotedblleft
what\textquotedblright ) and process (\textquotedblleft
how\textquotedblright ) utility in high-stakes strategic decisions, namely
in professional tennis. We first develop a nonparametric identification
strategy based solely on optimality conditions and the second-serve rule
that delivers a sufficient condition under which the nonparametric bounds on
process utility are positive. Under mild shape restrictions, we show that
this sufficient condition is likely met for a large majority of players. We
then propose a parametric approach to estimate each player's salience weight
and to conduct counterfactual analyses. We show that professional tennis
players in our sample are willing to trade off a lower point-winning
probability on their serve (outcome utility), on average $0.4\%$-pt, in
exchange for a higher probability of winning multi-shot rallies (process
utility). Although these differences are small in probability terms, due to
the rules of tennis, they translate into a substantial forgone expected
prize money. Using the 2025 US Open prize distribution as an example, we
find that, on average, players forgo $\$33,000$, i.e., $13.5\%$, in prize
money per tournament (Grand Slam), with $50\%$ of the players sacrificing $%
\$13,000$ or more, i.e., $5.5\%$, per Grand Slam. These results demonstrate
that process utility---i.e., the enjoyment of winning multi-shot
rallies---is a significant component of tennis players' utility. That
players are willing to sacrifice substantial expected prize money aligns
well with the psychological literature on flow and intrinsic motivation.
Indeed, according to \cite{csik90}, when individuals enter a state of flow,
they become fully absorbed in the activity itself, such that the experience
is intrinsically rewarding even at a great material cost. This
interpretation is further consistent with Self-Determination Theory, which
emphasizes that intrinsic motivation is fostered when individuals experience
autonomy and competence in the activity itself (e.g., \cite{RyanDeci24}). In
this light, multi-shot rallies may provide a richer environment for such
experiences than one-shot outcomes, helping to rationalize the observed
willingness to trade off performance for process enjoyment.

As discussed in the introduction, the tennis setting offers several
advantages for identifying and estimating individual-specific salience
weights while capturing a mechanism likely to be present in many other
economic contexts, particularly in labor markets. A growing body of evidence
on compensating differentials suggests that individuals are willing to make
substantial monetary sacrifices to access intrinsically rewarding job
attributes. Existing evidence, however, is typically based on either
revealed preferences over discrete job choices or on experimentally elicited
stated or incentivized choices. For instance, \cite{Stern04} uses multiple
job offers received by PhD students to measure the preference for
independent research and shows that scientists are willing to forgo
approximately 19$\%$ of their wages to engage in independent research, while 
\cite{Mas17} estimates willingness to pay for non-wage job amenities using
incentivized discrete-choice experiments over alternative work arrangements,
finding wage trade-offs of approximately 8$\%$ to 20$\%$ for working from
home and having regular time schedules respectively. More broadly, similar
trade-offs between material outcomes and process-related attributes are also
well documented in consumer markets, suggesting that the type of preferences
identified in this paper may represent a general feature of economic
behavior across domains. Consistent with these findings, this paper shows
that such trade-offs can be identified from revealed repeated,
high-frequency, continuous choices in a competitive field setting, thereby
providing a novel approach to measuring process utility from observed
behavior.

Our results are also important for policy interventions. For instance, in
the tennis context, one might be tempted to conclude that, from a coaching
perspective, a possible policy intervention would be to encourage players to
set aside their desire for enjoyment and adopt strategies more closely
aligned with outcome maximization. However, while such an approach may be
effective in certain points of a match, the long-run implications of
systematically neglecting process utility may be detrimental to performance.
As suggested by flow and self-determination theories, suppressing the need
for enjoyment and undermining autonomy and competence during play can lead
to boredom, anxiety, stress, or ``controled motivation'' which in turn may
reduce performance and contribute to adverse long-term outcomes, including
burnout, disengagement and withdrawal (e.g., \cite{RyanDeci24}).

\newpage


\bibliographystyle{ecta}
\bibliography{biblio_Pref}

\newpage

\appendix

\setcounter{section}{0} \renewcommand{\thesection}{\Alph{section}}


\section{Appendix}

\renewcommand{\thesubsection}{\thesection\arabic{subsection}}

\subsection{Derivatives of the map $\Lambda \left( \protect\lambda \right) $%
\label{Derivatives:Map}}

\begin{lemma}
\label{Lemma:IncrConcMap}The map $\Lambda \left( \lambda \right) $ defined
by Algorithm (\ref{Algorithm}) is strictly increasing and strictly convex on 
$%
\mathbb{R}
_{0}^{+}$ for all $1>x_{2i}>x_{1i}>0$.
\end{lemma}

\begin{proof}
Recall that%
\begin{equation*}
\Lambda \left( \lambda \right) =\left( \frac{x_{2i}}{x_{1i}}\right)
^{\lambda }\left( 1+\lambda \left( 1-x_{2i}\right) \right) -1
\end{equation*}%
and $\Lambda \left( \lambda \right) >0$ on $%
\mathbb{R}
_{0}^{+}$, for all $1>x_{2i}>x_{1i}>0$.

One has 
\begin{eqnarray*}
\Lambda ^{\prime }\left( \lambda \right) &=&\frac{\partial \left( \left( 
\frac{x_{2i}}{x_{1i}}\right) ^{\lambda }\left( 1+\lambda \left(
1-x_{2i}\right) \right) -1\right) }{\partial \lambda } \\
&=&\left( \frac{x_{2i}}{x_{1i}}\right) ^{\lambda }\left( \ln \frac{x_{2i}}{%
x_{1i}}\left( 1+\lambda \left( 1-x_{2i}\right) \right) +1-x_{2i}\right) \\
&=&\ln \frac{x_{2i}}{x_{1i}}\left( \frac{x_{2i}}{x_{1i}}\right) ^{\lambda
}\left( 1+\lambda \left( 1-x_{2i}\right) \right) +\left( \frac{x_{2i}}{x_{1i}%
}\right) ^{\lambda }\left( 1-x_{2i}\right) \\
&=&\ln \frac{x_{2i}}{x_{1i}}\left[ \left( \frac{x_{2i}}{x_{1i}}\right)
^{\lambda }\left( 1+\lambda \left( 1-x_{2i}\right) \right) -1+1\right]
+\left( \frac{x_{2i}}{x_{1i}}\right) ^{\lambda }\left( 1-x_{2i}\right) \\
&=&\left( \Lambda \left( \lambda \right) +1\right) \ln \frac{x_{2i}}{x_{1i}}%
+\left( 1-x_{2i}\right) \left( \frac{x_{2i}}{x_{1i}}\right) ^{\lambda }
\end{eqnarray*}

Since $\Lambda \left( \lambda \right) >0$ on $%
\mathbb{R}
_{0}^{+}$ and $1\geq x_{2i}>x_{1i}$, we conclude that $\Lambda ^{\prime
}\left( \lambda \right) >0$ on $%
\mathbb{R}
_{0}^{+}$. The second derivative is%
\begin{eqnarray*}
\Lambda ^{\prime \prime }\left( \lambda \right)  &=&\Lambda ^{\prime }\left(
\lambda \right) \ln \frac{x_{2i}}{x_{1i}}+\left( 1-x_{2i}\right) \left( 
\frac{x_{2i}}{x_{1i}}\right) ^{\lambda }\ln \frac{x_{2i}}{x_{1i}} \\
&=&\left( \Lambda ^{\prime }\left( \lambda \right) +\left( 1-x_{2i}\right)
\left( \frac{x_{2i}}{x_{1i}}\right) ^{\lambda }\right) \ln \frac{x_{2i}}{%
x_{1i}}.
\end{eqnarray*}

It follows that since $\Lambda ^{\prime }\left( \lambda \right) >0$ on $%
\mathbb{R}
_{0}^{+}$ and $x_{2i}>x_{1i}$, one has $\Lambda ^{\prime \prime }\left(
\lambda \right) >0$ on $%
\mathbb{R}
_{0}^{+}$. Hence the map $\Lambda \left( \lambda \right) $ is strictly
increasing and strictly convex on $%
\mathbb{R}
_{0}^{+}$.
\end{proof}


\subsection{Proof Lemma (\protect\ref{LemmaNonPara}).\label{App:Lemma}}

\begin{proof}
From Condition (\ref{Assumption:Theory}), since $f\left( x\right) $ and $%
k\left( x\right) $ are continuous, so are $B\left( x\right) $, $C\left(
x\right) $ and hence $A\left( x\right) $. From (b) and (c), an application
of the Intermediate Value Theorem, indicates that there exists a $x_{0}\in %
\left[ x_{1}^{\ast },x_{2}^{\ast }\right] $ so that $A\left( x_{0}\right) =0$%
. Under the assumption that $y\left( x\right) $ is strictly decreasing and
concave, i.e., $y^{\prime }\left( x\right) <0$ and $y^{\prime \prime }\left(
x\right) \leq 0$, $A\left( x\right) $ must also be strictly concave, i.e., $%
A^{\prime \prime }\left( x\right) =2y^{\prime }\left( x\right) +xy^{\prime
\prime }\left( x\right) <0$, so that $x_{0}$ is unique.
\end{proof}

\subsection{Proof Lemma (\protect\ref{LemmaNonPara2}).}

\begin{proof}
Let $x\in \left[ x_{1}^{\ast },x_{2}^{\ast }\right] $ and $m_{k}^{\ast
}=m\left( x_{k}^{\ast }\right) $ for $k=1,2$, and $m\left( x\right)
=xf\left( x\right) $. Note that $m\left( x\right) $ is strictly concave as $%
m^{\prime \prime }\left( x\right) =2f^{\prime }\left( x\right) +xf^{\prime
\prime }\left( x\right) <0$ since $f^{\prime }\left( x\right) <0$ and $%
f^{\prime \prime }\left( x\right) <0$ under Condition (\ref%
{Assumption:Theory}). We consider the sufficient condition represented in
the following inequality%
\begin{equation}
m\left( x\right) <m_{2}^{\ast }x+b_{12}\text{, }\forall x\in \left[
x_{1}^{\ast },x_{2}^{\ast }\right] ,  \label{eqSuffCondStrict}
\end{equation}%
where $b_{12}=\left( 1-x_{1}^{\ast }\right) m_{1}^{\ast }+m_{1}^{\ast
}-m_{2}^{\ast }$.

Under condition (b), one can deduce from concavity that $m\left(
x_{0}\right) >m_{2}^{\ast }$ and also verify that the sufficient condition
is met at the boundaries, since for $x=x_{1}^{\ast }$, one has%
\begin{equation*}
m_{1}^{\ast }<m_{2}^{\ast }x_{1}^{\ast }+b_{12}\Leftrightarrow m_{1}^{\ast
}+\left( 1-x_{1}^{\ast }\right) m_{2}^{\ast }<m_{1}^{\ast }+\left(
1-x_{1}^{\ast }\right) m_{1}^{\ast }
\end{equation*}%
which obtains from $\left( 1-x_{1}^{\ast }\right) m_{2}^{\ast }<\left(
1-x_{1}^{\ast }\right) m_{1}^{\ast }$ and $m_{2}^{\ast }<m_{1}^{\ast }$, and
for $x=x_{2}^{\ast }$, one has%
\begin{equation*}
m_{2}^{\ast }<m_{1}^{\ast }<m_{2}^{\ast }x_{1}^{\ast }+b_{12}<m_{2}^{\ast
}x_{2}^{\ast }+b_{12}
\end{equation*}%
which follows from $x_{2}^{\ast }>x_{1}^{\ast }$.

Because of the concavity of $m\left( x\right) $ and since $m\left( 0\right)
=0$, it must be that the graph of $m\left( x\right) $ lies above that of $%
l_{1}\left( x\right) :=\frac{m_{1}}{x_{1}}x$ (i.e., the line passing through 
$\left( 0,0\right) $ and $\left( x_{1},m_{1}\right) $) on $x\in \left[
0,x_{1}^{\ast }\right] $ and below $l_{1}\left( x\right) $ on $x\in \left[
x_{1}^{\ast },1\right] $. Similarly, because of the concavity of $m\left(
x\right) $ and since $m\left( 1\right) =f\left( x\right) \geq 0$, the graph
of $m\left( x\right) $ lies below that of $l_{2}\left( x\right) :=\frac{m_{2}%
}{1-x_{2}}\left( 1-x\right) $ (the line passing through $\left(
x_{2},m_{2}\right) $ and $\left( 1,0\right) $) on $x\in \left[ x_{2}^{\ast
},1\right] $ and above it on $x\in \left[ 0,x_{2}^{\ast }\right] $. Finally,
still because $m\left( x\right) $ is concave, its graph lies above that of $%
l_{3}\left( x\right) :=m_{1}-\frac{m_{2}-m_{1}}{x_{2}-x_{1}}x_{1}+\frac{%
m_{2}-m_{1}}{x_{2}-x_{1}}x$ (the line passing through $\left(
x_{1},m_{1}\right) $ and $\left( x_{2},m_{2}\right) $) on $x\in \left[
x_{1}^{\ast },x_{2}^{\ast }\right] $. This means that the following
inequalities are satisfied%
\begin{equation*}
l_{3}\left( x\right) <m\left( x\right) <\max \left\{ l_{1}\left( x\right)
,l_{2}\left( x\right) \right\} \text{ }\forall x_{1}^{\ast }<x<x_{2}^{\ast }.
\end{equation*}

The lines $l_{1}\left( x\right) $, $l_{2}\left( x\right) $ and $l_{3}\left(
x\right) $ form a triangle (convex polyhedron) within which the graph of the
concave function $m\left( x\right) $ must lie on $\left[ x_{1}^{\ast
},x_{2}^{\ast }\right] $. Moreover, by definition, $l_{1}\left( x\right) $
and $l_{2}\left( x\right) $ cross each other at $x_{12}=\frac{m_{2}^{\ast
}x_{1}^{\ast }}{\left( 1-x_{2}^{\ast }\right) m_{1}^{\ast }+m_{2}^{\ast
}x_{1}^{\ast }}\in \left[ x_{1}^{\ast },x_{2}^{\ast }\right] $ with $%
l_{1}\left( x_{12}\right) =\frac{m_{1}^{\ast }m_{2}^{\ast }}{\left(
1-x_{2}^{\ast }\right) m_{1}^{\ast }+m_{2}^{\ast }x_{1}^{\ast }}>l_{1}\left(
x_{1}^{\ast }\right) =m_{1}^{\ast }>m_{2}^{\ast }$, where one notes that the
index "$12$" indexes the x-coordintate of the point at which line $%
l_{1}\left( x\right) $ and line $l_{2}\left( x\right) $ cross each other.
This triangle, therefore, contains all feasible paths of the graph of $%
m\left( x\right) $ on $\left[ x_{1}^{\ast },x_{2}^{\ast }\right] $. Its area
can be computed using either integrals, as%
\begin{equation*}
A_{1}=\int_{x_{1}^{\ast }}^{x_{12}}\left( l_{1}\left( x\right) -l_{3}\left(
x\right) \right) dx+\int_{x_{12}}^{x_{2}^{\ast }}\left( l_{2}\left( x\right)
-l_{3}\left( x\right) \right) dx,
\end{equation*}%
or geometry, and obtains as%
\begin{equation*}
A_{1}=\frac{1}{2}\left\vert \frac{\left( m_{1}^{\ast }x_{2}^{\ast
}-m_{2}^{\ast }x_{1}^{\ast }\right) \left( \left( 1-x_{1}^{\ast }\right)
m_{2}^{\ast }-\left( 1-x_{2}^{\ast }\right) m_{1}^{\ast }\right) }{\left(
1-x_{2}^{\ast }\right) m_{1}^{\ast }+m_{2}^{\ast }x_{1}^{\ast }}\right\vert .
\end{equation*}

Let the right hand side of Inequality (\ref{eqSuffCondStrict}) be the line $%
l_{4}\left( x\right) :=b_{12}+m_{2}^{\ast }x$, of slope $m_{2}^{\ast }$ and
constant at origin $b_{12}>0$. Hence, the associated sufficient condition
requires that all points $\left( x,m\left( x\right) \right) $ on the graph
of $m\left( x\right) $ lie below the point $\left( x,m_{2}^{\ast
}x+b_{12}\right) $ for all $x\in \left[ x_{1}^{\ast },x_{2}^{\ast }\right] $%
, implying it is also true for $x=x_{0}$. We already know that $l_{4}\left(
x_{1}^{\ast }\right) =m_{2}^{\ast }x_{1}^{\ast }+b_{12}>m_{1}^{\ast }$.
Hence, either this line crosses the triangle defined above at say $\left(
x_{14},l_{4}\left( x_{14}\right) \right) $ and $\left( x_{24},l_{4}\left(
x_{24}\right) \right) $ with $x_{1}^{\ast }<x_{14}<x_{24}<x_{2}^{\ast }$ and 
$m_{1}^{\ast }<l_{4}\left( x_{14}\right) <l_{4}\left( x_{24}\right) $,
implying a restriction on the paths that $m\left( x\right) $ can take and
still satisfy the inequality, or it does not, which guarantees all paths of $%
m\left( x\right) $ satisfy the inequality.

To check this, one can easily obtain $x_{14}=b_{12}\frac{x_{1}^{\ast }}{%
m_{1}^{\ast }-m_{2}^{\ast }x_{1}^{\ast }}$ and $x_{24}=\frac{m_{2}^{\ast
}-\left( 1-x_{2}^{\ast }\right) b_{12}}{m_{2}^{\ast }\left( 2-x_{2}^{\ast
}\right) }$ as the respective solutions to $l_{1}\left( x\right)
=l_{4}\left( x\right) $ and $l_{2}\left( x\right) =l_{4}\left( x\right) $,
and then compute the area of the triangle above this constraint, say $A_{2}$%
. To do so, let $l_{4}\left( x\right) =m_{2}^{\ast }x+b_{12}$ and compute%
\begin{equation*}
A_{2}=\int_{x_{14}}^{x_{12}}\left( l_{1}\left( x\right) -l_{4}\left(
x\right) \right) dx+\int_{x_{12}}^{x_{24}}\left( l_{2}\left( x\right)
-l_{4}\left( x\right) \right) dx,
\end{equation*}%
or geometrically as%
\begin{equation*}
A_{2}=\frac{1}{2}\left\vert \left( x_{12}-x_{14}\right) \left( \frac{%
x_{2}^{\ast }f_{2}^{\ast }}{1-x_{2}^{\ast }}\left( 1-x_{24}\right)
-f_{1}^{\ast }x_{24}\right) \right\vert 
\end{equation*}%
when $x_{14}<x_{24}$ and $A_{2}=0$ otherwise. One can then check the share $%
A_{2}/A_{1}$ representing the extent to which this sufficient condition
restricts the feasible paths the graph of $m\left( x\right) $ can take on $%
\left[ x_{1}^{\ast },x_{2}^{\ast }\right] $ and still satisfy the condition.
\end{proof}


\subsection{Proof Theorem (\protect\ref{thm:lambda_unique}).\label{App:Proof}%
}

\begin{proof}
The proof proceeds in 6 steps.

\begin{enumerate}
\item \textit{Step 1: Definition of the map }$\Lambda _{i}$ and its
properties.

The map $\Lambda _{i}:%
\mathbb{R}
_{0}^{+}\rightarrow \mathbb{R}_{0}^{+}$, defined as%
\begin{equation*}
\Lambda _{i}\left( \lambda \right) =\left( \frac{x_{2i}}{x_{1i}}\right)
^{\lambda }\left( 1+\lambda \left( 1-x_{2i}\right) \right) -1,
\end{equation*}%
is strictly increasing and convex on $%
\mathbb{R}
_{0}^{+}$ from Lemma (\ref{Lemma:IncrConcMap}) in Appendix (\ref%
{Derivatives:Map}).

Its first and second derivatives read as%
\begin{eqnarray*}
\Lambda _{i}^{\prime }\left( \lambda \right) &=&\left( \Lambda _{i}\left(
\lambda \right) +1\right) \ln \frac{x_{2i}}{x_{1i}}+\left( 1-x_{2i}\right)
\left( \frac{x_{2i}}{x_{1i}}\right) ^{\lambda }>0, \\
\Lambda ^{\prime \prime }\left( \lambda \right) &=&\left( \Lambda ^{\prime
}\left( \lambda \right) +\left( 1-x_{2i}\right) \left( \frac{x_{2i}}{x_{1i}}%
\right) ^{\lambda }\right) \ln \frac{x_{2i}}{x_{1i}}>0.
\end{eqnarray*}%
Define the function $g_{i}\left( \lambda \right) $ as%
\begin{equation*}
g_{i}\left( \lambda \right) =\Lambda _{i}(\lambda )-\lambda =\left( \frac{%
x_{2i}}{x_{1i}}\right) ^{\lambda }\left( 1+\lambda \left( 1-x_{2i}\right)
\right) -1-\lambda ,\qquad \lambda >0.
\end{equation*}

\item \textit{Step 2: Continuity of }$g_{i}$\textit{.}

Since $\Lambda _{i}$ is continuously twice differentiable on $%
\mathbb{R}
_{0}^{+}$ so is $g_{i}\left( \lambda \right) $.

\item \textit{Step 3: Boundary behavior.}

For $\lambda \rightarrow 0$, $\Lambda _{i}(\lambda )\rightarrow 0$, since we
have 
\begin{eqnarray*}
\lim_{\lambda \rightarrow 0^{+}}g_{i}\left( \lambda \right) 
&=&\lim_{\lambda \rightarrow 0^{+}}\Lambda _{i}(\lambda )=\lim_{\lambda
\rightarrow 0^{+}}\left( \frac{x_{2i}}{x_{1i}}\right) ^{\lambda }\left(
1+\lambda \left( 1-x_{2i}\right) \right) -1 \\
&=&0.
\end{eqnarray*}%
In contrast, as $\lambda \rightarrow \infty $, $g_{i}(\lambda )\rightarrow
+\infty $ since%
\begin{eqnarray*}
\lim_{\lambda \rightarrow \infty }g_{i}(\lambda ) &=&\lim_{\lambda
\rightarrow \infty }\left[ \left( \frac{x_{2i}}{x_{1i}}\right) ^{\lambda
}\left( 1+\lambda \left( 1-x_{2i}\right) \right) -1-\lambda \right]  \\
&=&\left( 1-x_{2i}\right) \lim_{\lambda \rightarrow \infty }\lambda \left( 
\frac{x_{2i}}{x_{1i}}\right) ^{\lambda }=+\infty ,
\end{eqnarray*}%
using (A1).

\item \textit{Step 4: First and second derivatives of }$g_{i}\left( \lambda
\right) $\textit{.}

By definition 
\begin{eqnarray*}
g_{i}^{\prime }\left( \lambda \right)  &=&\Lambda _{i}^{\prime }\left(
\lambda \right) -1 \\
&=&\left( \Lambda _{i}\left( \lambda \right) +1\right) \ln \frac{x_{2i}}{%
x_{1i}}+\left( 1-x_{2i}\right) \left( \frac{x_{2i}}{x_{1i}}\right) ^{\lambda
}-1.
\end{eqnarray*}%
Hence,%
\begin{eqnarray*}
\lim_{\lambda \rightarrow 0^{+}}g_{i}^{\prime }\left( \lambda \right) 
&=&\lim_{\lambda \rightarrow 0^{+}}\left[ \left( \Lambda _{i}\left( \lambda
\right) +1\right) \ln \frac{x_{2i}}{x_{1i}}+\left( 1-x_{2i}\right) \left( 
\frac{x_{2i}}{x_{1i}}\right) ^{\lambda }\right] -1 \\
&=&\left[ \left( \lim_{\lambda \rightarrow 0^{+}}\Lambda _{i}\left( \lambda
\right) +1\right) \ln \frac{x_{2i}}{x_{1i}}+\left( 1-x_{2i}\right)
\lim_{\lambda \rightarrow 0^{+}}\left( \frac{x_{2i}}{x_{1i}}\right)
^{\lambda }\right] -1 \\
&=&\ln \frac{x_{2i}}{x_{1i}}-x_{2i}
\end{eqnarray*}%
This means that for $\ln \frac{x_{2i}}{x_{1i}}<x_{2i}$, $\lim_{\lambda
\rightarrow 0^{+}}g_{i}^{\prime }\left( \lambda \right) <0$.

One also has%
\begin{eqnarray*}
\lim_{\lambda \rightarrow +\infty }g^{\prime }\left( \lambda \right) 
&=&\lim_{\lambda \rightarrow +\infty }\left[ \left( \Lambda _{i}\left(
\lambda \right) +1\right) \ln \frac{x_{2i}}{x_{1i}}+\left( 1-x_{2i}\right)
\left( \frac{x_{2i}}{x_{1i}}\right) ^{\lambda }\right]  \\
&=&\left( \lim_{\lambda \rightarrow +\infty }\Lambda _{i}\left( \lambda
\right) +1\right) \ln \frac{x_{2i}}{x_{1i}}+\left( 1-x_{2i}\right)
\lim_{\lambda \rightarrow +\infty }\left( \frac{x_{2i}}{x_{1i}}\right)
^{\lambda } \\
&=&\ln \frac{x_{2i}}{x_{1i}}\lim_{\lambda \rightarrow +\infty }\Lambda
_{i}\left( \lambda \right) +\left( 1-x_{2i}\right) \lim_{\lambda \rightarrow
+\infty }\left( \frac{x_{2i}}{x_{1i}}\right) ^{\lambda }=+\infty .
\end{eqnarray*}%
Finally, note that%
\begin{equation*}
g_{i}^{\prime \prime }\left( \lambda \right) =\Lambda _{i}^{\prime \prime
}\left( \lambda \right) >0.
\end{equation*}

\item \textit{Step 5: }$g_{i}\left( 1\right) $\textit{.}

Note that%
\begin{equation*}
g_{i}\left( 1\right) =\left( \frac{x_{2i}}{x_{1i}}\right) \left( 1+\left(
1-x_{2i}\right) \right) -2.
\end{equation*}%
Hence, $g_{i}\left( 1\right) <0$ is implied by (A3) as indeed%
\begin{eqnarray*}
g_{i}\left( 1\right)  &<&0\Leftrightarrow x_{2i}\left( 2-x_{2i}\right)
<2x_{1i} \\
&\Leftrightarrow &2\left( x_{2i}-x_{1i}\right) <\left( x_{2i}\right) ^{2}.
\end{eqnarray*}

\item \textit{Step 6: conclusion.}

The function $g_{i}\left( \lambda \right) $ tends to 0 when $\lambda
\rightarrow 0$, is first decreasing for values of $\lambda $ close to $0$
under (A2), i.e., $\ln \frac{x_{2i}}{x_{1i}}<x_{2i}$, hence becomes negative
as $\lambda $ increases away from $0$, but tends to $+\infty \,\ $as $%
\lambda \rightarrow \infty $. Since $g_{i}$ is continuous, this means that $%
g_{i}\left( \lambda \right) $ must have at least one root on $\mathbb{R}%
_{0}^{+}$. However, since it is strictly convex on $\mathbb{R}_{0}^{+}$,
there can be at most one root.

So, under (A1)-(A2), for all $i$, the associated map $\Lambda _{i}(\lambda )$
has exactly one solution $\lambda _{i}^{\ast }=\Lambda _{i}(\lambda
_{i}^{\ast })$. Moreover, under (A3), this solution $\lambda _{i}^{\ast }$
is strictly greater than unity.
\end{enumerate}
\end{proof}

\newpage

\section{Online Appendix}

\numberwithin{table}{section} 
\numberwithin{figure}{section} 
\numberwithin{equation}{section} 
\renewcommand{\thecondition}{\thesection.\arabic{condition}}

\subsection{Players' Testimonies\label{App:Testimonies}}

Professional tennis players frequently emphasize enjoyment, intrinsic
motivation, and the pursuit of flow as key drivers of their performance. The
following quotes illustrate that top athletes prioritize the process of play
alongside outcome maximization:

\begin{quote}
\textquotedblleft \emph{I just want to step on court, try to do my things,
follow my goals, and try to enjoy as much as I can.}\textquotedblright\ ---
Carlos Alcaraz, before his 2022 US Open semifinal vs. Novak Djokovic. 
\newline
https://www.atptour.com/en/news/alcaraz-lehecka-us-open-2025-qf

\textquotedblleft \emph{I love the winning, I can take the losing, but most
of all I love to play.} \textquotedblright\ --- Boris Becker
https://www.allgreatquotes.com/quote-56608

\textquotedblleft \emph{My personal goal is to have fun and enjoy the
moment, not put too much pressure on myself.}\textquotedblright\ --- Venus
Williams, interview by Rory Carroll, 2025. \newline
https://www.reuters.com/sports/tennis/venus-williams-prioritising-fun-she-returns-after-16-month-absence-2025-07-20/

\textquotedblleft \emph{Every match that I've played since the beginning, I
was just trying to win every point, [...] but the most important thing is to
just have fun and enjoy.}\textquotedblright\ --- Mirra Andreeva, Tennis.com,
Wimbledon 2025. \newline
https://www.tennis.com/news/articles/carlos-alcaraz-joy-wimbledon-final

\textquotedblleft \emph{[...] I spend every single moment enjoying the play,
enjoying the fact that you're competing and just make it a fun game [...]
every decision I went for felt like it was absolutely the right decision at
the right time. It's what I like to call flow.}\textquotedblright\ ---
Stefanos Tsitsipas, Interview by James Richardson, 2023. \newline
https://www.tennis365.com/tennis-news/stefanos-tsitsipas-reveals-what-means-flow

\textquotedblleft \emph{I felt like I was moving in a flow.}%
\textquotedblright\ --- Iga Swiatek, post-match interview, Australian Open
2023. \newline
https://www.tntsports.co.uk/tennis/australian-open/2023/i-felt-like-i-was-moving-in-a-flow-iga-swiateks-post-match-interview-after-latest-win-at-australian-open\_vid1815515/video.shtml

\textquotedblleft \emph{I go out there because I love tennis and I love
playing.}\textquotedblright\ --- Bethanie Mattek-Sands. \newline
https://www.edwardssports.co.uk/news/post/practice-tennis-like-a-pro

\textquotedblleft \emph{The glory is being happy. The glory is not winning
here or winning there. The glory is enjoying practicing, enjoy every day,
enjoying to work hard, trying to be a better player than before.}%
\textquotedblright\ --- Rafael Nadal, Sports Illustrated (interview by Jon
Wertheim), 2010. \newline
https://www.si.com/more-sports/2010/07/16/nadal-interview

\textquotedblleft \emph{I play tennis for the love of the game, having fun
and what pleasure the game affords me.}\textquotedblright\ --- Evonne
Goolagong. \newline
https://www.tennis-prose.com/articles/scoop/some-fantastic-tennis-quotes

\textquotedblleft \emph{What separates good tennis players from the truly
great? Their ability to enter flow state.}\textquotedblright\ \newline
https://www.streetfamilytennis.com/flow-state-in-tennis-mastering-the-mental-game
\end{quote}

Taken together, these testimonies highlight a consistent theme: elite
athletes derive substantial intrinsic utility from the act of playing
itself, aligning closely with the concept of \emph{autotelic experience} in
flow theory.

\subsection{Relation to Klaassen and Magnus (2009)\label{Model:KM}}

In the Klaassen and Magnus (2009) (KM from now on) model, the probability of
a serve being in is denoted $x$. This probability reflects the choice of the
server. As depicted in the above observations, if the server wants to take
more risk, he will choose a lower value of $x$, hence a lower serve
percentage. In contrast, if he wants to take fewer risks, he will choose a
higher value of $x$, a higher serve percentage. Since in tennis, players can
serve a second serve if the first is out, the server's strategy consists, in
fact, of two numbers $x_{1}$ and $x_{2}$ reflecting respectively the
probability of the first and second serve to be in. A player's probability
to win a point on his serve, denoted $p\left( x_{1},x_{2}\right) $, depends
on his strategy $\left( x_{1},x_{2}\right) $. Denote $y\left( x\right) $ the
probability of winning a point conditional on the probability of the serve
being in, $x$. $y\left( x\right) $ reflects the skills of the player,
encompassing both his serving and rally skills. $y\left( x\right) $ is
assumed to be twice differentiable and in particular to be strictly
decreasing, i.e., $y^{\prime }\left( x\right) <0$, so that the easier the
serve, that is, the higher the probability that it is in, the lower the
probability to win the point conditional on the serve being in and concave, $%
y^{\prime \prime }\left( x\right) <0$.

With these definitions, the probability to win a point on one's own serve,
the outcome utility, reads as,%
\begin{eqnarray*}
p\left( x_{1},x_{2}\right) &=&x_{1}y\left( x_{1}\right) +\left(
1-x_{1}\right) x_{2}y\left( x_{2}\right) \\
&=&w\left( x_{1}\right) +\left( 1-x_{1}\right) w\left( x_{2}\right) ,
\end{eqnarray*}%
where $w\left( x\right) :=xy\left( x\right) $ is the unconditional
probability to win a point when the serve is in.

A server aiming at maximizing his probability to win a point on his serve
then does%
\begin{equation*}
\max_{x_{1},x_{2}}p\left( x_{1},x_{2}\right) .
\end{equation*}

With mild regular conditions on the conditional probability $y\left(
x\right) $, KM show that a unique solution exists.

The FOCs of this problem read as%
\begin{eqnarray*}
\frac{\partial p\left( x_{1},x_{2}\right) }{\partial x_{1}} &=&w^{\prime
}\left( x_{1}\right) -w\left( x_{2}\right) =0, \\
\frac{\partial p\left( x_{1},x_{2}\right) }{\partial x_{2}} &=&\left(
1-x_{1}\right) w^{\prime }\left( x_{2}\right) =0.
\end{eqnarray*}

Denote the unique optimal solution $\left( x_{1}^{\ast },x_{2}^{\ast
}\right) $. KM have shown that this solution is so that%
\begin{equation*}
x_{1}^{\ast }<x_{2}^{\ast },\text{ }y\left( x_{1}^{\ast }\right) >y\left(
x_{2}^{\ast }\right) ,\text{ }w\left( x_{1}^{\ast }\right) <w\left(
x_{2}^{\ast }\right) .
\end{equation*}

For further analysis, a parametric functional form must be adopted for the
conditional probability $y\left( x\right) $. A particularily attractive
shape is the power function, which, with only 3 parameters offers a very
good trade-off between tractability and flexibility\footnote{%
Power functions appear to fit well sports performance data in general. The
relationship between world record time and distance in running, speed
skating and swimming for instance is best fit with a power function.}%
\begin{equation*}
y\left( x\right) =\frac{a-x^{\lambda }}{\tau }
\end{equation*}%
where the regularity conditions are met when $1\leq a\leq \tau
+x_{0}^{\lambda }<\tau +x^{\ast \lambda }$.

The parameters $a$, $\tau $ and $\lambda $ are skills parameters in the
sense that conditional on the server's decision $x$, they determine the
probability of winning a point on own serve. Interestingly, if the curve $%
y\left( x\right) $ was linear, i.e., $\lambda =1$, it would suffice to
observe two points on the curve to identify the parameters $\left( \tau
,\alpha \right) $, i.e., slope and constant. Hence, data for the first and
second serve would be enough to identify the skills parameters. However, the
presence of the curvature parameter requires at least a third point or an
additional restriction for its identification.

To understand the identification strategy of KM, let us first assume that
players differ in their skills and hence in their conditional probability to
win a point and if $i$ indexes players, then $\left( a_{i},\tau _{i},\lambda
_{i}\right) $ are the skills parameters of player $i$. The KM approach to
identify the parameters $\left( a_{i},\tau _{i},\lambda _{i}\right) $ for
all players $i=1,..,N$, using data on $\left(
x_{1i},x_{2i},y_{1i},y_{2i}\right) $ for $i=1,...,N$, is the following.
First, note that the data for each player offer two points on the curve $%
y_{i}\left( x\right) $, namely $\left( x_{1i},y_{1i}\right) $ and $\left(
x_{2i},y_{2i}\right) $. This means that given a curvature parameter $\lambda
_{i}$, one can, for each player $i$, identify the slope $\tau _{i}$ and
constant $a_{i}$. Noting this, KM therefore first randomly draw, for each
player, a curvature parameter $\lambda _{i}$ from a Gaussian distribution
with mean $\mu _{\lambda }$ and standard deviation $\sigma _{\lambda }$ and
then use the two data points $\left( x_{1i},y_{1i}\right) $ and $\left(
x_{2i},y_{2i}\right) $ for each player $i$, to compute $\left( \alpha
_{i},\tau _{i}\right) $ given $\lambda _{i}$. Given the draw of $\lambda _{i}
$, the shape of $y_{i}\left( x\right) $ is fully determined and the optimal
solution $\left( x_{1i}^{\ast },x_{2i}^{\ast }\right) $ can be derived from
the FOCs. This yields two additional points $\left( x_{1i}^{\ast
},y_{1i}^{\ast }\right) $ and $\left( x_{2i}^{\ast },y_{2i}^{\ast }\right) $
that can be used to pin down the (distribution of the) curvature parameters $%
\lambda _{i}$, $i=1,...,N$. The KM strategy consists in picking the value of
the mean and standard deviation $\left( \mu _{\lambda },\sigma _{\lambda
}\right) $ that maximizes average efficiency across players and whose
definition is given by%
\begin{equation*}
eff=E\left( \frac{p\left( x_{1},x_{2}\right) }{p\left( x_{1}^{\ast
},x_{2}^{\ast }\right) }\right) \simeq \frac{1}{n}\sum_{i=1}^{n}\frac{%
x_{1i}y_{1i}+\left( 1-x_{1i}\right) x_{2i}y_{2i}}{x_{1i}^{\ast }y_{1i}^{\ast
}+\left( 1-x_{1i}^{\ast }\right) x_{2i}^{\ast }y_{2i}^{\ast }}.
\end{equation*}

With respect to our project that are 4 important findings from their
analysis: $\lambda $ is virtually the same across (male) players as $\sigma
_{\lambda }$ is not statistically different from $0$, the mean is $\mu
_{\lambda }=3.07$, players are not perfect optimizers but close enough (the
average/median player has an inefficiency of 1.1\%, $eff=0.989$), there is
more inefficiency on the second serve than on the first and, better players
are better optimizers.


\subsection{Second serve strategy illustrated\label{SecondServe}}

The additional source of identification provided by distinguishing between
one-shot and multi-shot rallies can be best illustrated using Figure (\ref%
{fig_Opt2ndRF}). This figure shows how the optimal second service strategy
is determined given the skills parameters of the player (the shapes of $f$
and $k$) and his preference parameter for multi-shot rallies ($\beta $) and
how this generates one point on the conditional probability of one-shot
rallies and one point on the conditional probability of multi-shot rallies.
We herewith use the example of Roger Federer.

\begin{figure}[tbph]
\caption{Optimal second serve strategy: Roger Federer.}
\label{fig_Opt2ndRF}\centering\includegraphics[width=0.8%
\textwidth]{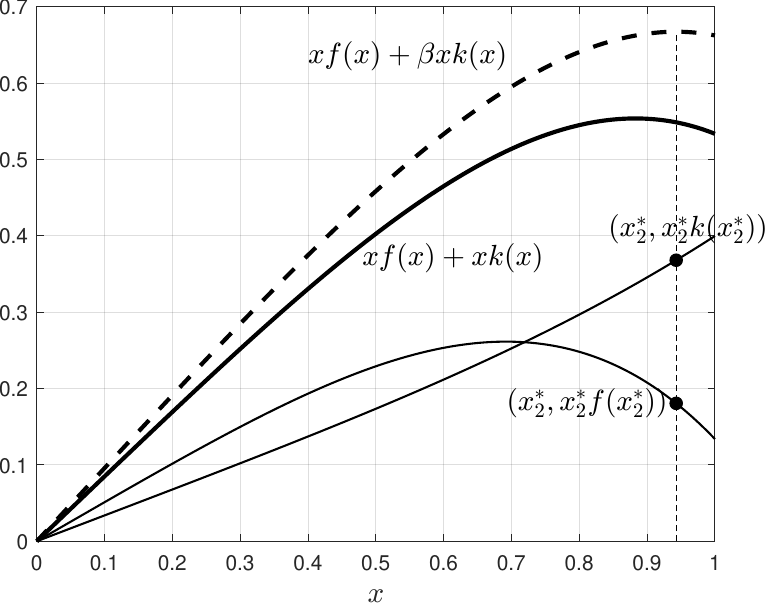}
\end{figure}

The optimal second serve strategy is obtained as the value $x_{2}^{\ast }$
that maximizes $\tilde{w}(x)=xf(x)+\beta xk(x)$. Note that it does not
maximize $w(x)=xf(x)+xk(x)$. Although changing the value of $\beta $ might
result in a non-negligible change in the second serve percentage, we see
that the probability of winning a point on one's serve might actually only
change marginally, $w(x_{2}^{\ast })$ being close to $\max_{x}(w(x))$. This
marginal change, however, because of the scoring system in tennis, magnifies
into a larger change in the probability to win a game, a set, and eventually
a match. This, in turn, transforms into significantly larger expected prize
money. 

\subsection{Deriving the first and second order conditions}

\subsubsection{Theory\label{FOCs:Theory}}

The first order condition, after substitution of the expression $\tilde{y}%
\left( x\right) $ yields 
\begin{eqnarray*}
\frac{\partial \tilde{w}\left( x_{1}\right) }{\partial x}-\tilde{w}\left(
x_{2}\right)  &:&=\tilde{y}\left( x_{1}\right) +x_{1}\frac{\partial \tilde{y}%
\left( x_{1}\right) }{\partial x}-x_{2}\tilde{y}\left( x_{2}\right) =0 \\
&\Leftrightarrow & \\
\left[ f\left( x_{1}\right) +\beta k\left( x_{1}\right) \right] +x_{1}\left[ 
\frac{\partial f\left( x_{1}\right) }{\partial x}+\beta \frac{\partial
k\left( x_{1}\right) }{\partial x}\right]  &=&x_{2}\tilde{y}\left(
x_{2}\right) 
\end{eqnarray*}%
while proceeding similarly for the second FOC yields%
\begin{eqnarray*}
\frac{\partial \tilde{w}\left( x_{2}\right) }{\partial x} &:&=\frac{\partial
x_{2}\tilde{y}\left( x_{2}\right) }{\partial x}=\tilde{y}\left( x_{2}\right)
+x_{2}\frac{\partial \tilde{y}\left( x_{2}\right) }{\partial x} \\
&=&\left[ f\left( x_{2}\right) +\beta k\left( x_{2}\right) \right] +x_{2}%
\left[ \frac{\partial f\left( x_{2}\right) }{\partial x}+\beta \frac{%
\partial k\left( x_{2}\right) }{\partial x}\right] =0.
\end{eqnarray*}

To derive the Hessian we need the second order partial derivatives of the
expected utility function. One has%
\begin{eqnarray*}
\frac{\partial ^{2}E\left( x_{1},x_{2}\right) }{\partial x_{1}^{2}} &=&\frac{%
\partial ^{2}\left( \tilde{w}\left( x_{1}\right) +\left( 1-x_{1}\right) 
\tilde{w}\left( x_{2}\right) \right) }{\partial x_{1}^{2}} \\
&=&\frac{\partial \left( \tilde{y}\left( x_{1}\right) +x_{1}\frac{\partial 
\tilde{y}\left( x_{1}\right) }{\partial x_{1}}-x_{2}\tilde{y}\left(
x_{2}\right) \right) }{\partial x_{1}} \\
&=&2\frac{\partial \tilde{y}\left( x_{1}\right) }{\partial x_{1}}+x_{1}\frac{%
\partial ^{2}\tilde{y}\left( x_{1}\right) }{\partial x_{1}^{2}} \\
&=&2\left( \frac{\partial f\left( x_{1}\right) }{\partial x_{1}}+\beta \frac{%
\partial k\left( x_{1}\right) }{\partial x_{1}}\right) +x_{1}\left( \frac{%
\partial ^{2}f\left( x_{1}\right) }{\partial x_{1}^{2}}+\beta \frac{\partial
^{2}k\left( x_{1}\right) }{\partial x_{1}^{2}}\right) 
\end{eqnarray*}%
and,%
\begin{eqnarray*}
\frac{\partial ^{2}E\left( x_{1},x_{2}\right) }{\partial x_{1}\partial x_{2}}
&=&\frac{\partial ^{2}\left( x_{1}\tilde{y}\left( x_{1}\right) +\left(
1-x_{1}\right) x_{2}\tilde{y}\left( x_{2}\right) \right) }{\partial
x_{1}\partial x_{2}} \\
&=&\frac{\partial \left( \left( 1-x_{1}\right) \left( \tilde{y}\left(
x_{2}\right) +x_{2}\frac{\partial \tilde{y}\left( x_{2}\right) }{\partial
x_{2}}\right) \right) }{\partial x_{1}}=-\left( \tilde{y}\left( x_{2}\right)
+x_{2}\frac{\partial \tilde{y}\left( x_{2}\right) }{\partial x_{2}}\right) 
\\
&=&-\left( f\left( x_{2}\right) +\beta k\left( x_{2}\right) +x_{2}\left( 
\frac{\partial f\left( x_{2}\right) }{\partial x_{2}}+\beta \frac{\partial
k\left( x_{2}\right) }{\partial x_{2}}\right) \right) 
\end{eqnarray*}%
and%
\begin{eqnarray*}
\frac{\partial ^{2}E\left( x_{1},x_{2}\right) }{\partial x_{2}^{2}} &=&\frac{%
\partial \left( \left( 1-x_{1}\right) \left( \tilde{y}\left( x_{2}\right)
+x_{2}\frac{\partial \tilde{y}\left( x_{2}\right) }{\partial x_{2}}\right)
\right) }{\partial x_{2}}=\left( 1-x_{1}\right) \left( 2\frac{\partial 
\tilde{y}\left( x_{2}\right) }{\partial x_{2}}+x_{2}\frac{\partial ^{2}%
\tilde{y}\left( x_{2}\right) }{\partial x_{2}^{2}}\right)  \\
&=&\left( 1-x_{1}\right) \left( 2\left( \frac{\partial f\left( x_{2}\right) 
}{\partial x_{2}}+\beta \frac{\partial k\left( x_{2}\right) }{\partial x_{2}}%
\right) +x_{2}\left( \frac{\partial ^{2}f\left( x_{2}\right) }{\partial
x_{2}^{2}}+\beta \frac{\partial ^{2}k\left( x_{2}\right) }{\partial x_{2}^{2}%
}\right) \right) .
\end{eqnarray*}

Collecting terms, this yields%
\begin{equation*}
\nabla ^{2}E\left( x_{1},x_{2}\right) =\left( 
\begin{array}{cc}
\begin{array}{c}
2\left( \frac{\partial f\left( x_{1}\right) }{\partial x_{1}}+\beta \frac{%
\partial k\left( x_{1}\right) }{\partial x_{1}}\right) \\ 
+x_{1}\left( \frac{\partial ^{2}f\left( x_{1}\right) }{\partial x_{1}^{2}}%
+\beta \frac{\partial ^{2}k\left( x_{1}\right) }{\partial x_{1}^{2}}\right)%
\end{array}
& -\left( 
\begin{array}{c}
f\left( x_{2}\right) +\beta k\left( x_{2}\right) \\ 
+x_{2}\left( \frac{\partial f\left( x_{2}\right) }{\partial x_{2}}+\beta 
\frac{\partial k\left( x_{2}\right) }{\partial x_{2}}\right)%
\end{array}%
\right) \\ 
-\left( 
\begin{array}{c}
f\left( x_{2}\right) +\beta k\left( x_{2}\right) \\ 
+x_{2}\left( \frac{\partial f\left( x_{2}\right) }{\partial x_{2}}+\beta 
\frac{\partial k\left( x_{2}\right) }{\partial x_{2}}\right)%
\end{array}%
\right) & \left( 1-x_{1}\right) \left( 
\begin{array}{c}
2\left( \frac{\partial f\left( x_{2}\right) }{\partial x_{2}}+\beta \frac{%
\partial k\left( x_{2}\right) }{\partial x_{2}}\right) \\ 
+x_{2}\left( \frac{\partial ^{2}f\left( x_{2}\right) }{\partial x_{2}^{2}}%
+\beta \frac{\partial ^{2}k\left( x_{2}\right) }{\partial x_{2}^{2}}\right)%
\end{array}%
\right)%
\end{array}%
\right)
\end{equation*}

However, note that, by the second FOC, one has $\frac{\partial ^{2}E\left(
x_{1},x_{2}^{\ast }\right) }{\partial x_{1}\partial x_{2}}=0$ so that the
Hessian at optimum second serve reads as%
\begin{equation*}
\nabla ^{2}E\left( x_{1},x_{2}^{\ast }\right) =\left( 
\begin{array}{cc}
\begin{array}{c}
2\left( \frac{\partial f\left( x_{1}\right) }{\partial x_{1}}+\beta \frac{%
\partial k\left( x_{1}\right) }{\partial x_{1}}\right) \\ 
+x_{1}\left( \frac{\partial ^{2}f\left( x_{1}\right) }{\partial x_{1}^{2}}%
+\beta \frac{\partial ^{2}k\left( x_{1}\right) }{\partial x_{1}^{2}}\right)%
\end{array}
& 0 \\ 
0 & \left( 1-x_{1}\right) \left( 
\begin{array}{c}
2\left( \frac{\partial f\left( x_{2}^{\ast }\right) }{\partial x_{2}}+\beta 
\frac{\partial k\left( x_{2}^{\ast }\right) }{\partial x_{2}}\right) \\ 
+x_{2}^{\ast }\left( \frac{\partial ^{2}f\left( x_{2}^{\ast }\right) }{%
\partial x_{2}^{2}}+\beta \frac{\partial ^{2}k\left( x_{2}^{\ast }\right) }{%
\partial x_{2}^{2}}\right)%
\end{array}%
\right)%
\end{array}%
\right) .
\end{equation*}

\subsubsection{Parametric\label{FOCs:Parametric}}

Using the parametric assumptions of the paper, the FOC for second serve
strategy reads as%
\begin{eqnarray*}
\left[ f\left( x_{2}\right) +\beta k\left( x_{2}\right) \right] +x_{2}\left[ 
\frac{\partial f\left( x_{2}\right) }{\partial x}+\beta \frac{\partial
k\left( x_{2}\right) }{\partial x}\right]  &=&0 \\
&\Leftrightarrow & \\
\left[ \frac{a_{f}-x_{2}^{\lambda }}{\tau _{f}}+\beta \frac{%
a_{k}-x_{2}^{\lambda }}{\tau _{k}}\right] -\lambda \left[ \frac{%
x_{2}^{\lambda }}{\tau _{f}}+\beta \frac{x_{2}^{\lambda }}{\tau _{k}}\right]
&=&0 \\
&\Leftrightarrow & \\
\frac{a_{f}}{\tau _{f}}+\beta \frac{a_{k}}{\tau _{k}}-\frac{x_{2}^{\lambda }%
}{\tau _{f}}-\beta \frac{x_{2}^{\lambda }}{\tau _{k}} &=&\lambda \frac{%
x_{2}^{\lambda }}{\tau _{f}}+\lambda \beta \frac{x_{2}^{\lambda }}{\tau _{k}}
\\
&\Leftrightarrow & \\
\left( \lambda \frac{\tau _{k}}{\tau _{f}\tau _{k}}+\frac{\tau _{k}}{\tau
_{f}\tau _{k}}+\lambda \beta \frac{\tau _{f}}{\tau _{f}\tau _{k}}+\beta 
\frac{\tau _{f}}{\tau _{f}\tau _{k}}\right) x_{2}^{\lambda } &=&\frac{%
a_{f}\tau _{k}}{\tau _{f}\tau _{k}}+\beta \frac{a_{k}\tau _{f}}{\tau
_{f}\tau _{k}} \\
&\Leftrightarrow & \\
\left( \lambda \tau _{k}+\tau _{k}+\lambda \beta \tau _{f}+\beta \tau
_{f}\right) x_{2}^{\lambda } &=&a_{f}\tau _{k}+\beta a_{k}\tau _{f} \\
&\Leftrightarrow & \\
x_{2}^{\ast } &=&\left( \frac{1}{\lambda +1}\frac{a_{f}\tau _{k}+\beta
a_{k}\tau _{f}}{\tau _{k}+\beta \tau _{f}}\right) ^{\frac{1}{\lambda }}
\end{eqnarray*}

The expression for the FOC for the first serve optimal strategy reads as 
\begin{eqnarray*}
\left[ f\left( x_{1}\right) +\beta k\left( x_{1}\right) \right] +x_{1}\left[ 
\frac{\partial f\left( x_{1}\right) }{\partial x}+\beta \frac{\partial
k\left( x_{1}\right) }{\partial x}\right]  &=&x_{2}\tilde{y}\left(
x_{2}\right)  \\
&\Leftrightarrow & \\
\left[ \frac{a_{f}-x_{1}^{\lambda }}{\tau _{f}}+\beta \frac{%
a_{k}-x_{1}^{\lambda }}{\tau _{k}}\right] -\lambda \left[ \frac{%
x_{1}^{\lambda }}{\tau _{f}}+\beta \frac{x_{1}^{\lambda }}{\tau _{k}}\right]
&=&x_{2}\left( \frac{a_{f}-x_{2}^{\lambda }}{\tau _{f}}+\beta \frac{%
a_{k}-x_{2}^{\lambda }}{\tau _{k}}\right)  \\
&\Leftrightarrow & \\
\left[ \frac{a_{f}}{\tau _{f}}-\frac{x_{1}^{\lambda }}{\tau _{f}}+\beta 
\frac{a_{k}}{\tau _{k}}-\beta \frac{x_{1}^{\lambda }}{\tau _{k}}\right] -%
\left[ \lambda \frac{x_{1}^{\lambda }}{\tau _{f}}+\lambda \beta \frac{%
x_{1}^{\lambda }}{\tau _{k}}\right]  &=&x_{2}\left( \frac{%
a_{f}-x_{2}^{\lambda }}{\tau _{f}}+\beta \frac{a_{k}-x_{2}^{\lambda }}{\tau
_{k}}\right)  \\
&\Leftrightarrow & \\
\left[ \frac{a_{f}\tau _{k}}{\tau _{f}\tau _{k}}-\frac{\tau
_{k}x_{1}^{\lambda }}{\tau _{f}\tau _{k}}+\beta \frac{a_{k}\tau _{f}}{\tau
_{f}\tau _{k}}-\beta \frac{\tau _{f}x_{1}^{\lambda }}{\tau _{f}\tau _{k}}%
\right] -\left[ \lambda \frac{\tau _{k}x_{1}^{\lambda }}{\tau _{f}\tau _{k}}%
+\lambda \beta \frac{\tau _{f}x_{1}^{\lambda }}{\tau _{f}\tau _{k}}\right] 
&=&x_{2}\left( \frac{a_{f}\tau _{k}-\tau _{k}x_{2}^{\lambda }}{\tau _{f}\tau
_{k}}+\beta \frac{a_{k}\tau _{f}-\tau _{f}x_{2}^{\lambda }}{\tau _{f}\tau
_{k}}\right)  \\
&\Leftrightarrow & \\
\left( \tau _{k}+\beta \tau _{f}+\lambda \tau _{k}+\lambda \beta \tau
_{f}\right) x_{1}^{\lambda } &=&a_{f}\tau _{k}+\beta a_{k}\tau _{f} \\
&&-x_{2}\left( a_{f}\tau _{k}-\tau _{k}x_{2}^{\lambda }+\beta a_{k}\tau
_{f}-\beta \tau _{f}x_{2}^{\lambda }\right)  \\
&\Leftrightarrow & \\
\left( \frac{a_{f}\tau _{k}+\beta a_{k}\tau _{f}-x_{2}\left( a_{f}\tau
_{k}+\beta a_{k}\tau _{f}-\left( \tau _{k}+\beta \tau _{f}\right)
x_{2}^{\lambda }\right) }{\left( 1+\lambda \right) \left( \tau _{k}+\beta
\tau _{f}\right) }\right) ^{\frac{1}{\lambda }} &=&x_{1}
\end{eqnarray*}

The expression for the Hessian of the expected utility at optimum reads as%
\begin{equation*}
\nabla ^{2}E\left( x_{1}^{\ast },x_{2}^{\ast }\right) =\left( 
\begin{array}{cc}
-\lambda \left( 1+\lambda \right) \left( x_{1}^{\ast }\right) ^{\lambda
-1}\left( \frac{1}{\tau _{f}}+\frac{\beta }{\tau _{k}}\right) & 0 \\ 
0 & -\left( 1-x_{1}\right) \lambda \left( 1+\lambda \right) \left(
x_{2}^{\ast }\right) ^{\lambda -1}\left( \frac{1}{\tau _{f}}+\frac{\beta }{%
\tau _{k}}\right)%
\end{array}%
\right)
\end{equation*}

Interestingly, as shown below, the optimal service strategy of the player
only depends on the relative slope $\frac{\tau _{f}}{\tau _{k}}$ as indeed
both expressions rewrite as%
\begin{eqnarray*}
x_{2}^{\ast \lambda } &=&\frac{a_{f}+\beta a_{k}\frac{\tau _{f}}{\tau _{k}}}{%
\left( 1+\lambda \right) \left( 1+\beta \frac{\tau _{f}}{\tau _{k}}\right) }
\\
x_{1}^{\ast \lambda } &=&\frac{a_{f}+\beta a_{k}\frac{\tau _{f}}{\tau _{k}}%
-x_{2}^{\ast }\left( a_{f}+\beta a_{k}\frac{\tau _{f}}{\tau _{k}}-\left(
1+\beta \frac{\tau _{f}}{\tau _{k}}\right) x_{2}^{\ast \lambda }\right) }{%
\left( 1+\lambda \right) \left( 1+\beta \frac{\tau _{f}}{\tau _{k}}\right) }
\end{eqnarray*}

and,

\begin{eqnarray*}
x_{2}^{\ast \lambda } &=&\frac{a_{f}\tau _{k}+\beta a_{k}\tau _{f}}{\left(
1+\lambda \right) \left( \tau _{k}+\beta \tau _{f}\right) } \\
&=&\frac{a_{f}+\beta a_{k}\frac{\tau _{f}}{\tau _{k}}}{\left( 1+\lambda
\right) \left( 1+\beta \frac{\tau _{f}}{\tau _{k}}\right) }
\end{eqnarray*}%
and for $x_{1}^{\ast \lambda }$:%
\begin{eqnarray*}
x_{1}^{\ast \lambda } &=&\frac{a_{f}\tau _{k}+\beta a_{k}\tau
_{f}-x_{2}^{\ast }\left( a_{f}\tau _{k}+\beta a_{k}\tau _{f}-\left( \tau
_{k}+\beta \tau _{f}\right) x_{2}^{\ast \lambda }\right) }{\left( 1+\lambda
\right) \left( \tau _{k}+\beta \tau _{f}\right) } \\
&=&\frac{a_{f}+\beta a_{k}\frac{\tau _{f}}{\tau _{k}}-x_{2}^{\ast }\left(
a_{f}+\beta a_{k}\frac{\tau _{f}}{\tau _{k}}-\left( 1+\beta \frac{\tau _{f}}{%
\tau _{k}}\right) x_{2}^{\ast \lambda }\right) }{\left( 1+\lambda \right)
\left( 1+\beta \frac{\tau _{f}}{\tau _{k}}\right) }
\end{eqnarray*}

\subsection{Comparative statics\label{CompStatics}}

Note that for $\lambda >1$, as estimated in the data, the signs of the
derivatives $\frac{\partial x_{j}^{\ast }}{\partial \beta }$ for $j=1,2$,
agree with those of $\frac{\partial x_{j}^{\ast \lambda }}{\partial \beta }$
which are simpler to derive. Hence, w.l.o.g., we produce comparative statics
on the preference parameter $\beta $ on the optimal strategy of the server
using $x_{j}^{\ast \lambda }$ rather than $x_{j}^{\ast }$.

Consider a small change in the preference parameter for winning multi-shot
rallies on the optimal probability of serving first and second serves in,
one has%
\begin{equation*}
\frac{\partial x_{1}^{\ast \lambda }}{\partial \beta }=-\tau _{f}\tau _{k}%
\frac{a_{f}-a_{k}}{\left( 1+\lambda \right) \left( \tau _{k}+\beta \tau
_{f}\right) ^{2}}\left( 1-x_{2}^{\ast }\right)
\end{equation*}%
so that $sign\left( \frac{\partial x_{1}^{\ast \lambda }}{\partial \beta }%
\right) =-sign\left( \tau _{k}\left( a_{f}-a_{k}\right) \right) $ since $%
x_{2}^{\ast }<1$ and

\begin{equation*}
\frac{\partial x_{2}^{\ast \lambda }}{\partial \beta }=\frac{\partial \frac{%
a_{f}+\beta a_{k}\frac{\tau _{f}}{\tau _{k}}}{\left( 1+\lambda \right)
\left( 1+\beta \frac{\tau _{f}}{\tau _{k}}\right) }}{\partial \beta }=-\tau
_{f}\tau _{k}\frac{a_{f}-a_{k}}{\left( 1+\lambda \right) \left( 1+\beta \tau
_{f}\right) ^{2}}
\end{equation*}%
so that $sign\left( \frac{\partial x_{2}^{\ast \lambda }}{\partial \beta }%
\right) =-sign\left( \tau _{k}\left( a_{f}-a_{k}\right) \right) =sign\left( 
\frac{\partial x_{2}^{\ast }}{\partial \beta }\right) $.

Hence, the optimal probability of serving first or second serves in, is
increasing with the preference parameter put on winning multi-shot points if 
$\tau _{k}\left( a_{f}-a_{k}\right) <0$ and vise versa.\footnote{%
In our data, all but one players meet the condition $\tau _{k}\left(
a_{f}-a_{k}\right) <0$.}

\begin{proposition}
\label{Proposition:CompStat}If $\tau _{k}\left( a_{f}-a_{k}\right) <0$
(resp. $\tau _{k}\left( a_{f}-a_{k}\right) >0$), the optimal probability of
serving first or second serves in is increasing (resp. decreasing) with the
preference for winning multi-shot rallies.
\end{proposition}

Note that, as can already be anticipated from the above comparative statics,
whenever $a_{f}=a_{k}$, the optimal serving strategy does not depend on the
preference parameter for winning multi-shot points of the player. This is
comfirmed by looking at this strategy for $a_{f}=a_{k}=a$. One indeed has%
\begin{equation*}
\left( x_{1}^{\ast \lambda },x_{2}^{\ast \lambda }\right) =\left( a\frac{%
1-x_{2}^{\ast }\left( 1-\frac{1}{a}x_{2}^{\ast \lambda }\right) }{1+\lambda }%
,a\frac{1}{1+\lambda }\right) .
\end{equation*}

These expressions do not depend on $\beta $.

\subsection{Sum of power functions with same curvature\label{SumPower}}

We aim to show that the sum of two power functions with same curvature
parameter $\lambda $, yield a power function with curvature parameter $%
\lambda $. Let $f\left( x\right) =\frac{a_{f}-x^{\lambda }}{\tau _{f}}$ and $%
k\left( x\right) =\frac{a_{k}-x^{\lambda }}{\tau _{k}}$. Then one has%
\begin{eqnarray*}
f\left( x\right) +k\left( x\right)  &=&\frac{a_{f}-x^{\lambda }}{\tau _{f}}+%
\frac{a_{k}-x^{\lambda }}{\tau _{k}} \\
&=&\frac{a_{f}\tau _{k}+a_{k}\tau _{f}}{\tau _{f}\tau _{k}}-\frac{\tau
_{f}+\tau _{k}}{\tau _{f}\tau _{k}}x^{\lambda } \\
&=&\frac{\frac{a_{f}\tau _{k}+a_{k}\tau _{f}}{\tau _{f}+\tau _{k}}}{\frac{%
\tau _{f}\tau _{k}}{\tau _{f}+\tau _{k}}}-\frac{1}{\frac{\tau _{f}\tau _{k}}{%
\tau _{f}+\tau _{k}}}x^{\lambda }
\end{eqnarray*}%
so that letting $\tau =\frac{\tau _{f}\tau _{k}}{\tau _{f}+\tau _{k}}$ and $%
a=\frac{a_{f}\tau _{k}+a_{k}\tau _{f}}{\tau _{f}+\tau _{k}}$ one indeed has%
\begin{equation*}
y\left( x\right) =f\left( x\right) +k\left( x\right) =\frac{a_{f}-x^{\lambda
}}{\tau _{f}}+\frac{a_{k}-x^{\lambda }}{\tau _{k}}=\frac{a-x^{\lambda }}{%
\tau }
\end{equation*}%
and 
\begin{eqnarray*}
\frac{\partial y\left( x\right) }{\partial x} &=&\frac{\partial f\left(
x\right) }{\partial x}+\frac{\partial k\left( x\right) }{\partial x}%
=-\lambda \left( \frac{x^{\lambda -1}}{\tau _{f}}+\frac{x^{\lambda -1}}{\tau
_{k}}\right)  \\
&=&-\lambda x^{\lambda -1}\left( \frac{\tau _{k}+\tau _{j}}{\tau _{f}\tau
_{k}}\right) =-\lambda \frac{x^{\lambda -1}}{\tau }
\end{eqnarray*}%
with in particular%
\begin{equation*}
x\frac{\partial y\left( x\right) }{\partial x}=-\lambda \frac{x^{\lambda }}{%
\tau }.
\end{equation*}

\subsection{Derivation of expression of $\protect\lambda ^{(t+1)}$\label%
{ExpressionLambda}}

Note first that the slopes, constants and preference parameter have the
following expressions in terms of $\lambda $ and $z_{1i}$ and $z_{2i}$ where
we recall that%
\begin{equation*}
z_{1i}=\left( x_{1i}\right) ^{\lambda },\text{ }z_{2i}=\left( x_{2i}\right)
^{\lambda }.
\end{equation*}

The slopes parameters read as%
\begin{equation*}
\tau _{fi}=-\frac{z_{1i}-z_{2i}}{f_{1i}-f_{2i}},\text{ }\tau _{ki}=-\frac{%
z_{1i}-z_{2i}}{k_{1i}-k_{2i}}.
\end{equation*}%
\bigskip 

The constants read as

\begin{equation*}
a_{fi}=-\frac{z_{1i}-z_{2i}}{f_{1i}-f_{2i}}f_{1i}+z_{1i},\text{ }a_{ki}=-%
\frac{z_{1i}-z_{2i}}{k_{1i}-k_{2i}}k_{1i}+z_{1i},
\end{equation*}%
and the attention parameter as%
\begin{equation*}
\beta _{i}=-\frac{\tau _{ki}}{\tau _{fi}}\frac{a_{fi}-z_{2i}\left( 1+\lambda
\right) }{a_{ki}-z_{2i}\left( 1+\lambda \right) }.
\end{equation*}

Substituting the epressions of the previous terms into $\beta _{i}$ yields 
\begin{eqnarray*}
\beta _{i} &=&-\frac{\tau _{ki}}{\tau _{fi}}\frac{\tau
_{fi}f_{2i}+z_{2i}\left( 1-1-\lambda \right) }{\tau _{ki}k_{2i}+z_{2i}\left(
1-1-\lambda \right) } \\
&=&-\frac{f_{1i}-f_{2i}}{k_{1i}-k_{2i}}\frac{\left( -\frac{z_{1i}-z_{2i}}{%
f_{1i}-f_{2i}}\right) f_{2i}-\lambda z_{2i}}{\left( -\frac{z_{1i}-z_{2i}}{%
k_{1i}-k_{2i}}\right) k_{2i}-\lambda z_{2i}} \\
&=&-\frac{\left( z_{1i}-z_{2i}\right) f_{2i}+\lambda z_{2i}\left(
f_{1i}-f_{2i}\right) }{\left( z_{1i}-z_{2i}\right) k_{2i}+\lambda
z_{2i}\left( k_{1i}-k_{2i}\right) } \\
&=&-\frac{z_{1i}f_{2i}+\lambda z_{2i}f_{1i}-z_{2i}f_{2i}\left( 1+\lambda
\right) }{z_{1i}k_{2i}+\lambda z_{2i}k_{1i}-z_{2i}k_{2i}\left( 1+\lambda
\right) } \\
&=&-\frac{f_{2i}}{k_{2i}}\frac{\Delta z_{i}+\lambda \Delta f_{i}}{\Delta
z_{i}+\lambda \Delta k_{i}}
\end{eqnarray*}%
where $a_{fi}=\tau _{fi}f_{2i}+z_{2i}$, $a_{ki}=\tau _{ki}k_{2i}+z_{2i}$, $%
\frac{\tau _{fi}}{\tau _{ki}}=\frac{k_{1i}-k_{2i}}{f_{1i}-f_{2i}}$, $\tau
_{fi}=-\frac{z_{1i}-z_{2i}}{f_{1i}-f_{2i}}$, $\tau _{ki}=-\frac{z_{1i}-z_{2i}%
}{k_{1i}-k_{2i}}$ and $\Delta l_{i}=\frac{l_{1i}-l_{2i}}{l_{2i}}$ $\forall
l=x,f,k$.

Note also that rearranging the expression of $\beta _{i}$ offers the
following interesting expression that is used below%
\begin{equation*}
\beta _{i}\frac{\tau _{fi}}{\tau _{ki}}=-\frac{a_{fi}-z_{2i}\left( 1+\lambda
\right) }{a_{ki}-z_{2i}\left( 1+\lambda \right) }.
\end{equation*}

Recall now the expression of $\lambda _{i}$ and rearranging yields

\begin{eqnarray*}
\lambda _{i} &=&\frac{a_{fi}+a_{ki}\beta _{i}\frac{\tau _{fi}}{\tau _{ki}}}{%
\left( 1+\beta _{i}\frac{\tau _{fi}}{\tau _{ki}}\right) z_{1i}}\left(
1-x_{2i}\right) +x_{2i}\frac{z_{2i}}{z_{1i}}-1 \\
&=&\frac{a_{fi}+a_{ki}\beta _{i}\frac{\tau _{fi}}{\tau _{ki}}}{1+\beta _{i}%
\frac{\tau _{fi}}{\tau _{ki}}}\frac{1-x_{2i}}{z_{1i}}+x_{2i}\frac{z_{2i}}{%
z_{1i}}-1=\Gamma \frac{1-x_{2i}}{z_{1i}}+x_{2i}\frac{z_{2i}}{z_{1i}}-1
\end{eqnarray*}

where $\Gamma =\frac{a_{fi}+a_{ki}\beta _{i}\frac{\tau _{fi}}{\tau _{ki}}}{%
1+\beta _{i}\frac{\tau _{fi}}{\tau _{ki}}}$.

Let us now rewrite the term $\Gamma $ by using the interesting expression
derived above for $\beta _{i}\frac{\tau _{fi}}{\tau _{ki}}$ so that%
\begin{eqnarray*}
\Gamma  &=&\frac{a_{fi}\left( a_{ki}-z_{2i}\left( 1+\lambda \right) \right)
-a_{ki}\left( a_{fi}-z_{2i}\left( 1+\lambda \right) \right) }{%
a_{ki}-z_{2i}\left( 1+\lambda \right) -\left( a_{fi}-z_{2i}\left( 1+\lambda
\right) \right) } \\
&=&\frac{a_{fi}a_{ki}-a_{fi}a_{ki}+a_{ki}z_{2i}\left( 1+\lambda \right)
-a_{fi}z_{2i}\left( 1+\lambda \right) }{a_{ki}-a_{fi}+z_{2i}\left( 1+\lambda
\right) -z_{2i}\left( 1+\lambda \right) } \\
&=&\frac{a_{ki}-a_{fi}}{a_{ki}-a_{fi}}z_{2i}\left( 1+\lambda \right)
=z_{2i}\left( 1+\lambda \right) 
\end{eqnarray*}

Substituting this into the expression of $\lambda _{i}$, it follows that 
\begin{eqnarray*}
\lambda _{i} &=&\Gamma \frac{1-x_{2i}}{z_{1i}}+x_{2i}\frac{z_{2i}}{z_{1i}}-1
\\
&=&z_{2i}\left( 1+\lambda \right) \frac{1-x_{2i}}{z_{1i}}+x_{2i}\frac{z_{2i}%
}{z_{1i}}-1 \\
&=&\frac{z_{2i}}{z_{1i}}\left( \left( 1+\lambda \right) \left(
1-x_{2i}\right) +x_{2i}\right) -1 \\
&=&\frac{z_{2i}}{z_{1i}}\left( 1+\lambda \left( 1-x_{2i}\right) \right) -1
\end{eqnarray*}

Hence the compact expression of $\lambda _{i}$ as%
\begin{equation*}
\lambda _{i}=\left( \frac{x_{2i}}{x_{1i}}\right) ^{\lambda }\left( 1+\lambda
\left( 1-x_{2i}\right) \right) -1.
\end{equation*}

Note that $\Lambda \left( \lambda \right) >0$ since%
\begin{eqnarray*}
\left( \frac{x_{2i}}{x_{1i}}\right) ^{\lambda }\left( 1+\lambda \left(
1-x_{2i}\right) \right) -1 &>&0 \\
&\Leftrightarrow & \\
\left( \frac{x_{2i}}{x_{1i}}\right) ^{\lambda }\left( 1+\lambda \left(
1-x_{2i}\right) \right)  &>&1\Leftrightarrow \left( \frac{x_{2i}}{x_{1i}}%
\right) ^{\lambda }>\frac{1}{1+\lambda \left( 1-x_{2i}\right) }
\end{eqnarray*}%
since $\left( \frac{x_{2i}}{x_{1i}}\right) ^{\lambda }>1$ for $x_{2i}>x_{1i}$
and $\frac{1}{1+\lambda \left( 1-x_{2i}\right) }<1$ since 
\begin{equation*}
\frac{1}{1+\lambda \left( 1-x_{2i}\right) }<1\Leftrightarrow 1<1+\lambda
\left( 1-x_{2i}\right) \Leftrightarrow 0<\lambda \left( 1-x_{2i}\right) 
\end{equation*}%
which obtains from $\lambda >0$ and $0<x_{2i}<1$.

\newpage

\subsection{Tree representation of the statistics for each player $i$.\label%
{Tree}}

In Figure (\ref{fig:tree}) below, each node indicates the mass of points
observed for player $i$, and each edge represents a probability of success
or failure, conditional on having reached the previous node. At the staring
node $N_{i}$, the two edges represent the probability that the first serve
is in or out. If the service is in, we move to node $n_{x_{1}i}$ where the
two edges indicate the probability that the first serve is returned in or
not, where the former occurs with probability $f_{1i}$ and the latter with
probability $1-f_{1i}$. The terminal nodes that are accessed with horizontal
edges, correspond to the masses of points won by the server, and the paths
leading to them are highlighted in bold.

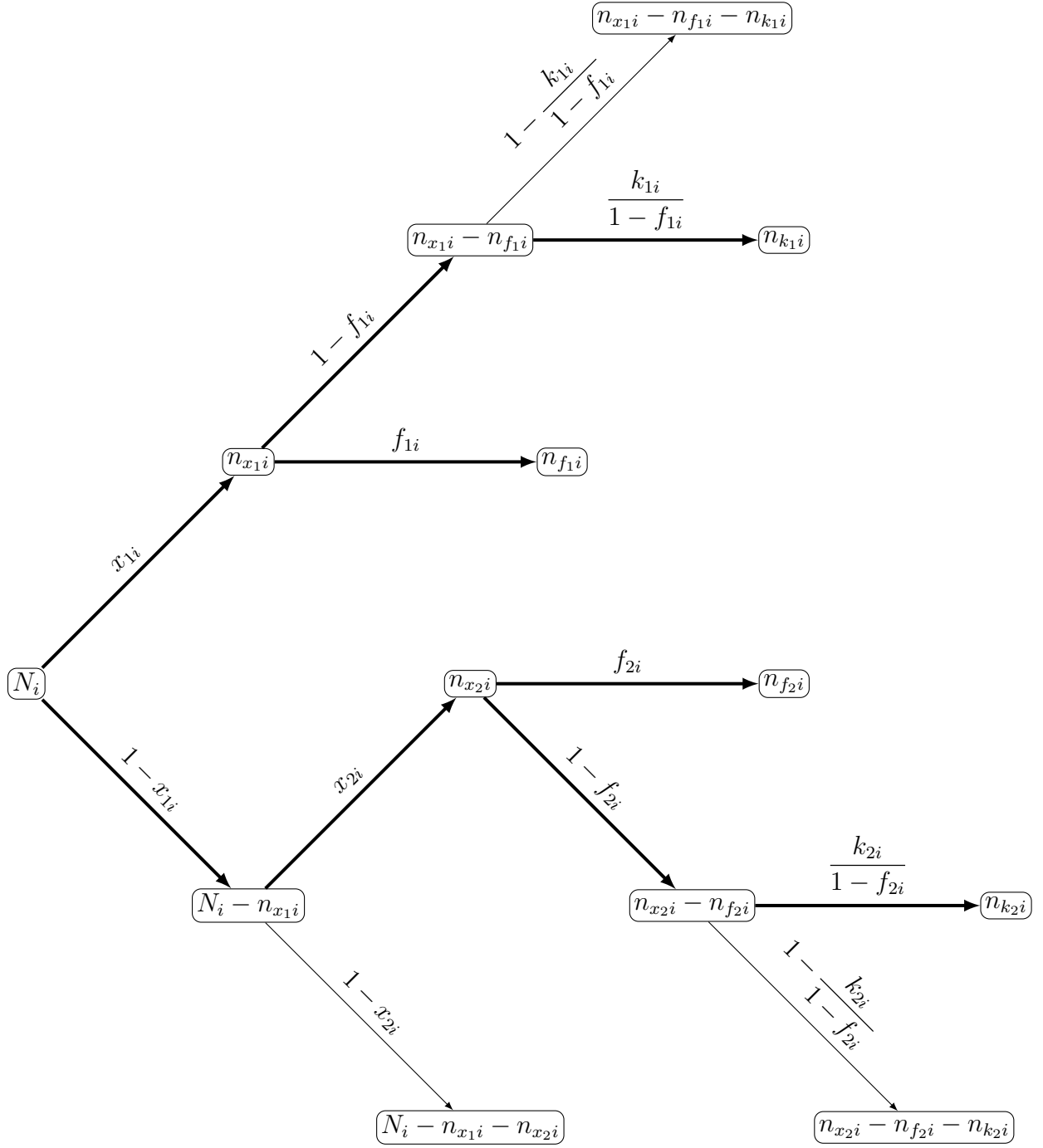
\begin{figure}[htbp]
\caption{Tree representation of tennis statistics for player $i$.}
\label{fig:tree}\makebox[\textwidth][c]{
\begin{tikzpicture}[>=latex]

\def\L{5cm}   
\def\A{45}    

\tikzset{
  treenode/.style={align=center, draw, rectangle, rounded corners, inner sep=2pt}
}

\node[treenode] (Ni) at (0,0) {$N_i$};

\node[treenode] (nx1)  at ($(Ni)+(\A:\L)$) {$n_{x_{1}i}$};
\node[treenode] (Nmx1) at ($(Ni)+(-\A:\L)$) {$N_i-n_{x_{1}i}$};

\draw[->, line width=1.5pt] (Ni) -- (nx1) node[midway, sloped, above] {$x_{1i}$};
\draw[->, line width=1.5pt] (Ni) -- (Nmx1) node[midway, sloped, above] {$1-x_{1i}$};

\node[treenode] (nf1) at ($(nx1)+(\L,0)$) {$n_{f_{1}i}$};             
\node[treenode] (nx1mf1) at ($(nx1)+(\A:\L)$) {$n_{x_{1}i}-n_{f_{1}i}$}; 

\draw[->, line width=1.5pt] (nx1) -- (nf1) node[midway, above] {$f_{1i}$};
\draw[->, line width=1.5pt] (nx1) -- (nx1mf1) node[midway, sloped, above] {$1-f_{1i}$};

\node[treenode] (nk1) at ($(nx1mf1)+(\L,0)$) {$n_{k_{1}i}$};             
\node[treenode] (nx1mf1mk1) at ($(nx1mf1)+(\A:\L)$) {$n_{x_{1}i}-n_{f_{1}i}-n_{k_{1}i}$}; 

\draw[->, line width=1.5pt] (nx1mf1) -- (nk1) node[midway, above] {$\dfrac{k_{1i}}{1-f_{1i}}$};             
\draw[->] (nx1mf1) -- (nx1mf1mk1) node[midway, sloped, above] {$1-\dfrac{k_{1i}}{1-f_{1i}}$};  

\node[treenode] (nx2) at ($(Nmx1)+(\A:\L)$) {$n_{x_{2}i}$};          
\node[treenode] (Nmx1mx2) at ($(Nmx1)+(-\A:\L)$) {$N_i-n_{x_{1}i}-n_{x_{2}i}$}; 

\draw[->, line width=1.5pt] (Nmx1) -- (nx2) node[midway, sloped, above] {$x_{2i}$};
\draw[->] (Nmx1) -- (Nmx1mx2) node[midway, sloped, above] {$1-x_{2i}$};

\node[treenode] (nf2) at ($(nx2)+(\L,0)$) {$n_{f_{2}i}$};            
\node[treenode] (nx2mf2) at ($(nx2)+(-\A:\L)$) {$n_{x_{2}i}-n_{f_{2}i}$};

\draw[->, line width=1.5pt] (nx2) -- (nf2) node[midway, above] {$f_{2i}$};
\draw[->, line width=1.5pt] (nx2) -- (nx2mf2) node[midway, sloped, above] {$1-f_{2i}$};

\node[treenode] (nk2) at ($(nx2mf2)+(\L,0)$) {$n_{k_{2}i}$};
\node[treenode] (nx2mf2mk2) at ($(nx2mf2)+(-\A:\L)$) {$n_{x_{2}i}-n_{f_{2}i}-n_{k_{2}i}$};

\draw[->, line width=1.5pt] (nx2mf2) -- (nk2) node[midway, above] {$\dfrac{k_{2i}}{1-f_{2i}}$};
\draw[->] (nx2mf2) -- (nx2mf2mk2) node[midway, sloped, above] {$1-\dfrac{k_{2i}}{1-f_{2i}}$};

\end{tikzpicture}
} 
\end{figure}

\newpage

\subsection{Model with softmax parametrization\label{Section:Softmax}}

Suppose that we replace Condition (\ref{Assumption:Elasticity}) used to
parametrize the elasticity of the marginal probabilities $f^{\prime }\left(
x\right) $ and $k^{\prime }\left( x\right) $, by the following condition:

\begin{condition}
\label{Assumption:Elasticity2}$x\frac{f^{\prime \prime }\left( x\right) }{%
f^{\prime }\left( x\right) }=x\frac{k^{\prime \prime }\left( x\right) }{%
k^{\prime }\left( x\right) }=\lambda x>0$.
\end{condition}

This condition tells us that the elasticity of the marginal conditional
probabilities $f^{\prime }$ and $k^{\prime }$ are not constant anymore but
rather linear. In this case, the associated shape for the function $f$ for
instance is%
\begin{equation*}
f\left( x\right) =a_{f}+\tau _{f}\exp \left( \lambda x\right) ,
\end{equation*}%
with%
\begin{equation*}
f^{\prime }\left( x\right) =\lambda \tau _{f}\exp \left( \lambda x\right) ,%
\text{ }f^{\prime \prime }\left( x\right) =\lambda ^{2}\tau _{f}\exp \left(
\lambda x\right) 
\end{equation*}%
where $\tau _{f}=\frac{f_{2}-f_{1}}{z_{2}-z_{1}}$ and $a_{f}=f_{1}-\tau
_{f}z_{1}$ and $z_{1}=\exp \left( \lambda x_{1}\right) $ and $z_{2}=\exp
\left( \lambda x_{2}\right) $.\footnote{%
Indeed one has 
\begin{eqnarray*}
x\frac{f^{\prime \prime }\left( x\right) }{f^{\prime }\left( x\right) } &=&x%
\frac{\lambda ^{2}\tau _{f}\exp \left( \lambda x\right) }{\lambda \tau
_{f}\exp \left( \lambda x\right) } \\
&=&x\lambda .
\end{eqnarray*}%
}

With the same logic for $k(x)$ and plugging these expressions into the
solution for preference parameter $\beta $, one has

\begin{eqnarray*}
\beta  &=&-\frac{f_{2}+x_{2}f^{\prime }\left( x_{2}\right) }{%
k_{2}+x_{2}k^{\prime }\left( x_{2}\right) }=-\frac{f_{2}+x_{2}\lambda \tau
_{f}\exp \left( \lambda x_{2}\right) }{k_{2}+x_{2}\lambda \tau _{k}\exp
\left( \lambda x_{2}\right) } \\
&=&-\frac{f_{2}+x_{2}\lambda \frac{f_{2}-f_{1}}{z_{2}-z_{1}}z_{2}}{%
k_{2}+x_{2}\lambda \frac{k_{2}-k_{1}}{z_{2}-z_{1}}z_{2}}=-\frac{f_{2}\left(
z_{2}-z_{1}\right) +\lambda x_{2}z_{2}\left( f_{2}-f_{1}\right) }{%
k_{2}\left( z_{2}-z_{1}\right) +\lambda x_{2}z_{2}\left( k_{2}-k_{1}\right) }%
.
\end{eqnarray*}

Similarily, using the FOC for the first serve strategy one obtains an
expression for the curvature parameter $\lambda $ as follows.

Step 1:

Remember that the solution for the curvature parameter obtains from:%
\begin{equation*}
\left[ f\left( x_{1}\right) +\beta k\left( x_{1}\right) \right] +x_{1}\left[
f^{\prime }\left( x_{1}\right) +\beta k^{\prime }\left( x_{1}\right) \right]
=x_{2}\left[ f\left( x_{2}\right) +\beta k\left( x_{2}\right) \right]
\end{equation*}%
which given the parametric assumptions become%
\begin{equation*}
\left[ f_{1}+\beta k_{1}\right] \left( \exp \left( \lambda \left(
x_{2}-x_{1}\right) \right) -1\right) +\lambda x_{1}\left[ f_{2}-f_{1}+\beta
\left( k_{2}-k_{1}\right) \right] =x_{2}\left[ f_{2}+\beta k_{2}\right]
\left( \exp \left( \lambda \left( x_{2}-x_{1}\right) \right) -1\right) .
\end{equation*}

Using the notation $z_{j}=\exp \left( \lambda x_{j}\right) $ this rewrite as%
\begin{equation*}
\left[ f_{1}+\beta k_{1}\right] \left( \frac{z_{2}}{z_{1}}-1\right) +\lambda
x_{1}\left[ f_{2}-f_{1}+\beta \left( k_{2}-k_{1}\right) \right] =x_{2}\left[
f_{2}+\beta k_{2}\right] \left( \frac{z_{2}}{z_{1}}-1\right) .
\end{equation*}

Note that this writes compactly as

\begin{equation*}
A\left( E-1\right) +\lambda B=0
\end{equation*}%
where $E=\frac{z_{2}}{z_{1}}$, $A=f_{1}+\beta k_{1}-x_{2}\left[ f_{2}+\beta
k_{2}\right] $ and $B=x_{1}\left[ f_{2}-f_{1}+\beta \left(
k_{2}-k_{1}\right) \right] $.

Step 2: Note that we can rewrite the preference parameter more compactly too
as 
\begin{equation*}
\beta =-\frac{f_{2}\left( z_{2}-z_{1}\right) +\lambda x_{2}z_{2}\left(
f_{2}-f_{1}\right) }{k_{2}\left( z_{2}-z_{1}\right) +\lambda
x_{2}z_{2}\left( k_{2}-k_{1}\right) }=-\frac{N}{M}
\end{equation*}

where $N=f_{2}+\lambda \left( f_{2}-f_{1}\right) Q$ and $M=k_{2}+\lambda
\left( k_{2}-k_{1}\right) Q$.

Substituting this last expression of $\beta $ into $A$ and $B$ of the
compact expression of the FOC for $x_{1}$ yields%
\begin{eqnarray*}
A &=&f_{1}-\frac{f_{2}+\lambda \left( f_{2}-f_{1}\right) Q}{k_{2}+\lambda
\left( k_{2}-k_{1}\right) Q}k_{1}-x_{2}\left[ f_{2}-\frac{f_{2}+\lambda
\left( f_{2}-f_{1}\right) Q}{k_{2}+\lambda \left( k_{2}-k_{1}\right) Q}k_{2}%
\right] \\
B &=&x_{1}\left[ f_{2}-f_{1}-\frac{f_{2}+\lambda \left( f_{2}-f_{1}\right) Q%
}{k_{2}+\lambda \left( k_{2}-k_{1}\right) Q}\left( k_{2}-k_{1}\right) \right]
,
\end{eqnarray*}%
where $Q=\frac{x_{2}E}{E-1}$ and which rewrites compactly as%
\begin{eqnarray*}
A &=&\frac{Mf_{1}-Nk_{1}-x_{2}\left[ Mf_{2}-Nk_{2}\right] }{M} \\
B &=&x_{1}\frac{M\left( f_{2}-f_{1}\right) -N\left( k_{2}-k_{1}\right) }{M},
\end{eqnarray*}

Note that%
\begin{eqnarray*}
M\left( f_{2}-f_{1}\right) -N\left( k_{2}-k_{1}\right)  &=&\left(
k_{2}+\lambda \left( k_{2}-k_{1}\right) Q\right) \left( f_{2}-f_{1}\right)
-\left( f_{2}+\lambda \left( f_{2}-f_{1}\right) Q\right) \left(
k_{2}-k_{1}\right)  \\
&=&k_{2}\left( f_{2}-f_{1}\right) -f_{2}\left( k_{2}-k_{1}\right)
=k_{2}f_{2}-k_{2}f_{1}-f_{2}k_{2}+f_{2}k_{1} \\
&=&f_{2}k_{1}-k_{2}f_{1},
\end{eqnarray*}%
so that $B=x_{1}\frac{f_{2}k_{1}-k_{2}f_{1}}{M}$.

Note further that $A=\frac{M\left( f_{1}-x_{2}f_{2}\right) -N\left(
k_{1}-x_{2}k_{2}\right) }{M}$, and after substitution and simple
rearranging, one has

\begin{eqnarray*}
M\left( f_{1}-x_{2}f_{2}\right) &=&k_{2}\left( f_{1}-x_{2}f_{2}\right)
+\left( f_{1}-x_{2}f_{2}\right) \left( k_{2}-k_{1}\right) \lambda Q, \\
N\left( k_{1}-x_{2}k_{2}\right) &=&f_{2}\left( k_{1}-x_{2}k_{2}\right)
+\left( k_{1}-x_{2}k_{2}\right) \left( f_{2}-f_{1}\right) \lambda Q.
\end{eqnarray*}

Hence, the numerator of $A$ reads as%
\begin{eqnarray*}
M\left( f_{1}-x_{2}f_{2}\right) -N\left( k_{1}-x_{2}k_{2}\right) 
&=&k_{2}\left( f_{1}-x_{2}f_{2}\right) -f_{2}\left( k_{1}-x_{2}k_{2}\right) 
\\
&&+\left[ \left( f_{1}-x_{2}f_{2}\right) \left( k_{2}-k_{1}\right) -\left(
k_{1}-x_{2}k_{2}\right) \left( f_{2}-f_{1}\right) \right] \lambda Q \\
&=&\left( f_{1}k_{2}-x_{2}f_{2}k_{2}\right) -\left(
f_{2}k_{1}-x_{2}f_{2}k_{2}\right)  \\
&&+\left[ f_{1}\left( k_{2}-k_{1}\right) -x_{2}f_{2}\left(
k_{2}-k_{1}\right) -\left( k_{1}\left( f_{2}-f_{1}\right) -x_{2}k_{2}\left(
f_{2}-f_{1}\right) \right) \right] \lambda Q \\
&=&f_{1}k_{2}-f_{2}k_{1}+\left[ f_{1}\left( k_{2}-k_{1}\right) -k_{1}\left(
f_{2}-f_{1}\right) +x_{2}f_{2}k_{1}-x_{2}k_{2}f_{1}\right] \lambda Q \\
&=&f_{1}k_{2}-f_{2}k_{1}+\left[ f_{1}k_{2}\left( 1-x_{2}\right)
-f_{2}k_{1}\left( 1-x_{2}\right) \right] \lambda Q \\
&=&\left( f_{1}k_{2}-f_{2}k_{1}\right) \left( 1+\lambda Q\left(
1-x_{2}\right) \right) ,
\end{eqnarray*}

and one obtains $A=\frac{\left( f_{1}k_{2}-f_{2}k_{1}\right) \left(
1+\lambda Q\left( 1-x_{2}\right) \right) }{M}$. Now, substituting these
expressions of $A$ and $B$ back into the compact expression of the FOC for $%
x_{1}$ gives

\begin{eqnarray*}
A\left( E-1\right) +\lambda B &=&0 \\
&\Leftrightarrow & \\
\frac{\left( f_{1}k_{2}-f_{2}k_{1}\right) \left( 1+\lambda Q\left(
1-x_{2}\right) \right) }{M}\left( E-1\right) +\lambda x_{1}\frac{%
f_{2}k_{1}-k_{2}f_{1}}{M} &=&0 \\
&\Leftrightarrow & \\
\left( f_{1}k_{2}-f_{2}k_{1}\right) \left( \left( 1+\lambda Q\left(
1-x_{2}\right) \right) \left( E-1\right) -\lambda x_{1}\right)  &=&0 \\
&\Leftrightarrow & \\
\left( 1+\lambda Q\left( 1-x_{2}\right) \right) \left( E-1\right) 
&=&\lambda x_{1} \\
&\Leftrightarrow & \\
\left( E-1\right)  &=&\lambda \left( x_{1}-\left( E-1\right) Q\left(
1-x_{2}\right) \right)  \\
&\Leftrightarrow & \\
\lambda  &=&\frac{E-1}{x_{1}-\left( E-1\right) Q\left( 1-x_{2}\right) }
\end{eqnarray*}%
provided $f_{1}k_{2}\neq f_{2}k_{1}$.

We hence obtain:%
\begin{eqnarray*}
\lambda  &=&\frac{E-1}{x_{1}-\left( E-1\right) Q\left( 1-x_{2}\right) }=%
\frac{E-1}{x_{1}-\left( E-1\right) \frac{x_{2}E}{E-1}\left( 1-x_{2}\right) }
\\
&=&\frac{\frac{z_{2}}{z_{1}}-1}{x_{1}-x_{2}\left( 1-x_{2}\right) \frac{z_{2}%
}{z_{1}}}=\frac{z_{2}-z_{1}}{x_{1}z_{1}-x_{2}\left( 1-x_{2}\right) z_{2}} \\
&=&\frac{\exp \left( \lambda x_{2}\right) -\exp \left( \lambda x_{1}\right) 
}{x_{1}\exp \left( \lambda x_{1}\right) -x_{2}\left( 1-x_{2}\right) \exp
\left( \lambda x_{2}\right) }
\end{eqnarray*}%
where we made use of the definitions: $Q=\frac{x_{2}E}{E-1}$ and $E=\exp
\left( \lambda \left( x_{2}-x_{1}\right) \right) =\frac{z_{2}}{z_{1}}$.

To summarize, in the softmax model we have%
\begin{eqnarray*}
\lambda &=&\frac{\exp \left( \lambda x_{2}\right) -\exp \left( \lambda
x_{1}\right) }{x_{1}\exp \left( \lambda x_{1}\right) -x_{2}\left(
1-x_{2}\right) \exp \left( \lambda x_{2}\right) }, \\
\beta &=&-\frac{f_{2}\left( \exp \left( \lambda x_{2}\right) -\exp \left(
\lambda x_{1}\right) \right) +\lambda \left( f_{2}-f_{1}\right) x_{2}\exp
\left( \lambda x_{2}\right) }{k_{2}\left( \exp \left( \lambda x_{2}\right)
-\exp \left( \lambda x_{1}\right) \right) +\lambda \left( k_{2}-k_{1}\right)
x_{2}\exp \left( \lambda x_{2}\right) }.
\end{eqnarray*}

Table (\ref{SmallEsti_softmax_CI}) shows the estimates of the parameters of
the model with this parametric specification.

\begin{table}[H]
\centering
\caption{Parameter estimates with 95\% bootstrap confidence intervals for selected players\label{SmallEsti_softmax_CI}}
\resizebox{\textwidth}{!}{%
\begin{tabular}{l|c|c|cc|cc}
\toprule
Player & Salience weight & Curvature & Slope $f$ & Const. $f$ & Slope $k$ & Const. $k$ \\
 & $\delta$ & $\lambda$ & $\tau_f$ & $a_f$ & $\tau_k$ & $a_k$ \\
\midrule
Novak Djokovic & \shortstack[l]{$0.32$\textsuperscript{***} \\ \footnotesize [0.20, 0.47]} & \shortstack[l]{$3.54$\textsuperscript{***} \\ \footnotesize [3.40, 3.70]} & \shortstack[l]{$-0.01$\textsuperscript{***} \\ \footnotesize [-0.01, -0.01]} & \shortstack[l]{$0.45$\textsuperscript{***} \\ \footnotesize [0.44, 0.46]} & \shortstack[l]{$-0.00$\textsuperscript{***} \\ \footnotesize [-0.00, -0.00]} & \shortstack[l]{$0.41$\textsuperscript{***} \\ \footnotesize [0.40, 0.41]} \\
Rafael Nadal & \shortstack[l]{$0.27$\textsuperscript{***} \\ \footnotesize [0.08, 0.54]} & \shortstack[l]{$4.97$\textsuperscript{***} \\ \footnotesize [4.75, 5.21]} & \shortstack[l]{$-0.00$\textsuperscript{***} \\ \footnotesize [-0.00, -0.00]} & \shortstack[l]{$0.33$\textsuperscript{***} \\ \footnotesize [0.32, 0.33]} & \shortstack[l]{$-0.00$\textsuperscript{} \\ \footnotesize [-0.00, 0.00]} & \shortstack[l]{$0.43$\textsuperscript{***} \\ \footnotesize [0.43, 0.44]} \\
Roger Federer & \shortstack[l]{$0.38$\textsuperscript{***} \\ \footnotesize [0.28, 0.49]} & \shortstack[l]{$2.46$\textsuperscript{***} \\ \footnotesize [2.35, 2.57]} & \shortstack[l]{$-0.04$\textsuperscript{***} \\ \footnotesize [-0.05, -0.04]} & \shortstack[l]{$0.60$\textsuperscript{***} \\ \footnotesize [0.58, 0.61]} & \shortstack[l]{$-0.01$\textsuperscript{***} \\ \footnotesize [-0.01, -0.00]} & \shortstack[l]{$0.38$\textsuperscript{***} \\ \footnotesize [0.38, 0.39]} \\
Pete Sampras & \shortstack[l]{$0.45$\textsuperscript{***} \\ \footnotesize [0.33, 0.60]} & \shortstack[l]{$1.35$\textsuperscript{***} \\ \footnotesize [1.20, 1.51]} & \shortstack[l]{$-0.26$\textsuperscript{***} \\ \footnotesize [-0.32, -0.21]} & \shortstack[l]{$1.10$\textsuperscript{***} \\ \footnotesize [1.04, 1.18]} & \shortstack[l]{$-0.06$\textsuperscript{***} \\ \footnotesize [-0.08, -0.05]} & \shortstack[l]{$0.40$\textsuperscript{***} \\ \footnotesize [0.38, 0.43]} \\
Boris Becker & \shortstack[l]{$2.73$\textsuperscript{***} \\ \footnotesize [0.50, 3.70]} & \shortstack[l]{$-11.45$\textsuperscript{} \\ \footnotesize [-11.89, 1.25]} & \shortstack[l]{$-0.00$\textsuperscript{***} \\ \footnotesize [-0.35, -0.00]} & \shortstack[l]{$0.49$\textsuperscript{***} \\ \footnotesize [0.48, 1.12]} & \shortstack[l]{$-0.00$\textsuperscript{***} \\ \footnotesize [-0.05, -0.00]} & \shortstack[l]{$0.30$\textsuperscript{***} \\ \footnotesize [0.29, 0.39]} \\
Carlos Alcaraz & \shortstack[l]{$0.10$\textsuperscript{} \\ \footnotesize [-0.08, 0.32]} & \shortstack[l]{$3.55$\textsuperscript{***} \\ \footnotesize [3.28, 3.81]} & \shortstack[l]{$-0.01$\textsuperscript{***} \\ \footnotesize [-0.01, -0.01]} & \shortstack[l]{$0.41$\textsuperscript{***} \\ \footnotesize [0.39, 0.43]} & \shortstack[l]{$-0.00$\textsuperscript{***} \\ \footnotesize [-0.00, -0.00]} & \shortstack[l]{$0.41$\textsuperscript{***} \\ \footnotesize [0.40, 0.42]} \\
Jannik Sinner & \shortstack[l]{$0.10$\textsuperscript{} \\ \footnotesize [-0.01, 0.24]} & \shortstack[l]{$2.11$\textsuperscript{***} \\ \footnotesize [1.93, 2.28]} & \shortstack[l]{$-0.06$\textsuperscript{***} \\ \footnotesize [-0.07, -0.05]} & \shortstack[l]{$0.60$\textsuperscript{***} \\ \footnotesize [0.57, 0.63]} & \shortstack[l]{$-0.01$\textsuperscript{***} \\ \footnotesize [-0.02, -0.01]} & \shortstack[l]{$0.42$\textsuperscript{***} \\ \footnotesize [0.41, 0.43]} \\
Ivo Karlovic & \shortstack[l]{$1.23$\textsuperscript{***} \\ \footnotesize [0.66, 2.44]} & \shortstack[l]{$4.54$\textsuperscript{***} \\ \footnotesize [3.79, 5.49]} & \shortstack[l]{$-0.01$\textsuperscript{***} \\ \footnotesize [-0.02, -0.00]} & \shortstack[l]{$0.73$\textsuperscript{***} \\ \footnotesize [0.69, 0.78]} & \shortstack[l]{$-0.00$\textsuperscript{***} \\ \footnotesize [-0.00, -0.00]} & \shortstack[l]{$0.27$\textsuperscript{***} \\ \footnotesize [0.25, 0.29]} \\
John Isner & \shortstack[l]{$0.82$\textsuperscript{***} \\ \footnotesize [0.54, 1.20]} & \shortstack[l]{$5.50$\textsuperscript{***} \\ \footnotesize [4.98, 6.10]} & \shortstack[l]{$-0.00$\textsuperscript{***} \\ \footnotesize [-0.00, -0.00]} & \shortstack[l]{$0.66$\textsuperscript{***} \\ \footnotesize [0.63, 0.68]} & \shortstack[l]{$-0.00$\textsuperscript{***} \\ \footnotesize [-0.00, -0.00]} & \shortstack[l]{$0.27$\textsuperscript{***} \\ \footnotesize [0.26, 0.28]} \\
Reilly Opelka & \shortstack[l]{$0.39$\textsuperscript{***} \\ \footnotesize [0.17, 0.68]} & \shortstack[l]{$3.80$\textsuperscript{***} \\ \footnotesize [3.28, 4.35]} & \shortstack[l]{$-0.01$\textsuperscript{***} \\ \footnotesize [-0.02, -0.01]} & \shortstack[l]{$0.72$\textsuperscript{***} \\ \footnotesize [0.68, 0.76]} & \shortstack[l]{$-0.01$\textsuperscript{***} \\ \footnotesize [-0.01, -0.00]} & \shortstack[l]{$0.29$\textsuperscript{***} \\ \footnotesize [0.27, 0.31]} \\
David Ferrer & \shortstack[l]{$0.04$\textsuperscript{} \\ \footnotesize [-0.29, 0.74]} & \shortstack[l]{$2.61$\textsuperscript{***} \\ \footnotesize [2.24, 3.02]} & \shortstack[l]{$-0.02$\textsuperscript{***} \\ \footnotesize [-0.03, -0.01]} & \shortstack[l]{$0.33$\textsuperscript{***} \\ \footnotesize [0.30, 0.37]} & \shortstack[l]{$0.00$\textsuperscript{} \\ \footnotesize [-0.00, 0.01]} & \shortstack[l]{$0.39$\textsuperscript{***} \\ \footnotesize [0.37, 0.41]} \\
Diego Schwartzman & \shortstack[l]{$-0.35$\textsuperscript{} \\ \footnotesize [-0.73, 0.38]} & \shortstack[l]{$3.46$\textsuperscript{***} \\ \footnotesize [2.93, 4.09]} & \shortstack[l]{$-0.01$\textsuperscript{***} \\ \footnotesize [-0.01, -0.00]} & \shortstack[l]{$0.27$\textsuperscript{***} \\ \footnotesize [0.24, 0.30]} & \shortstack[l]{$0.00$\textsuperscript{} \\ \footnotesize [-0.00, 0.00]} & \shortstack[l]{$0.41$\textsuperscript{***} \\ \footnotesize [0.39, 0.43]} \\
\hline
Mean & 0.47 & 2.13 & -0.08 & 0.58 & -0.01 & 0.38 \\
Std & 0.78 & 3.15 & 0.15 & 0.23 & 0.03 & 0.05 \\
Min & -3.65 & -14.92 & -1.24 & 0.22 & -0.25 & 0.16 \\
Median & 0.41 & 2.56 & -0.03 & 0.55 & -0.00 & 0.38 \\
Max & 3.09 & 6.66 & -0.00 & 1.98 & 0.17 & 0.64 \\
\bottomrule
\end{tabular}}

\vspace{0.5em}
{\footnotesize Notes: 95\% confidence intervals obtained by bootstrap with 2000 replications.$^{***}$ indicates significance at the 5\% level: for $\lambda$, CI excludes 1; for other parameters, CI excludes 0.}
\end{table}

\subsection{Model with different curvature for $f\left( x\right) $ and $%
k\left( x\right) $\label{Section:DiffCurv}}

How stable are the estimates of the salience weight to departure from the
assumption that the curvatures of $f$ and $k$ are the same? This question
relates to the second aspect of Condition (\ref{Assumption:Elasticity})
which forces the elasticity of the marginal probabilities of winning
one-shot and multi-shot points to be the same. We propose to do so by
replacing Condition (\ref{Assumption:Elasticity}) by:

\begin{condition}
\label{Assumption:Elasticity3}$x\frac{f^{\prime \prime }\left( x\right) }{%
f^{\prime }\left( x\right) }=\lambda - 1$, $x\frac{k^{\prime \prime }\left(
x\right) }{k^{\prime }\left( x\right) }=t\lambda -1$, where $\lambda >0$ and 
$t>0$.
\end{condition}

Note that we cannot estimate the parameter $t$ together with the other
parameters since we have already exhausted all the degrees of freedom
offered by the data. However, one can still check the robustness of our
estimates of the salience weight, by estimating them \emph{given} $t$ and
repeating for various values of $t$. We follow this approach and given
Condition (\ref{Assumption:Elasticity3}) adopt the following parametrization:%
\begin{equation*}
f\left( x\right) =\frac{a_{f}-x_{1}^{\lambda }}{\tau _{f}},\text{ }k\left(
x\right) =\beta \frac{a_{k}-x_{1}^{t\lambda }}{\tau _{k}},
\end{equation*}%
where the parameter $t>0$ allows us to introduce different curvature for the
two functions.

For $t=1$, we recover the main specification from Condition (\ref%
{Assumption:Elasticity}). For $t<1$, there is less curvature in $k\left(
x\right) $ and for $t>1$ there is more curvature in $k\left( x\right) $.

Before we can proceed to the estimation of the salience weights for
different values of $t$, we need to recover the formulas necessary to search
for the value of $\lambda $ given the data and the value of $t$.

This is done by reworking the FOCs under this specification. Recall first
that using the FOC for the second serve optimal strategy one recovers the
value of the salience weight condition on $\lambda $ (and now also
\thinspace $t$). This yields

\begin{eqnarray*}
\left[ f\left( x_{2}\right) +\beta k\left( x_{2}\right) \right] +x_{2}\left[ 
\frac{\partial f\left( x_{2}\right) }{\partial x}+\beta \frac{\partial
k\left( x_{2}\right) }{\partial x}\right]  &=&0 \\
&\Leftrightarrow & \\
\left[ \frac{a_{f}-x_{2}^{\lambda }}{\tau _{f}}+\beta \frac{%
a_{k}-x_{2}^{t\lambda }}{\tau _{k}}\right]  &=&\lambda \left[ \frac{%
x_{2}^{\lambda }}{\tau _{f}}+\beta t\frac{x_{2}^{t\lambda }}{\tau _{k}}%
\right]  \\
&\Leftrightarrow & \\
\frac{a_{f}}{\tau _{f}}+\beta \frac{a_{k}}{\tau _{k}}-\frac{x_{2}^{\lambda }%
}{\tau _{f}}-\beta \frac{x_{2}^{t\lambda }}{\tau _{k}} &=&\lambda \frac{%
x_{2}^{\lambda }}{\tau _{f}}+\lambda t\beta \frac{x_{2}^{t\lambda }}{\tau
_{k}} \\
&\Leftrightarrow & \\
a_{f}-x_{2}^{\lambda }-\lambda x_{2}^{\lambda } &=&\beta \left( t\lambda 
\frac{\tau _{f}}{\tau _{k}}x_{2}^{t\lambda }-\frac{a_{k}\tau _{f}}{\tau _{k}}%
+\frac{\tau _{f}}{\tau _{k}}x_{2}^{t\lambda }\right) 
\end{eqnarray*}%
and hence%
\begin{equation*}
\beta =-\frac{a_{f}\tau _{k}-\tau _{k}\left( 1+\lambda \right) z_{2}}{%
a_{k}\tau _{f}-\tau _{f}\left( 1+t\lambda \right) w_{2}},
\end{equation*}%
where $z_{j}=x_{j}^{\lambda }$ and $w_{j}=x_{j}^{t\lambda }$ for $j=1,2$.

Doing a similar exercise using the FOC for the first serve optimal strategy
to recover the updating equation for $\lambda $ conditional on $t$ yields

\begin{eqnarray*}
\left[ f\left( x_{1}\right) +\beta k\left( x_{1}\right) \right] +x_{1}\left[ 
\frac{\partial f\left( x_{1}\right) }{\partial x}+\beta \frac{\partial
k\left( x_{1}\right) }{\partial x}\right] &=&x_{2}\tilde{y}\left(
x_{2}\right) \\
&\Leftrightarrow & \\
\left[ f\left( x_{1}\right) +\beta k\left( x_{1}\right) \right] +x_{1}\left[ 
\frac{\partial f\left( x_{1}\right) }{\partial x}+\beta \frac{\partial
k\left( x_{1}\right) }{\partial x}\right] &=&x_{2}\left( f\left(
x_{2}\right) +\beta k\left( x_{2}\right) \right) \\
&\Leftrightarrow & \\
\left[ \frac{a_{f}-x_{1}^{\lambda }}{\tau _{f}}+\beta \frac{%
a_{k}-x_{1}^{t\lambda }}{\tau _{k}}\right] -\lambda \left[ \frac{%
x_{1}^{\lambda }}{\tau _{f}}+\beta t\frac{x_{1}^{t\lambda }}{\tau _{k}}%
\right] &=&x_{2}\left( \frac{a_{f}-x_{2}^{\lambda }}{\tau _{f}}+\beta \frac{%
a_{k}-x_{2}^{t\lambda }}{\tau _{k}}\right) \\
&\Leftrightarrow & \\
\left[ \frac{a_{f}}{\tau _{f}}-\frac{x_{1}^{\lambda }}{\tau _{f}}+\beta 
\frac{a_{k}}{\tau _{k}}-\beta \frac{x_{1}^{t\lambda }}{\tau _{k}}\right] -%
\left[ \lambda \frac{x_{1}^{\lambda }}{\tau _{f}}+t\lambda \beta \frac{%
x_{1}^{t\lambda }}{\tau _{k}}\right] &=&x_{2}\left( \frac{%
a_{f}-x_{2}^{\lambda }}{\tau _{f}}+\beta \frac{a_{k}-x_{2}^{t\lambda }}{\tau
_{k}}\right)
\end{eqnarray*}

We can now proceed and isolate $\lambda $,

\begin{eqnarray*}
\left[ \frac{a_{f}}{\tau _{f}}-\frac{x_{1}^{\lambda }}{\tau _{f}}+\beta 
\frac{a_{k}}{\tau _{k}}-\beta \frac{x_{1}^{t\lambda }}{\tau _{k}}\right] -%
\left[ \lambda \frac{x_{1}^{\lambda }}{\tau _{f}}+t\lambda \beta \frac{%
x_{1}^{t\lambda }}{\tau _{k}}\right]  &=&x_{2}\left( \frac{%
a_{f}-x_{2}^{\lambda }}{\tau _{f}}+\beta \frac{a_{k}-x_{2}^{t\lambda }}{\tau
_{k}}\right)  \\
&\Leftrightarrow & \\
\left[ a_{f}-z_{1}+\beta \frac{\tau _{f}a_{k}}{\tau _{k}}-\beta \frac{\tau
_{f}w_{1}}{\tau _{k}}\right] -\left[ \lambda z_{1}+t\lambda \beta \frac{\tau
_{f}w_{1}}{\tau _{k}}\right]  &=&x_{2}\left( a_{f}-z_{2}+\beta \frac{\tau
_{f}\left( a_{k}-w_{2}\right) }{\tau _{k}}\right)  \\
&\Leftrightarrow & \\
\left[ a_{f}-z_{1}+R\left( a_{k}-w_{1}\right) \right] -\left[ \lambda
z_{1}+t\lambda Rw_{1}\right]  &=&x_{2}\left( a_{f}-z_{2}+R\left(
a_{k}-w_{2}\right) \right) ,
\end{eqnarray*}%
where $R=\beta \frac{\tau _{f}}{\tau _{k}}$.

A simple rearrangement yields%
\begin{equation*}
\lambda \left[ z_{1}+tw_{1}R\right] =\left[ a_{f}-z_{1}+R\left(
a_{k}-w_{1}\right) \right] -x_{2}\left( a_{f}-z_{2}+R\left(
a_{k}-w_{2}\right) \right) ,
\end{equation*}%
and it follows that the updating equation given $t$ reads as%
\begin{equation*}
\lambda =\frac{a_{f}-z_{1}+R\left( a_{k}-w_{1}\right) -x_{2}\left(
a_{f}-z_{2}+R\left( a_{k}-w_{2}\right) \right) }{z_{1}+tw_{1}R}.
\end{equation*}

We note that when $t=1$, this equation reads as the one obtained under
Condition (\ref{Assumption:Elasticity}), for then $z_{j}=w_{j}$ for $j=1,2$.

The algorithm presented in the paper can be amended to account for this more
general structure. In particular, it takes as input the data $x_{1}$, $x_{2}$%
, $f_{1}$, $f_{2}$, $k_{1}$ and $k_{2}$ together with a value for $t$. The
output is still the curvature parameter $\lambda $. However, the slopes and
constants of $f(x)$ and $k(x)$ now obtain as follows. Let $%
z_{j}=x_{j}^{\lambda }$ and $w_{j}=x_{j}^{t\lambda }$ for $j=1,2$. Note that%
\begin{equation*}
\tau _{f}=-\frac{z_{2}}{f_{2}}\frac{\frac{z_{1}-z_{2}}{z_{2}}}{\frac{%
f_{1}-f_{2}}{f_{2}}},\text{ }\tau _{k}=-\frac{w_{2}}{k_{2}}\frac{\frac{%
w_{1}-w_{2}}{w_{2}}}{\frac{k_{1}-k_{2}}{k_{2}}},
\end{equation*}%
and%
\begin{equation*}
a_{f}=\tau _{f}f_{1}+z_{1},\text{ }a_{k}=\tau _{k}k_{1}+w_{1}.
\end{equation*}%
The preference parameter is obtained as%
\begin{equation*}
\beta =-\frac{a_{f}\tau _{k}-\tau _{k}\left( 1+\lambda \right) z_{2}}{%
a_{k}\tau _{f}-\tau _{f}\left( 1+t\lambda \right) w_{2}},
\end{equation*}%
and, writing $R=\beta \frac{\tau _{f}}{\tau _{k}}$, the updating equation
now becomes%
\begin{equation*}
\lambda =\frac{a_{f}-z_{1}+R\left( a_{k}-w_{1}\right) -x_{2}\left(
a_{f}-z_{2}+R\left( a_{k}-w_{2}\right) \right) }{z_{1}+tw_{1}R}.
\end{equation*}

Table (\ref{SmallEsti_t}) shows the estimates of the salience weight and
that Table (\ref{SmallEsti_lambda_t}) of the curvature of the model with
this parametric specification. Note that for a few players, e.g., Pete
Sampras and Reilly Opelka, for instance, there is no solution for the
curvature parameter $\lambda$ at $t=0.5$.

\begin{table}[H]
\centering
\caption{Stability of salience weight estimates to differences in relative curvature between $f(x)$ and $k(x)$, for selected players.\label{SmallEsti_t}}
\resizebox{\textwidth}{!}{%
\begin{tabular}{l|ccccccc}
\toprule
Player & \multicolumn{7}{c}{$t$} \\
\cmidrule(lr){2-8}
  & $0.50$ & $0.75$ & $1.00$ & $1.25$ & $1.50$ & $1.75$ & $2.00$ \\
\midrule
Novak Djokovic & $0.26$ & $0.27$ & $0.27$ & $0.27$ & $0.25$ & $0.23$ & $0.19$ \\
Rafael Nadal & $0.21$ & $0.21$ & $0.21$ & $0.21$ & $0.21$ & $0.20$ & $0.19$ \\
Roger Federer & $0.31$ & $0.32$ & $0.32$ & $0.32$ & $0.30$ & $0.27$ & $0.24$ \\
Pete Sampras & n.a.  & $0.42$ & $0.42$ & $0.42$ & $0.39$ & $0.35$ & $0.29$ \\
Boris Becker & $0.56$ & $0.56$ & $0.56$ & $0.56$ & $0.55$ & $0.54$ & $0.51$ \\
Carlos Alcaraz & $0.05$ & $0.05$ & $0.05$ & $0.05$ & $0.04$ & $0.03$ & $0.01$ \\
Jannik Sinner & $0.05$ & $0.06$ & $0.06$ & $0.06$ & $0.04$ & $0.02$ & $-0.01$ \\
Ivo Karlovic & $1.15$ & $1.18$ & $1.19$ & $1.17$ & $1.08$ & $0.95$ & $0.81$ \\
John Isner & $0.73$ & $0.77$ & $0.79$ & $0.73$ & $0.59$ & $0.43$ & $0.29$ \\
Reilly Opelka & n.a.  & $0.35$ & $0.36$ & $0.33$ & $0.23$ & $0.12$ & $0.01$ \\
David Ferrer & $-0.01$ & $-0.01$ & $-0.01$ & $-0.01$ & $-0.01$ & $-0.00$ & $0.00$ \\
Diego Schwartzman & $-0.38$ & $-0.38$ & $-0.38$ & $-0.38$ & $-0.38$ & $-0.38$ & $-0.38$ \\
\hline
Mean & $0.31$ & $0.33$ & $0.34$ & $0.38$ & $0.35$ & $0.30$ & $0.27$ \\
Std & $0.91$ & $0.72$ & $0.70$ & $0.55$ & $0.58$ & $0.61$ & $0.59$ \\
Min & $-5.25$ & $-4.04$ & $-3.86$ & $-2.54$ & $-1.85$ & $-2.62$ & $-2.48$ \\
Median & $0.32$ & $0.33$ & $0.33$ & $0.33$ & $0.31$ & $0.28$ & $0.25$ \\
Max & $2.92$ & $2.81$ & $2.79$ & $2.80$ & $2.84$ & $2.90$ & $2.96$ \\
\bottomrule
\end{tabular}}
\vspace{0.5em}
{\footnotesize Notes: Each column reports the estimate of the salience weight $\delta$ for a given value of the relative curvature $t$ of multi-shot probability of winning a point, relative to one-shot.}
{\footnotesize n.a. means non available because there is no solution for $\lambda$.}
\end{table}

\begin{table}[H]
\centering
\caption{Stability of curvature estimates to differences in relative curvature between $f(x)$ and $k(x)$, for selected players.\label{SmallEsti_lambda_t}}
\resizebox{\textwidth}{!}{%
\begin{tabular}{l|ccccccc}
\toprule
Player & \multicolumn{7}{c}{$t$} \\
\cmidrule(lr){2-8}
  & $0.50$ & $0.75$ & $1.00$ & $1.25$ & $1.50$ & $1.75$ & $2.00$ \\
\midrule
Novak Djokovic & $3.22$ & $3.44$ & $3.67$ & $3.90$ & $4.10$ & $4.27$ & $4.41$ \\
Rafael Nadal & $4.74$ & $4.82$ & $4.89$ & $4.96$ & $5.01$ & $5.06$ & $5.10$ \\
Roger Federer & $2.44$ & $2.62$ & $2.81$ & $3.00$ & $3.17$ & $3.31$ & $3.42$ \\
Pete Sampras & n.a.  & $1.71$ & $1.93$ & $2.18$ & $2.44$ & $2.67$ & $2.84$ \\
Boris Becker & $1.54$ & $1.63$ & $1.73$ & $1.83$ & $1.93$ & $2.03$ & $2.12$ \\
Carlos Alcaraz & $3.34$ & $3.50$ & $3.66$ & $3.81$ & $3.95$ & $4.07$ & $4.17$ \\
Jannik Sinner & $2.15$ & $2.33$ & $2.52$ & $2.72$ & $2.90$ & $3.06$ & $3.18$ \\
Ivo Karlovic & $3.14$ & $3.67$ & $4.37$ & $5.21$ & $5.99$ & $6.55$ & $6.88$ \\
John Isner & $3.11$ & $3.95$ & $5.33$ & $7.07$ & $8.14$ & $8.54$ & $8.66$ \\
Reilly Opelka & n.a.  & $2.97$ & $3.89$ & $5.07$ & $5.89$ & $6.25$ & $6.36$ \\
David Ferrer & $3.08$ & $2.99$ & $2.90$ & $2.83$ & $2.77$ & $2.71$ & $2.66$ \\
Diego Schwartzman & $3.69$ & $3.62$ & $3.57$ & $3.52$ & $3.47$ & $3.43$ & $3.40$ \\
\hline
Mean & $2.94$ & $2.86$ & $2.94$ & $3.08$ & $3.28$ & $3.48$ & $3.55$ \\
Std & $1.30$ & $1.10$ & $0.98$ & $1.09$ & $1.38$ & $1.55$ & $1.55$ \\
Min & $1.17$ & $1.19$ & $1.28$ & $1.28$ & $0.59$ & $1.33$ & $1.30$ \\
Median & $2.73$ & $2.74$ & $2.86$ & $3.00$ & $3.07$ & $3.17$ & $3.26$ \\
Max & $9.71$ & $9.15$ & $6.27$ & $7.32$ & $9.70$ & $9.69$ & $9.69$ \\
\bottomrule
\end{tabular}}
\vspace{0.5em}
{\footnotesize Notes: Each column reports the estimate of the curvature $\lambda$ for a given value of the relative curvature $t$ of multi-shot probability of winning a point, relative to one-shot.}
{\footnotesize n.a. means non available because there is no solution for $\lambda$.}
\end{table}

\subsection{Model of aversion for double faults\label{Section:AversionDF}}

Suppose a player has a disutility from making double faults, and let $\gamma
>0$ capture the magnitude of this disutility. In this case, the utility the
player would maximize is%
\begin{equation*}
x_{1}y\left( x_{1}\right) +\left( 1-x_{1}\right) \left[ x_{2}y\left(
x_{2}\right) -\gamma \left( 1-x_{2}\right) \right] .
\end{equation*}

It is easy to show that, after some rearranging, this utility function
rewrites as%
\begin{equation*}
x_{1}\check{y}\left( x_{1}\right) +\left( 1-x_{1}\right) x_{2}\check{y}%
\left( x_{2}\right) -\gamma
\end{equation*}%
where $\check{y}\left( x\right) =y\left( x\right) +\gamma $.

The FOCs of the maximization problem then read as%
\begin{eqnarray*}
x_{2}y^{\prime }\left( x_{2}\right) +y\left( x_{2}\right) +\gamma &=&0 \\
x_{1}y^{\prime }\left( x_{1}\right) +y\left( x_{1}\right) +\gamma \left(
1-x_{2}\right) &=&x_{2}y\left( x_{2}\right)
\end{eqnarray*}%
since $\check{y}^{\prime }\left( x\right) =y^{\prime }\left( x\right) $.

Assuming, as we did in the main model of the paper, that $y\left( x\right) $
is a power function, $y\left( x\right) =\frac{a}{\tau }-\frac{1}{\tau }%
x^{\lambda }$ where $\tau >0$, so that $y^{\prime }\left( x\right) =-\frac{1%
}{\tau }\lambda x^{\lambda -1}<0$ and $xy^{\prime }\left( x\right) =-\frac{1%
}{\tau }\lambda x^{\lambda }$, so that the FOCs read as

\begin{eqnarray*}
-\frac{1}{\tau }\lambda x_{2}^{\lambda }+\frac{a}{\tau }-\frac{1}{\tau }%
x_{2}^{\lambda }+\gamma &=&0, \\
-\frac{1}{\tau }\lambda x_{1}^{\lambda }+\frac{a}{\tau }-\frac{1}{\tau }%
x_{1}^{\lambda }+\gamma &=&x_{2}\left( \frac{a}{\tau }-\frac{1}{\tau }%
x_{2}^{\lambda }+\gamma \right) .
\end{eqnarray*}

We can then solve for the optimal serve strategies. Solving the first
equation for $x_{2}^{\lambda }$ yields $x_{2}^{\lambda }=\frac{a+\tau \gamma 
}{\lambda +1}$ while solving the second for $x_{1}^{\lambda }$ yields $%
x_{1}^{\lambda }=x_{2}^{\lambda }\left( 1-\frac{\lambda }{\lambda +1}%
x_{2}\right) $.

Comparing the optimal strategy from the aversion to the double fault model
to that of the process utility model, one first notes that the first serve
strategy has the same expression in both models. Since this equation is used
to estimate the curvature parameter, when applied on the same data, i.e.,
given $x_{1}$ and $x_{2}$, both models deliver the same estimate of $\lambda 
$. It follows that the slope and constant of $y\left( x\right) $ that would
obtain from data $\left( x_{1},x_{2},y_{1},y_{2}\right) $ are also the same
in the two models and obtain as: $\tau =\frac{x_{2}^{\lambda
}-x_{1}^{\lambda }}{y_{1}-y_{2}}$ and $a=\frac{x_{2}^{\lambda
}-x_{1}^{\lambda }}{y_{1}-y_{2}}y_{2}+x_{2}^{\lambda }$. The differentce
between the two models is apparent in the expression of the second serve
strategies which read respectively as

\begin{eqnarray}
x_{2}^{\lambda } &=&\frac{a+\tau \gamma }{\lambda +1},  \label{eqOpti2A} \\
x_{2}^{\lambda } &=&\frac{a_{f}+a_{k}\beta \frac{\tau _{f}}{\tau _{k}}}{%
\left( 1+\lambda \right) \left( 1+\beta \frac{\tau _{f}}{\tau _{k}}\right) }.
\label{eqOpti2P}
\end{eqnarray}

Picking the right values of the parameters $\gamma $ for the aversion to
double faults model and ($\beta $, $a_{f}$ and $\tau _{f}$) for the process
utility model, we conclude that the two models fit equally well the data $%
\left( x_{1},x_{2},y_{1},y_{2}\right) $ and share in common the curvature
parameter $\lambda $ and the slope and constant of $y\left( x\right) $,
i.e., $\tau $ and $a$. However, note that from equation (\ref{eqOpti2A}),
the estimate of $\gamma $ is invariant to data $\left( f_{1},f_{2}\right) $
conditional on $\left( x_{1},x_{2},y_{1},y_{2}\right) $. Stated otherwise,
the optimal decision of the player only reflects his general ability to win
points, i.e., parameters $a$ and $\tau $, and his aversion to double faults $%
\gamma $, regardless of his ability to win one-shot points, i.e., $a_{f}$
and $\tau _{f}$. If the true data-generating model were that where players
had an aversion to making double faults, the parameter $\gamma $ should be
unrelated to the conditional probabilities of winning one-shot rallies on
either the first or second serve, since all that matters for such a player's
strategy is $y\left( x\right) $ ($y_{1}$ and $y_{2}$) and not their
constituents ($f_{1}$, $f_{2}$, $k_{1}$ and $k_{2}$). However, if the true
model is one where players value process utility, then the observed optimal
strategies $x_{1}$ and $x_{2}$ would depend on the constituents of $y\left(
x\right) $, i.e., $f\left( x\right) $ and $k\left( x\right) $. It follows
that the parameters obtained from estimating the double-fault aversion model
and, in particular, the computed value of $\gamma $, would be related to $%
f\left( x\right) $, since they are is obtained using the optimal strategies $%
x_{1}$ and $x_{2}$ that depend on $a_{f}$ and $\tau _{f}$.

In particular, comparing players with similar serve strategies ($x_{1},x_{2}$%
) and conditional probabilities of winning points ($y_{1}$ and $y_{2}$), the
conditional mixed moments $\mathbb{E}\left[ \hat{\gamma}\times f\left(
x_{j}\right) |x_{1},x_{2},y_{1},y_{2}\right] $ for $j=1,2$ should be zero.
To test these hypotheses, we first estimate, for each player, the aversion
for double faults parameter $\hat{\gamma}$ using the procedure outlined in
the paper, but applied to the model with aversion to double faults. Then, we
estimate the conditional expectations $\mathbb{E}\left[ \hat{\gamma}\times
f_{1}|x_{1},x_{2},y_{1},y_{2}\right] $ and $\mathbb{E}\left[ \hat{\gamma}%
\times f_{2}|x_{1},x_{2},y_{1},y_{2}\right] $ nonparametrically using
locally weighted scatterplot smoothing (LOWESS) and apply a bootstrap
procedure (300 replications) to construct $95\%$ confidence intervals. The
in-sample predictions of the conditional expectations $\mathbb{E}\left[ \hat{%
\gamma}_{i}\times f_{ji}|x_{1i},x_{2i},y_{1i},y_{2i}\right] $ for $j=1,2$
and $i=1,...N$, and their confidence interval can then be used to check if
the associated values are significantly different from 0. We find that for
99 out of the 151 players, at least one of the conditional expectations is
statistically significantly larger than 0, rejecting the null hypothesis
that the parameter $\hat{\gamma}$ is orthogonal to the way rallies are won.
The estimates of the aversion to the double fault parameter $\gamma $ seem
to be systematically dependent on the conditional probability of winning
one-shot rallies on first and second serves, conditional on the serve
strategy and probability of winning points. This suggests that there is
information in the conditional probability of winning one-shot and
multi-shot rallies that a model with aversion to double fault does not
capture, and that is being forced into the parameter $\gamma $.

Moreover, the fact that these expectations are positive tells a
counterintuitive story, where, holding the serve strategy and probabilities
of winning a point constant, those that rely more on winning points through
one-shot rallies, i.e., the better serves, are the most averse to making
double faults according to the model.

\end{document}